\documentclass{article}
\usepackage{template}
\definecolor{orange}{rgb}{.5941,.1255,.1137}
\everymath{\color{black}}
\usepackage{float}
\usetikzlibrary{shapes}
\usepackage{booktabs}
\usepackage{threeparttable}
\usepackage{blindtext}
\usepackage{caption}
\usepackage{ragged2e}
\usepackage[toc,page]{appendix}

\title{Risk in Network Economies\footnote{I am particularly thankful to Allan Timmermann for useful comments and discussions throughout the course of this work.  I am also thankful for helpful comments from Alexis Toda, Johannes Wieland, Valerie Ramey, Ross Valkanov, Joey Engelberg, Munseob Lee, Fabian Trottner, Tjeerd de Vries, Edoardo Briganti, and Carlos Goes.}}
\author{ \color{black} Victor Sellemi\footnote{PhD Student UC San Diego, email: \href{vsellemi@ucsd.edu}{vsellemi@ucsd.edu}.}}

\date{This Version: July 11, 2022}

\begin{document}

\maketitle
\begin{abstract}
   Economic models with input-output networks assume that firm or sector (unit) growth is driven by a weighted sum of trade partners' growth and an independently-drawn idiosyncratic shock. I show that the idiosyncratic risk assumption in a broad class of network models implicitly generates restrictions on the network weights which are unrealistic. When allowing for correlated shocks, units are exposed to an additional risk term which captures the ability to substitute away from supply and demand shocks propagating through the network. I provide empirical evidence that changes in substitutability between trade partners are inversely related to changes in the panel of realized industry variance. Moreover, I find that supply-side (demand-side) substitutability is closely related to technological (product) dispersion of a unit's suppliers (customers). To synthesize these results, I propose a production-based asset pricing model in which supply chain substitutability is a function of dispersion in product/technology space and correlation in supply and demand shocks is driven by shared customers and suppliers between firms. The model predicts that assets which are positively exposed to average propagation of upstream and downstream shocks are useful hedges and thus earn lower average risk premia. Consistently, I find that estimated upstream (downstream) propagation factors earn return spreads of -11.4\% (-4.2\%) and are negatively associated with aggregate consumption, output, and dividend growth.    \\
    
    \noindent \textbf{Keywords}: Production Networks, Volatility, Systematic Risk, Asset Pricing

\end{abstract}
\newpage
\onehalfspacing
\newpage

\section{Introduction}
\label{sec:introduction}

The final goods share of consumption in the United States stands at a little over one-third. Most of what remains is flows of intermediate inputs through production networks. Recent research finds that production networks play an important role in shock propagation, business cycles, and systematic risk in asset markets. However, the relationship between network linkages and comovement in economic risk is not yet entirely clear, especially at a granular level. Features of the input-output network are crucial to understanding the relationship between firm or industry-level risk and economy-wide aggregate risk. 

The benchmark network model assumes that idiosyncratic shocks are drawn independently across units (firms or industries) in the network, and then propagate to connected units as a function of network weights. Network weights capture the relative importance of each connection (edge) between units. As a result, shocks to any individual unit can have systematic effects. \citet{Acemoglu2012} argues that fat tails in the distribution of network weights can inhibit diversification and amplify the systematic effects of idiosyncratic shocks. Similarly, \citet{Gabaixorigins} argues that skewness in the firm size also inhibits diversification away from shocks.

In this work, I argue that the assumption that shocks are idiosyncratic is not consistent with realistic input-output network models of the economy. More specifically, I consider a general reduced-form equation for static propagation of shocks, which links each unit's growth to the growth of its network connections plus a unit-specific shock. This reduced-form is consistent with a broad class of structural economic models that assume Cobb-Douglas aggregation of intermediate inputs in a production function. In this equation, imposing a diagonal structure on the variance-covariance matrix of shocks generates implicit restrictions on the network weights permitted in the model.\footnote{I provide both mathematical and numerical evidence in support of this.} More specifically, the set of permissible input-output networks has weights which are overly sparse and/or economically uninteresting. For example, the idiosyncratic risk assumption might constrain researchers to studying economies in which pairs of sectors or firms that use each other's inputs can only use each other's inputs. 

As a result, I argue that researchers should account for correlation in shocks when making use of such a network model.  Practically speaking, there are several reasons why shocks to units in the input-output network might be correlated, especially as the unit definitions become more granular. For instance, two firms which produce the same goods should experience correlation in demand shocks at the product level. If the two firms also produce using the same inputs and technologies, then supply-side shocks associated with that technology are likely to be correlated as well. \citet{Hoberg2016textindustries} show from text data that firms that produce similar products often belong to different industries, which suggests that industries should also experience some degree of comovement in demand shocks. Along these lines, \citet{Hottman2016} show using scanner data that 69\% of firms, which account for 99\% of their industries' output, supply multiple and intersecting product varieties. 

Like product similarity, technological and geographic proximity might also generate comovement in supply and demand shocks. \citet{bloomtechspillovers2013ecma} shows that regional shocks to research and development (R\&D) incentives have correlated effects on the growth of firms who operate in closely related technology spaces. Similarly, firms operating in nearby locations are likely to be exposed to the same underlying geographic shocks. For instance, \citet{autordornhanson2013} and \citet{Mian2014} provide evidence that local employment shocks have correlated effects within a region, and \citet{tuzel2017local} studies correlated exposure to regional risk associated with changes in local prices for factors of production. Even local climate risk could expose multiple firms to the same regional risks (see e.g., \citet{BarrotSauvagnat2016} and \citet{Kruttli2019pricingposeidon}). 

When shocks are idiosyncratic, each unit's growth rate variance is the sum of unit-specific shock variance and a network-weighted sum of shock variances of its trade partners. Of course, the former term is unrelated to the presence of network connections. In the homoskedastic case, the second term simplifies to a constant times a concentration measure across each unit's trade partners. \citet{Acemoglu2012} show that aggregate volatility shocks to this component decay at a rate slower than $\sqrt{n}$ when weights follow a power law distribution.  \citet{herskovic_firm_2020} focus on customer concentration, and argue that increases in concentration are driven by increases in size dispersion. To my knowledge, this is the first work to investigate the relationship between exposure to correlated shocks in the production network and realized variance. 

In particular, when allowing for non-negligible correlation in shocks, the expression for growth rate variance gains an additional covariance component, denoted concentration ``between" trade partners. This new term captures the ability of each unit in the network to substitute away from correlated shocks to its trade partners. In particular, units are more substitutable (less concentrated) when they diversify trade between partners that are exposed to negatively correlated shocks. When units trade with partners that experience correlated shocks, they are more concentrated when the relative importance of those trade partners is similar. 

Building on this intuition, I estimate concentration between trade partners using panel data at the disaggregated industry level. Consistent with theory, I find that this new component explains a significant amount of variation in the panel of realized industry variance. More concentrated (less substitutable) industries are more volatile both in terms of market returns and output growth. This relationship is robust and holds even when controlling for relevant industry characteristics such as size, centrality, concentration across trade partners, vertical position in the supply chain, and durability of output. 

This finding alone does not provide any insight on the underlying source of correlated risks between trade partners. Diving deeper, I consider the results of \citet{AcemogluAkjitKerr2016}, who argue that total factor productivity (TFP) shocks primarily propagate downstream while government spending shocks primarily propagate upstream from customers to suppliers. Consistent with this finding, I show that the elasticity of realized variance to concentration between trade partners is more precisely estimated on the supply-side when constructed using pairwise industry correlations in TFP growth. On the other hand, the elasticity of realized variance to  concentration between customers is more precisely estimated when using correlations in federal procurement demand shocks. 

Additionally, I suppose that correlation between upstream and downstream propagating shocks is driven by proximity of industries on a latent surface. For tractability, I assume that correlation between demand shocks is a function of product similarity, while correlation between supply shocks is a function of technological similarity. I proxy product similarity using the text-based scores from \citet{Hoberg2016textindustries} and technological proximity following \citet{bloomtechspillovers2013ecma}. Similarly, I find that the elasticity of realized variance to between concentration is more precisely estimated on the demand-side using product similarity and on the supply-side using technological proximity. 

These findings suggest a structural foundation for incorporating correlation in supply and demand shocks in network models of the economy. To fully investigate the risk implications of this correlation, I propose a production-based asset pricing model with input-output networks in which firm-level technology shocks propagate downstream from suppliers to customers and demand shocks propagate upstream from customers to suppliers. Firms are both customers and suppliers. Unlike most existing models, I account for both directions of propagation.\footnote{For example, \citet{johnshea2002} and \citet{kramarz_volatility_2020}, and\citet{herskovic_firm_2020} focus on upstream propagation of demand shocks, while \citet{Acemoglu2012} focus on downstream propagation.} Additionally, I introduce a novel mechanism for correlation in shocks in which the propensity of transmission is a function of customer and supplier substitutability at the firm level. 

In particular, I define substitutability as network weighted sum of latent distances between a firm's trade partners. Product distance characterizes customer substitutability, while technological distance characterizes supplier substitutability. Shared customers and suppliers between firms induces comovement in substitutability and thus correlation in propagated shocks. The proportion of firms that are affected by network propagated shocks is related to the changes in average propensity and average supply chain substitutability. Importantly, the model predicts that average propagation in the upstream and downstream directions represent distinct and negatively priced sources of systematic risk in the economy.

I test this prediction by calibrating the model and constructing empirical analogues of upstream and downstream network propagation risk factors. Consistent with theory, I confirm empirically that upstream and downstream propagation risk factors are negatively related to aggregate consumption growth, output growth, and dividend growth. Moreover, a trading strategy which buys the highest and sells the lowest quintile upstream (downstream) propagation beta-sorted portfolio generates excess returns of -11.42\% (-4.18\%). These factors survive the standard set of robustness checks.










\section{Idiosyncratic Risk in Input-Output Networks}
\label{sec:motivation}
In an economy where sectors or firms are connected through a network of input-output linkages, shocks to any individual unit might generate larger systematic effects. Intuitively, firms or industries with close trade relationships should also experience some degree of comovement in risk. Recent research proposes several approaches for modeling the spread of small shocks from firms or disaggregated sectors.\footnote{See e.g., \citet{Gabaixorigins}, \citet{Acemoglu2012}, \citet{taschereau-dumouchel_cascades_nodate}, and \citet{baqaeeFarhiMacroimpact} for discussion on how microeconomic shocks can generate macroeconomic effects.} Since input-output networks are observed in the data, these approaches lend themselves to empirical studies on the importance of various channels of shock propagation.\footnote{Generally these studies focus on propagation at business-cycle frequencies. }

The benchmark model studied in much of the literature involves static propagation of shocks through a deterministic network.\footnote{Some examples include \citet{Acemoglu2012}, \citet{AcemogluAkjitKerr2016}, \citet{OzdagliWeber2017}, \citet{herskovic_networks_2018}, \citet{herskovic_firm_2020}.} The main idea is that a sector or firm's growth rates depend on a network-weighted sum of growth rates of trade partners and an idiosyncratic shock which is drawn independently from the other units. Network weights capture the importance of direct trade relationships between sectors and are generally non-negative. Moreover, if the entries of the matrix are sales or purchase shares of inputs, these weights are also bounded above by 1 and in most cases assumed to sum to 1 or less than 1 for every unit. Typically no additional restrictions are imposed on the input-output network structure (e.g., symmetry or sparsity). 

However, in this section I show that in this benchmark model, the assumption of idiosyncratic shocks across units is not consistent with such a general class of networks. In particular, stronger restrictions on the input-output network weights are required when the variance-covariance matrix of idiosyncratic shocks is diagonal. These additional restrictions are inconsistent with almost all empirically observed input-output networks, and cannot be relaxed by adding omitted macroeconomic factors or by accounting for multiple networks. Additionally, even if these restrictions are satisfied, there is no definitive empirical evidence that supply and demand shocks have zero pairwise correlation across all pairs of units.

After formally establishing this result, I explore the implications of accounting for correlated shocks in this static framework. In this modified setting, sectors and firms are still exposed to risk from direct and indirect trade partners, but now can also substitute away from risk by having trade partners that are differentially exposed to supply and demand shocks. More specifically, the variance of a sector's growth rate inherits the standard network component which is related to the concentration of risk across trade partners, but also two additional components which capture a trade-off of concentration and substitutability \textit{between} trade partners. Intuitively, higher concentration implies less diversification in supply-chains and should imply more volatility. On the other hand, high substitutability implies that units can better average away the effects of shocks across customers or suppliers. In the following sections, I provide both theoretical and empirical motivation the researchers should account for correlated shocks when studying risk in network economies.

\subsection{Networks and Risk Comovement}
In this section, I argue that realistic input-output models of the economy should account for correlation in supply and demand shocks across units. In the benchmark static model of sectoral shock propagation, I find that the set of stable input-output networks that are consistent with the idiosyncratic risk assumption is unrealistic. Mathematically, in this broad class of reduced-form linear models, additional restrictions are required on the input-output network weights to be consistent with  an arbitrary covariance matrix of sector or firm-level growth rates and an arbitrary diagonal covariance matrix of shocks. Although this result is not immediately intuitive, the assumption of idiosyncratic shocks implicitly generates a strict relationship on the interaction between network weights and elements of the variance-covariance matrix of growth rates.

To illustrate this point, I start from the general reduced-form model of shock propagation in which a firm's output growth is driven by a network component and firm-specific shocks.\footnote{Similar models are used in \citet{Acemoglu2012}, \citet{AcemogluAkjitKerr2016}, \citet{Herskovic2018JoF} and \citet{herskovic_firm_2020}.} In Appendix \ref{asec:General Equilibrium Model of Input-Output Linkages}, I show that this model is consistent with the equilibrium outcome of a constant returns to scale economy in which Cobb-Douglas producers experience productivity shocks that propagate downstream from suppliers to customers and demand shocks that propagate upstream from customers to suppliers.  In particular, for an $n$-firm economy, consider the static relationship:
\begin{align}
\label{eq:reduced form static propagation}
    \mathbf{y} = \mathbf{W} \mathbf{y} + \mathbf{u},
\end{align}
where $\mathbf{y}$ is the $n \times 1$ vector of firm-level output growth, $\mathbf{W}$ is the $n\times n$ network matrix capturing interactions between industries, and $\mathbf{u}$ is the $n\times 1$ vector of firm-specific supply or demand shocks. This framework is compatible with either direction of propagation, upstream from customers to suppliers or downstream from suppliers to customers.  The following two assumptions require that the propagation matrix $\mathbf{W}$ implies is stable, and that firm-specific shocks are idiosyncratic, respectively.
\begin{assumption}[Stable Weighting Matrix] \label{ass:stable weighting matrix} The weighting matrix $\mathbf{W} \in M_{n}$ is non-negative, and has bounded spectral radius $\rho(\mathbf{W}) \leq 1$. 
\end{assumption}
\begin{assumption}[Idiosyncratic Shocks]
\label{ass:idiosyncratic static shocks}
Firm-specific shocks $u_i \sim \mathcal{P}_i(0,\sigma_{i}^2)$ are drawn independently across firms where $\mathcal{P}_i \in L^2$ has finite second moments. In other words, there exists a positive diagonal matrix $\mathbf{D}\in M_n$ such that $\mathbb{E}[\mathbf{uu}^\top] = \mathbf{D}$. 

\end{assumption}
\noindent In the following proposition, I characterize a set of additional necessary restrictions on the matrix $\mathbf{W} \in M_n$ to satisfy Assumption \ref{ass:idiosyncratic static shocks} and \eqref{eq:reduced form static propagation}. 
\begin{proposition}[Necessary Restrictions on $\mathbf{W}$] \label{prop:necessary restrictions on W} Any weighting matrix $\mathbf{W}$ which satisfies \ref{ass:stable weighting matrix} and \ref{ass:idiosyncratic static shocks} must have that for any pair of off-diagonal nodes $w_{ij}$ with $i\neq j$, either $w_{ij}w_{ji} \geq 1$ or $w_{ij} > 0$ and $w_{ji} = 0$. 
\end{proposition}
\begin{proof}
See Appendix \ref{asec:necessary restrictions on W}.
\end{proof}
This proposition highlights a key limitation of equation \eqref{eq:reduced form static propagation}. To apply network models in a way that is consistent with reality, researchers must either restrict their focus to a very particular set of networks or allow for correlation in sectoral or firm-level shocks. Under Assumptions \ref{ass:stable weighting matrix} and \ref{ass:idiosyncratic static shocks}, consistent networks have only one-way connections or entries which are overly sparse and/or economically uninteresting. In the structural model developed in Appendix \ref{asec:General Equilibrium Model of Input-Output Linkages}, the entries of $\mathbf{W}$ are primitives of the production function and depend on each unit's sales and cost shares. For example, To capture the effect of demand shocks propagating from $j$ to $i$, the implied weight is $w_{ij} = \frac{sales_{i\to j}}{sales_i}$. Proposition \ref{prop:necessary restrictions on W} requires that $w_{ij}w_{ji}  = \frac{sales_{i\to j}}{sales_i} \cdot \frac{sales_{j\to i}}{sales_j} \geq 1$ for all $i\neq j$, which implies that sectors which use each other's inputs can only use each other's inputs. 

Intuitively, one might argue that the static network model in \eqref{eq:reduced form static propagation} is too parsimonious to capture all the sources of risk comovement in the economy. Although this is likely true, the restrictions on $\mathbf{W}$ cannot be relaxed by adding omitted macroeconomic factors driving common variation in risk nor by adding an omitted network component. Moreover, Proposition \ref{prop:necessary restrictions on W} implies that there is no sufficient statistic that can be obtained from $\mathbf{W}$ which fully characterizes cross-sectional variation in granular risk, even in a world where sectoral shocks are identically distributed. See Appendix \ref{asec:simulation evidence} for supporting numerical evidence. In the remainder of this section, I explore the implications of allowing for correlation in demand and supply shocks across units.

\subsection{Granular Volatility with Correlated Shocks}
I investigate the volatility predictions of \eqref{eq:reduced form static propagation} when shocks $u_i$ are allowed to be correlated across units $i$ (i.e., $\var({\mathbf{u}})$ is not diagonal). Practically speaking, there are several reasons why supply and demand shocks to units might be correlated, especially at the granular level. For example, two sectors or firms that produce related goods or services are likely to experience correlated demand shocks. If the two sectors produce using the same inputs, then supply-side shocks might be correlated as well. \citet{Hoberg2016textindustries} show that firms with similar products might belong to different industries (according to SIC or NAICS classifications).\footnote{More specifically, \citet{Hoberg2016textindustries} find that firms in the newspaper, printing, and publishing industry (SIC3 271) are similar to firms in the radio broadcasting industry (SIC3 483) and argue that this is driven by common customers who demand advertising services. They also find that Disney and Pixar have similar products (movies) although they are in different industries (business services (SIC3 737) and motion pictures (SIC3 781) industries, respectively). In this case, the differences in industry stem from the production method and not the product offering.} In the even more simple setting where multiple firms produce the exact same goods and services, supply and demand shocks at the product level mechanically generate correlation in supply and demand shocks at the firm level. \citet{Hottman2016} provide empirical evidence that this is generally the case, with 69\% of firms, which account for 99\% of industrial output, supplying multiple (and intersecting) products.

Like product proximity, both technological and geographic proximity might also generate correlation in firm and sectoral shocks. For instance, \citet{bloomtechspillovers2013ecma} show that shocks to research and development (R\&D) have correlated effects on the productivity and growth of firms with similar technologies. Similarly, industries or firms operating in nearby geographies are exposed to the same underlying shocks associated with local labor markets (see e.g., \citet{autordornhanson2013}, \citet{Mian2014}), local factor prices (\citet{tuzel2017local}, Grigoris (2019)), local technological progress (\citet{oberfield_theory_2018}), or local  weather events (\citet{BarrotSauvagnat2016}, \citet{Kruttli2019pricingposeidon}).

In the benchmark network model with uncorrelated shocks, the variance of growth rates depends solely on the concentration of risk across independent suppliers and/or customers. \citet{herskovic_firm_2020} provide theoretical and empirical evidence linking firm volatility and customer concentration in terms of size dispersion in this setting. However, allowing for correlated shocks implies two additional variance components. These components capture the concentration and substitutability of risk \textit{between} trade partners, respectively. The distinction between concentration ``across" and ``between" trade partners is important. Concentration across refers to the composition of a unit's reliance on any particular customer or supplier, while concentration between refers to how a unit's the distribution of reliance on a set of customers or suppliers that are exposed to the same shocks. On the other hand, substitutability between customers and suppliers captures the distribution of reliance on a diversified set of customers or suppliers that are exposed to shocks of the opposite sign.

In other words, concentration between customers and suppliers captures compounding effects of positively related shocks to similar trade partners, while substitutability captures mitigating effects of spreading reliance on trade partners that are exposed to negatively related shocks. Intuitively, a supplier with major customers that tend to reduce demand at the same time is more risky than a supplier with some customers that increase demand when the others reduce it. To see this mathematically, define $[h_{ij}]_{ij}$ to be the set of entries in the Leontief inverse matrix $\mathbf{H}:= (\mathbf{I} - \mathbf{W})^{-1}$ and recall that equation \eqref{eq:reduced form static propagation} can equivelently by written as $\mathbf{y} = \mathbf{H} \mathbf{u}$. Note that in this setup the element $h_{ij}$ captures the percent change in unit $i$'s growth after a 1\% shock to unit $j$. Then the variance of unit $i$'s growth rates can be written:
\begin{align*}
    \var(y_i)& = \var\bigg(\sum_{j=1}^n h_{ij} u_j\bigg) = \sum_{j=1}^n h_{ij}^2 \cdot \var(u_j) + \sum_{j\neq k} h_{ij} h_{ik} \cdot \text{cov}(u_{j}, u_k).
\end{align*}
The first term is the standard expression for variance in this network model (see e.g., \citet{Acemoglu2012}), while the second term is only non-zero when inter-industry shocks are correlated. Next, I define the scalar $s_{jk}$ to be the sign of the pairwise correlation between shocks to $j$ and $k$ (i.e., $s_{jk} := \sgn(\sigma_{jk}) \equiv \sgn(\cov(u_j,u_k))$ where  $\sgn(.)$ is the sign function). To build some more intuition on the additional terms, I can further decompose the covariance term as follows:

\begin{align}
    \var(y_i) = \underbrace{h_{ii}^2 \cdot \sigma_i^2}_{\text{self}} + \underbrace{\sum_{j\neq i} h_{ij}^2 \cdot \sigma_{j}^2}_{\text{concentration ``across"}}  + \underbrace{\sum_{j\neq k, s_{jk}=1} h_{ij} h_{ik}\cdot \sigma_{jk}}_{\text{concentration ``between"}}+  \underbrace{\sum_{j\neq k, s_{jk} = -1} h_{ij} h_{ik} \cdot \sigma_{jk}}_{\text{substitutability}}
    \label{eq:volatility decomposition static network model}
\end{align}

Consider a first order approximation of the Leontief inverse matrix such that $\mathbf{H} \approx \mathbf{I} + \mathbf{W}$, where the weights in the propagation matrix $\mathbf{W}$ are related to sales shares when modeling downstream propagation supply-side shocks, and purchase shares when modeling upstream propagation of demand-side shocks.\footnote{More specifically, the weight of downstream propagation of supply-side shocks from supplier $j$ to customer $i$ is captured by $w_{ij}^d = sales_{j\to i} / purchases_i$ and the weight of upstream propagation of demand shocks from customer $j$ to supplier $i$ is captured by $w_{ij}^u = sales_{i\to j} / sales_i$. In general, these weights are both asymmetric (i.e., $w_{ij} \neq w_{ji})$ and different depending on the direction of propagation (i.e., $w_{ij}^u \neq w_{ij}^d$).} Suppose also that units are homoskedastic such that $\var(u_j) = \sigma^2$ for all $j$ and $\text{cov}(u_j,u_k) = \nu \cdot s_{jk} $  for all $j\neq k$, where $\sigma$ and $\nu$ are positive scalars. In this case, the first term ($h_{ii}\sigma_i^2 = \sigma_{i}^2$) is unrelated to the network and captures the variance of supply or demand shocks to sector $i$. On the other hand, the second component (concentration across network linkages) is non-negative and large when reliance is highly concentrated across trade partners. Similarly, the third term (concentration between network linkages) is non-negative and large when reliance is concentrated between trade partners who experience positively correlated shocks. 

Finally, the last term (substitutability of network linkages) is always non-positive and is large in magnitude when reliance is spread equally between trade partners who are likely to experience shocks of opposite sign. Additionally, the sum of the final two terms captures explicitly the trade-off between concentration and substitutability of correlated supply or demand shocks.  Although this simplification is useful for building intuition, the more realistic version of the variance decomposition should also take into account unit heteroskedasticity. That is, two sectors with an equal set of input-output weights have different network-implied variance only if their trade partners are exposed to differential volatility in supply or demand shocks. 

Consider for example the Printed Circuit Boards industry (SIC 3672), whose top 3 major manufacturing industry customers include Electronic Components (SIC 3679) and Electronic Computers (SIC 3571), and Communications Equipment (SIC 3669). At first glance, these customers appear very similar, and one might suspect that a negative demand shock to one customer is likely to be correlated with a negative shock to the other, which amplifies upstream propagation to their shared supplier. In other words, the Printed Circuit Boards industry has high concentration between customers and a harder time substituting away from upstream effects demand shocks to its major customers.\footnote{I find that the average product similarity score between these customer industries is in the top 10\% (based on the similarity score developed in \citet{Hoberg2016textindustries}). Additionally, I find significant positive correlation in demand shocks to these industries such as changes of newly awarded federal defense procurement contracts.} On the other hand, the three most important customers of the Jewelry and Precious Metal industry (SIC 3911) include Watches, Clocks, and Clockwork Operated Devices (SIC 3873), Perfumes and Cosmetics (SIC 5048), and Drawing and Insulating of Nonferrous Wire (SIC 3357). In this case, customers produce seemingly unrelated goods (both durable and non-durable) and there is evidence the demand shocks have zero or negative pairwise correlation.\footnote{The average pairwise correlation in federal procurement shocks and Chinese import penetration shocks is -37\% and -22\%, respectively for the full set of Jewlery and Precious Metal customers.} In other words, the Jewelry and Precious Metal industry has is able able to substitute away from demand shocks propagating upstream from any individual customer.

There are similar examples of high concentration and substitutability on the supply-side. For instance, the Computer Storage Devices industry (SIC 3572) has a highly concentrated customer base composed of Electronic Components (SIC 3679), Electronic Coils, Transformers, and other Inductors (SIC 3677), Semiconductors and Related Devices (SIC 3674), and Electronic Connectors (SIC 3678). This industry is thus more exposed to correlated supply-side risk. On the contrary, the Meat Packing Plants industry (SIC 2011) can more easily substitute away from supply-side risk, with a more diversified set of major suppliers like Poultry Slaughtering and Processing (SIC 2015), Plastics Film and Sheet (3081), and Paper Mills (2621). 

Although these network variance components are intuitive and theoretically justified if supply and demand shocks are correlated, an important practical concern is that granular supply and demand shocks are not easily identified from available data, especially at a high frequency. In the following section, I address this challenge and propose an empirical methodology for estimating customer and supplier concentration and substitutability at the industry and firm levels. I show that both supply and demand channels explain cross-sectional heterogeneity in risk exposure, beyond what can be explained by other determinants of variance identified by the literature.


\section{Empirical Evidence}
\label{sec:empirical evidence}

In this section, I provide empirical estimates of the network-implied variance components motivated in equation \eqref{eq:volatility decomposition static network model}. This requires granular data on input-output relationships and estimates of the variance-covariance matrix of supply and demand shocks. Consistent with theory, I find that these additional components explain important variation in realized volatility, controlling for characteristics such as size, centrality, concentration across trade partners, vertical position in the supply chain, and durability of output. These results hold at both industry and firm levels. The main takeaway here is simple. When accounting for input-output linkages and non-negligible correlation in supply and demand shocks, heterogeneity in risk exposure is at least in part driven by differences in the ability of network units to substitute away from correlated supply and demand shocks.

\subsection{Setup}
Consider the $n$-sector network model from equation \eqref{eq:reduced form static propagation} and add a time subscript $t$. Suppose that in each period I obtain estimates for the $n\times n$ Leontief Inverse matrices $\hat{\mathbf{H}}_{q,t}$ and the variance-covariance matrices of supply and demand shocks $\hat{\mathbf{\Sigma}}_{q,t}$ where $q=u$ for upstream propagation and $q=d$ for downstream propagation. Then for both supply and demand-side shocks, I can compute three empirical network-implied variance components, denoted by ``self-originating", ``across", and ``between" risk. The final component sums the final two covariance terms from \eqref{eq:volatility decomposition static network model} and captures the concentration/substitutability trade-off  between trade partners. Low values of concentration between customers and suppliers implies high substitutability. Then for each industry, direction, and time triple $(i,q,t)$, I compute self-originating risk as:
\begin{align}
    \hat{\sigma}_{iqt,self}^2 = [\hat{h}_{qt}]_{ii}^2 \cdot \hat{\sigma}_i^2, 
    \label{eq: self empirical}
\end{align}
and across and between risk as:
\begin{align}
\hat{\sigma}_{iqt,acr}^2 & = \sum_{j\neq i}[\hat{h}_{qt}]_{ij}^2  \cdot \hat{\sigma}_j^2,
\label{eq:across empirical}\\\hat{\sigma}_{iqt,bet}^2 & = \sum_{j\neq k}[\hat{h}_{qt}]_{ij} \cdot  [\hat{h}_{qt}]_{ik}  \cdot \hat{\sigma}_{jk}.
\label{eq:between empirical}
\end{align}
In the next section, I provide details on the data sources, assumptions, and methodologies used to estimate $\hat{\mathbf{H}}_{qt}$ and $\hat{\mathbf{\Sigma}}_{qt}$. While the former can be observed directly, I need to make some assumptions to identify the latter from available data sources. Then I compute all three components at the industry-level and study their empirical relationship with realized industry variance. I find that the elasticity of realized variance to all three components is significant and positive for both directions of propagation, controlling for a variety of industry characteristics.

\subsection{Upstream and Downstream Propagation Networks}
I begin by constructing the network of input-output linkages at the disaggregated industry level from the Make and Use tables published by the Bureau of Economic Analysis (BEA). The goal is to build a directed weighted network which captures the importance of trade relationships over time and for the population of industries.\footnote{As far as I know, this is the most disaggregated database on the entire population of input-output relationships.} Network weights represent the strength of each unit's reliance on customers and suppliers, and the network is directed to capture differences in shock propagation in the upstream (customer to supplier) and downstream (supplier to customer) directions.  In particular, the BEA publishes these tables annually between 1997-2020 for 66 industry groups.

More specifically, I construct downstream and upstream propagation matrices $\mathbf{W}_d = [w_d]_{ij}$ and $\mathbf{W}_u = [w_u]_{ij}$ with entries:
\begin{align}
    \label{eq:industry upstream and downstream propagation defintions}
    [w_d]_{ij} = \frac{sales_{j\to i}}{costs_i}, \qquad [w_u]_{ij} = \frac{sales_{i \to j}}{sales_i},
\end{align}
where $sales_{i\to j}$ represents gross trade flows from $i$ to $j$, and $sales_{i}$ and $costs_i$ represent the total sales and costs of industry $i$, respectively. The downstream (upstream) weights are non-negative and capture the direct reliance of industry $i$ on supplier (customer) $j$. When weights are large, direct effects of propagated shocks should also be large. To account for higher order (indirect) network effects as well, I calculate the strength of network propagation based on the Leontief inverse $\mathbf{H}_q := (\mathbf{I} - \mathbf{W}_q)^{-1}$ of the propagation matrices $\{\mathbf{W}_q: q\in u,d\}$.\footnote{See e.g., \citet{baqaeeFarhiMacroimpact} and \citet{Herskovic2020JPE} for discussion on the importance of higher order network effects.} The entries of $\mathbf{H}_q = [h_q]_{ij}$ capture the total percent effect on $i$ of a 1\% shock to $j$ traveling in the $q$-stream direction when accounting for all weighted direct and indirect connections. 

Appendix  \ref{asec:summary input output tables} reports summary statistics for observed input-output connections. At the 66-industry granularity, I find that both propagation and Leontief inverse weights are highly persistent with an average autocorrelation of more than 95\% for each entry. The cross-sectional correlation is about 8.55\% between upstream and downstream weights and about 11.08\% between upstream and downstream Leontief matrix entries, suggesting that propagation occurs differently in either direction.

\subsection{Variance-Covariance Matrix of Supply and Demand Shocks}

Unlike the sectoral input-output network, there is no definitive data source on supply and demand shocks and their variance-covariance matrix. As a baseline, I implement an empirical analog of the reduced-form equation in \eqref{eq:reduced form static propagation}. In particular, consider the spatial panel regression:
\begin{align}
 \tilde{y}_{it} = \delta_{t} + \phi \cdot \tilde{y}_{i,t-1} + \beta_u \cdot \sum_{j=1}^n w_{u,ij} \tilde{y}_{jt}  + \beta_d \cdot \sum_{j=1}^n w_{d,ij} \tilde{y}_{jt} + \varepsilon_{it},
 \label{eq:spatial panel}
\end{align}
where $\tilde{y}_{it} := \log(y_{it}/y_{i,t-1})$ is output growth in industry $i$ at quarter $t$ and $w_{q,ij}$ is the $(i,j)$ entry of the $q$-stream propagation matrix $\mathbf{W}_q$. Assuming the variance-covariance matrix of residuals is static, then $\hat{\mathbf{\Sigma}}$ is the empirical variance-covariance matrix of estimated residuals $\hat{\varepsilon}_{it}$. To ensure that my estimates for network components \eqref{eq:across empirical} and \eqref{eq:between empirical} are robust to estimation error in $\hat{\mathbf{\Sigma}}$, I calculate the average value over samples in which I randomly drop 10\% of pairwise non-zero correlations.\footnote{Note that estimation error from $\hat{\mathbf{\Sigma}}_q$ is magnified in estimated network components \eqref{eq:across empirical} and \eqref{eq:between empirical} at a rate proportional to the number of nonzero row entries in the Leontief inverse matrix $\mathbf{H}_q$.} See Appendix \ref{asec:determinants of network variance} for more details and alternative specifications.

Pairwise correlation in residuals is centered with a mean value of 0.5\% (0.4\%) and a standard deviation of 25\% (26\%). The largest positive pairwise correlation is 82\% between the Primary Metals (BEA Code 331) and Wholesale Trade (BEA Code 42) and 81\% between Housing (HS) and Educational Services (61). On the other hand, the largest negative pairwise correlation is -80\% between Primary Metals (331) and Federal Reserve Banks, Credit Intermediation, and Related Activities (521CI) and -72\% between Food and Beverage and Tobacco Products (311FT) and Wholesale Trade (42).

\subsection{Network Determinants of Realized Variance}
After relaxing the idiosyncratic shock assumption, the benchmark input-output propagation model predicts that realized variance should depend positively on three network risk components: risk that is self-originating, risk across trade partners, and risk between trade partners. In the baseline setup, this might hold mechanically for self-originating risk since it is estimated from the variance of residual output growth in equation \eqref{eq:spatial panel}. However, both risk across and between trade partners contain only variance-covariance information associated with other industries.  I verify these predictions empirically using panel regressions of the log of realized industry variance on the log of network components, controlling for a variety of characteristics such as size, centrality, durability of output, and industry cluster and time fixed effects.\footnote{I adjust network components by a constant to ensure that the minimum value is positive so the log is well defined. Industry clusters are defined by major industry groups (2-digit NAICS code).} I measure realized industry variance using both stock market and output growth data. I define market variance as the annual return variance of an equal-weighted industry portfolio and fundamental variance as the variance of quarterly year-on-year output growth. I obtain similar results when using idiosyncratic variance as the dependent variable.\footnote{I define idiosyncratic market variance as the variance of equal weighted residual returns from a Fama and French three-factor model. Similarly, I define idiosyncratic output growth as the residual of industry output growth after a regression on aggregate output growth. Results also replicate for value-weighted industry portfolios, or industry sales growth, which is constructed as the year-on-year change in the sum of quarterly sales (reported on Compustat) for all public firms in the industry.} Although the annual variance across quarterly year-on-year output growth and monthly returns are fairly noisy proxies for true realized cash-flow variance, the results are robust for several specifications.

I summarize the main results in Table \ref{tab:var regressions main}. Consistent with theoretical predictions in equation \eqref{eq:volatility decomposition static network model}, the elasticity of realized industry variance to concentration across and between customers are both positive and significant in all specifications. This holds for both market and output growth measures of variance. Conditional on both directions of propagation and all controls, increasing concentration between customers from the median to the $90^{th}$ percentile increases industry sales growth variance by over 45\% (about 0.37 standard deviations) and market variance by 15\% (about 0.09 standard deviations). Similarly, increasing concentration between suppliers from  the median to the $90^{th}$ percentile increases industry sales growth variance by over 20\% (about 0.17 standard deviations) and market variance by 19\% (about 0.11 standard deviations). Without controls, downstream network risk explains 23\% of time series variation in market variance and 22\% of time series variation in output growth variance. Similarly, upstream network risk explains 22\% and 31\% of market and output growth variance, respectively. Both directions of propagation are important for explaining the panel dynamics of industry variance. 

Consistent with the firm-level findings of \citet{herskovic_firm_2020}, I find that industry variance has a positive elasticity to concentration across customers and a negative elasticity to average size. A new but related result is the positive elasticity of variance to concentration across suppliers. Additionally, \citet{Ahern2013NetworkCA} argues that more central industries have greater market risk since they are more exposed to aggregate shocks, and thus earn higher returns on average. On the other hand, my results suggest that more central industries have less volatile stock returns, but also have less exposure to aggregate volatility risk and lower idiosyncratic volatility.\footnote{I find that industries in the highest average upstream (downstream) centrality decile have 31\% (21\%) less exposure to systematic volatility risk than the lowest decile. Average upstream and downstream centrality are positively correlated (56\% cross-sectionally), and industries who are in the top decile for both average centrality measures have a 52\% lower exposure to aggregate volatility risk than industries in the bottom decile for both. Moreover, top centrality decile stocks have 25\% lower idiosyncratic volatility, on average.} My results are thus consistent with \citet{Ahern2013NetworkCA}, since stocks with lower exposure to aggregate volatility risk or and lower idiosyncratic volatility earn higher returns on average (see e.g., \citet{AngVolatilityPuzzle}). Table \ref{tab:corr vol predictors} shows that there is no significant relationship between centrality and concentration between or across trade partners.

\subsection{Sources of Correlation in Network Propagating Shocks}
\label{sec:alternative covariance measures}
So far, I have established both theoretically and empirically the importance of accounting for correlation in shocks that propagate through the input-output network. In particular, I show that concentration between trade partners explains a large amount of variation in the industry panel of realized variance. However, statistical estimates for the variance-covariance matrix of shocks do not provide much insight on the underlying sources of correlation between industries. In this section, I argue that correlation in supply-side shocks that propagate downstream can be explained by technological proximity between sectors, while correlation in demand-side shocks that propagate upstream can be explained by product similarity.

\subsubsection{Observed Supply and Demand Shocks} 
 \citet{AcemogluAkjitKerr2016} argue that productivity shocks primarily propagate downstream while government spending and trade shocks primarily propagate upstream. In this case, these shocks might help to capture differences in inter-industry correlations which are specific to the direction of propagation. Along these lines, I construct an annual industry panel of 5-factor total factor productivity (TFP) growth between 1959-2018 from the NBER-CES Database \citep{nbercesdata}. Since this measure of TFP controls for materials, it does not mechanically encode any information related to downstream effects such as changes in price and/or quantity. Similarly, I construct a monthly panel of newly awarded federal procurement contracts between Jan 2000-Jan 2021 from the universe of contracts published in the Federal Procurement Data System (FPDS).\footnote{I also consider other observed shocks in Appendix \ref{asec:determinants of network variance}}.  

 To focus on inter-industry correlations which are unrelated to common aggregate factors (e.g., the secular decline in several manufacturing industries), I estimate the variance-covariance matrix of residuals after an OLS regression on the cross-sectional average of shocks.\footnote{The cross-sectional mean approximates the first principal component of shocks when there are missing values. For shocks $dz_t$, I calculate the covariance between sectors $k$ and $j$ as $\cov(u_{kt},u_{jt})$ where $u_{kt}$ is the residual in the regression $dz_{kt} = \alpha + \beta \cdot \bar{dz}_{.,t} + u_{kt}$. Endogeneity of shocks is not a major concern assuming any confounding shocks largely propagate in the same direction in the network. Given such a confounder, my estimate for the variance-covariance matrix of shocks can be written as the true estimate plus some measurement error. } I then estimate the corresponding network components using equations \eqref{eq:across empirical} and \eqref{eq:between empirical} and study their relationship with realized variance. Table \ref{tab:industry variance panel regressions (tfp)} shows that the elasticity of realized variance to supplier concentration is positive and more precisely estimated when calculating supply-side shock covariance as a function of productivity growth. On the other hand, Table \ref{tab:industry variance panel regressions (trade)} reports more precise estimates for the elasticity of variance to customer concentration when calculating demand shock covariance as a function of federal procurement shocks. This suggests that productivity growth is more informative about upstream network risk, while changes in government demand are more informative about downstream network risk.

\subsubsection{Technological and Product Proximity}
On the other hand, suppose that correlation between supply and demand-side shocks is a function of underlying firm and industry characteristics. Intuitively, I might assume that correlation in demand shocks propagating upstream is driven by product similarity and/or geographic proximity of customers, while correlation in supply shocks propagating downstream is driven by technological similarity and/or geographic proximity of suppliers. 

More generally, I assume that each industry is associated with a vector of positions in some latent surface $\mathbf{z}_{it} \in \Omega_z \subset \mathbb{R}^d$ and that the correlation between industry shocks can be written as a function of the distance between these latent vectors.\footnote{Latent surface models are often used to impute network relationships in microeconomic applications (see e.g., \citet{McCormickZhengLatentARD}, \citet{BrezaChandrasekar2020}). } Following \citet{McCormickZhengLatentARD}, I suppose that industry positions $\mathbf{z}_{it}$ lie on the surface of a $p$-dimensional latent surface on the $p+1$-dimensional unit hypersphere $\mathcal{S}^{p+1}$. This implicitly implies that latent positions follow a uniform distribution across the sphere's surface. Moreover, since the hypersphere has bounded surface area, the distance between any two points is bounded. I further assume that points in the same position have correlation 1 and points on opposite sides of the sphere have correlation -1.

In practice, I experiment with constructing latent positions of industries using several combinations of industry variables. For simplicity, my main results rely on univariate distances in product and technology space.\footnote{Moreover, contours of the sphere present some calibration difficulties in higher dimension.} I measure product distance using use the text-based scores developed in \citet{Hoberg2016textindustries} and technology distance using patent-based technological proximity scores along the lines of \citet{bloomtechspillovers2013ecma}. Since both of these scores are available at the firm-level, I first construct a firm-by-firm product distance network where distances are inversely related to proximity. To get the distance between sectors, I use the median length of the shortest weighted path between firms in the two sectors, rescaled such that the furthest pairwise distance is 1 and the shortest pairwise distance is zero. I calculate the shortest pairwise path between any two nodes using Dijkstra's Algorithm. I calculate these measures annually. 

Transforming distances to correlations, I rescale by the variance estimates from residuals in equation \eqref{eq:spatial panel} and recompute network components. When using product distance to calculate network risk, the elasticity of realized variance to concentration between customers is significant and positive. On the other hand, the analogous elasticity to concentration between suppliers is significantly negative, which suggests that product similarity across suppliers actually indicates better substitutability away from supply-side shocks. When approximating correlations based on technological distance, realized variance has a positive elasticity to concentration between suppliers and customers, but the elasticity is more precisely estimated on the supply side. Taken together, these results suggest that technological proximity is a good proxy for correlated exposure to supply shocks propagating downstream, while product proximity is a good proxy for correlated exposure to demand shocks propagating upstream. Along these lines, Table \ref{tab:correlation network variance components} shows that average technological proximity between sectors is closely related to correlation in TFP growth shocks, while product similarity is closely related to correlation in federal defense procurement shocks.

\subsubsection{Accounting for Dynamics}

To account for potential time-variation in industry correlations, I also compute pairwise inter-industry correlations using the dynamic conditional correlation (DCC) estimator from \citet{engeldcc}. In particular, I estimate bivariate DCC models for all pairs of industries using Bayesian MCMC following \citet{Fioruci2013}. Since product and technological proximity are computed at an annual frequency, I am thus able to obtain a one-to-one comparison between correlation in spatial panel residuals, observed shocks, and distance based measures. Table \ref{tab:correlation network variance components} reports consistent results when using dynamic correlations.

\newpage
\begin{table}[H]
    \centering
  \scriptsize
  \caption{\textbf{Network Determinants of Industry Variance}}
    \begin{tabular}{lcccccc}
    \toprule
    \multicolumn{7}{c}{Panel A: Market Return Variance} \\
    \midrule
          & (1)   & (2)   & (3)   & (4)   & (5)   & (6) \\
    \midrule
    Self-origin (demand) & 0.055** & 0.056** &       &       & 0.003 & 0.017 \\
          & (0.022) & (0.024) &       &       & (0.028) & (0.028) \\
    Across (demand) & 0.094** & 0.083** &       &       & 0.081** & 0.072** \\
          & (0.036) & (0.030) &       &       & (0.024) & (0.021) \\
    Between (demand) & 0.147*** & 0.122** &       &       & 0.197*** & 0.089** \\
          & (0.051) & (0.049) &       &       & (0.053) & (0.038) \\
    Self-origin (supply) &       &       & 0.072*** & 0.073*** & 0.067*** & 0.083** \\
          &       &       & (0.021) & (0.023) & (0.022) & (0.023) \\
    Across (supply) &       &       & 0.216*** & 0.156*** & 0.160*** & 0.154* \\
          &       &       & (0.056) & (0.054) & (0.060) & (0.067) \\
    Between (supply) &       &       & 0.416*** & 0.317*** & 0.210** & 0.196** \\
          &       &       & (0.093) & (0.095) & (0.088) & (0.073) \\
    Size  &       & -0.378*** &       & -0.361*** &       & -0.289*** \\
          &       & (0.094) &       & (0.091) &       & (0.104) \\
    Upstream centrality &       & -0.182* &       & -0.289*** &       & -0.232** \\
          &       & (0.103) &       & (0.095) &       & (0.101) \\
    Downstream centrality &       & -0.051 &       & -0.086** &       & -0.023 \\
          &       & (0.054) &       & (0.043) &       & (0.054) \\
    Durability &       & -0.167 &       & -0.404 &       & -0.637 \\
          &       & (0.573) &       & (0.662) &       & (0.615) \\
    Vertical position &       & 1.550** &       & -0.756** &       & 1.970*** \\
          &       & (0.701) &       & (0.332) &       & (0.710) \\
    Constant & -6.279 & -3.26 & -4.692 & -1.082 & -5.005 & -3.608 \\
    Obs   & 1484  & 1484  & 1484  & 1484  & 1484  & 1484 \\
    Adj R2 & 0.231 & 0.292 & 0.159 & 0.223 & 0.245 & 0.330 \\
    \midrule
    \multicolumn{7}{c}{Panel B: Output Growth Variance} \\
    \midrule
          & (1)   & (2)   & (3)   & (4)   & (5)   & (6) \\\midrule
    Self-origin (demand) & 0.034 & 0.006 &       &       & 0.026 & 0.006 \\
          & (0.026) & (0.026) &       &       & (0.030) & (0.028) \\
    Across (demand) & 0.159** & 0.074 &       &       & 0.157** & 0.095 \\
          & (0.069) & (0.024) &       &       & (0.070) & (0.057) \\
    Between (demand) & 0.210*** & 0.281*** &       &       & 0.196*** & 0.280*** \\
          & (0.064) & (0.078) &       &       & (0.065) & (0.082) \\
    Self-origin (supply) &       &       & 0.081** & 0.072 & 0.017 & 0.042 \\
          &       &       & (0.021) & (0.022) & (0.024) & (0.025) \\
    Across (supply) &       &       & 0.128** & 0.114** & 0.098** & 0.091** \\
          &       &       & (0.053) & (0.033) & (0.021) & (0.031) \\
    Between (supply) &       &       & 0.332** & 0.241*** & 0.221** & 0.198** \\
          &       &       & (0.098) & -0.079 & (0.103) & (0.092) \\
    Size  &       & -0.046 &       & -0.185* &       & -0.110 \\
          &       & (0.099) &       & (0.098) &       & (0.106) \\
    Upstream centrality &       & -0.127 &       & -0.096 &       & -0.131 \\
          &       & (0.108) &       & (0.095) &       & (0.111) \\
    Downstream centrality &       & -0.218*** &       & -0.067 &       & -0.222*** \\
          &       & (0.056) &       & (0.042) &       & (0.056) \\
    Durability &       & 0.115 &       & -0.038 &       & 0.024 \\
          &       & (0.647) &       & (0.643) &       & (0.703) \\
    Vertical position &       & 1.545** &       & 0.181 &       & 1.375** \\
          &       & (0.654) &       & (0.355) &       & (0.682) \\
    Constant & -5.088 & -6.954 & -6.818 & -5.327 & -6.17 & -6.718 \\
    Obs   & 1484  & 1484  & 1484  & 1484  & 1484  & 1484 \\
    Adj R2 & 0.221 & 0.382 & 0.198 & 0.319 & 0.277 & 0.412 \\
    \bottomrule
    \end{tabular}%
    \fnote{\tiny \textit{Notes}: This table reports panel regressions of realized industry variance on a variety of characteristics including the log variance of supply and demand shocks (self-origin), log concentration across trade partners, log concentration between trade partners, log total output (size), log centrality of the upstream and downstream propagation networks, durability of output, vertical position in the supply chain, and industry cluster and year fixed effects. In Panel A, the dependent variable is the log variance of annualized monthly returns on an equal-weighted industry portfolio. In Panel B, the dependent variable is the log variance of total quarterly year-on-year industry sales growth. I obtain return data from CRSP and GDP data from the BEA. Concentration between and across trade partners is calculated as the average value of estimates obtained using equations. \ref{eq:across empirical} and \ref{eq:between empirical}. I calculate the variance-covariance matrix of shocks using residuals from equation \ref{eq:spatial panel} and calculate the average value of components over 1000 random samples each randomly dropping 10\% of pairwise correlations.  Following \citet{Ahern2013NetworkCA}, I compute industry centrality as the eigenvector centrality of upstream and downstream propagation adjacency matrices. I calculate durability as the proportion of sub-industries classified as durable by \citet{Gomes2009durability}, and I calculate vertical position of each industry following \citet{Antrs2012} and \citet{GofmanSegalWu2020}. $***$, $**$, and $*$ indicate significance at the 1\% and 5\% and 10\% levels, respectively. Standard errors are clustered at the BEA major 15 major industry group level. Sample is at an annual frequency from 1997 to 2019 for 66 BEA non-government industries.}
  \label{tab:var regressions main}%
\end{table}%

\newpage
\begin{table}[H]
    \centering
    \caption{\textbf{Sources of Correlation in Supply and Demand Substitutability}}
     \begin{tabular}{lccccc}
    \toprule
    \multicolumn{6}{c}{Panel A: Upstream Supplier Substitutability (static)} \\
    \midrule
    \multicolumn{1}{c}{Covariance Method} & Spatial & TFP   & Procurement & Prod Similarity & Tech Proximity \\ \midrule
    Spatial  & 1     & 0.59  & 0.37  & 0.30  & 0.48 \\
    TFP   &       & 1     & 0.31  & 0.35  & 0.46 \\
    Procurement &       &       & 1     & 0.62  & 0.39 \\
    Prod Similarity &       &       &       & 1     & 0.27 \\
    Tech Proximity &       &       &       &       & 1 \\
    \midrule
    \multicolumn{6}{c}{Panel B: Downstream Customer Substitutability (static)} \\
    \midrule
    \multicolumn{1}{c}{Covariance Method} & Spatial & TFP   & Procurement & Prod Similarity & Tech Proximity \\ \midrule
    Spatial  & 1     & 0.55  & 0.61  & 0.43  & 0.35 \\
    TFP   &       & 1     & 0.27  & 0.20  & 0.48 \\
    Procurement &       &       & 1     & 0.69  & 0.40 \\
    Prod Similarity &       &       &       & 1     & 0.24 \\
    Tech Proximity &       &       &       &       & 1 \\
    \midrule
    \multicolumn{6}{c}{Panel C: Upstream Supplier Substitutability (dynamic)} \\ 
    \midrule
    \multicolumn{1}{c}{Covariance Method} & Spatial & TFP   & Procurement & Prod Similarity & Tech Proximity \\\midrule
    Spatial  & 1     & 0.61  & 0.40  & 0.34  & 0.49 \\
    TFP   &       & 1     & 0.35  & 0.39  & 0.48 \\
    Procurement &       &       & 1     & 0.65  & 0.42 \\
    Prod Similarity &       &       &       & 1     & 0.25 \\
    Tech Proximity &       &       &       &       & 1 \\
    \midrule
    \multicolumn{6}{c}{Panel D: Downstream Customer Substitutability (dynamic)} \\
    \midrule
    \multicolumn{1}{c}{Covariance Method} & Spatial & TFP   & Procurement & Prod Similarity & Tech Proximity \\ \midrule
    Spatial  & 1     & 0.57  & 0.68  & 0.45  & 0.30 \\
    TFP   &       & 1     & 0.22  & 0.27  & 0.52 \\
    Procurement &       &       & 1     & 0.70  & 0.38 \\
    Prod Similarity &       &       &       & 1     & 0.30 \\
    Tech Proximity &       &       &       &       & 1 \\
    \bottomrule
    \end{tabular}%
    \fnote{\scriptsize \textit{Notes:} This table reports the correlation across different measures of upstream and downstream substitutability measures (negative of concentration ``between"). Substitutability is calculated as the negative value of the log of \eqref{eq:between empirical} (plus a large enough constant) and correlation between two measures $x_{it}$ and $y_{it}$ is defined by $\hat{\rho}$ from the regression $y_{it} = \hat{\rho }x_{it} + u_{it}$, where $x_{it}$ and $y_{it}$ are transformed to have mean zero and standard deviation one. Panels A and B rely no a static assumption for the variance-covariance matrix across shocks, while Panels C and D estimate a dynamic conditional variance-covariance matrix \`{a} la \citet{Fioruci2013}. Panels A and C report results in the upstream (supply-side) direction and Panels B and D report results in the downstream (customer-side) direction. Spatial covariance is based on the panel model in equation \eqref{eq:spatial panel}. TFP covariance is based on TFP growth measured in \citet{nbercesdata}, procurement based on federal government shares interacted with the procurement proxy in \citet{BrigantiSellemiProcurement}, product similarity using the latent distance method and scores from \citet{Hoberg2016textindustries}, and tech proximity constructed following \citet{bloomtechspillovers2013ecma}. }
    \label{tab:correlation network variance components}
\end{table}


\newpage
\section{Dynamic Network Model of Supply-Chain Substitutability}
\label{sec:theoretical model}

In this section, I incorporate the supply chain substitutability-concentration trade-off and my empirical results from the previous section in a structural dynamic asset pricing model detailed in Appendix \ref{asec:General Equilibrium Model of Input-Output Linkages}. This model builds on existing production-based models with input output networks (see e.g., \citet{Ramirez2017}, \citet{herskovic_networks_2018}, or \citet{GofmanSegalWu2020}). Unlike existing models, I introduce a correlation structure in shocks to firm growth rates which propagate both upstream and downstream in an input-output network.\footnote{To my knowledge, this is the first network model to feature both correlated shocks and two directions of propagation.} Firms are subject to both productivity shocks which propagate downstream and demand shocks which propagate upstream. Shocks are drawn from a joint distribution with finite second moments and in which correlation across firm-level shocks is induced by shared variation in firms' input-output substitutability. More specifically, firms' ability to substitute away from productivity (demand) shocks is inversely related to concentration of trade partners in latent technology (product) space and correlated between firms who share trade partners.

\subsection{Setting}

Consider a discrete-time economy with $n$ distinct goods and $n$ firms. Output goods are characterized by vector in product-technology space, which is fixed exogenously for each good. Firm output (cash flow) depends on aggregate economic conditions and the cash flows of its customers and suppliers. There are two kinds of random shocks in this economy, productivity shocks which propagate downstream from suppliers to customers, and demand shocks which propagate upstream from customers to suppliers. The input-output network is captured by two sequences of graphs with $n$ nodes for each firm and weighted directed edges capturing the importance of firm-to-firm trade relationships. In the customer (supplier) network, the edge from $i$ to $j$ represents the relative reliance of $j$ on customer (supplier) $i$.

For tractability, I assume that trade relationships are exogenously determined at the start of each period. Additionally, the model features a representative investor with constant relative risk aversion (CRRA) preferences who owns all firms and lives off labor wages and dividends. Next, I describe the process for firm cash flows, the network structure, and the mechanism of shock propagation through the input-output network. Then I derive equilibrium consumption growth and asset prices. For ease of exposition, I provide details on the production side of the economic since that is the primary source of risk. Further details are left to Appendix \ref{asec:General Equilibrium Model of Input-Output Linkages}.

\subsection{Substitutability and Firms' Cash Flows}

Firms are exposed to undiversifiable aggregate risk factors and risk from trade partners which can be mitigated with diversification of customers and suppliers. Every firm is both a customer who purchases inputs from other firms, and a supplier who produces a single final good. Final goods are characterized by a latent position in technology-product space, which fluctuates according to a persistent stationary process discussed in Appendix \ref{asec:product varieties latent space}.\footnote{This assumption is justified empirically by the results of Section \ref{sec:empirical evidence}, which suggest both that distances in product space and technology space change over time.} Latent position dynamics are exogenous to firm and household decisions and can be interpreted as random changes in product differentiation. For example, \citet{Syverson2004sub} argues that the same products might be perceived differently as a result of intangible factors like delivery speed, documentation, product support, or branding and advertising.

In reduced-form, firm cash flows are determined by random shocks each period which propagate stochastically both downstream to customers and upstream to suppliers. 
The probability that a shock propagates through the supply chain is a function of firms' customer and supplier substitutability. Consistent with the empirical results from Section \ref{sec:empirical evidence}, I assume that a firm's customer substitutability depends on the product diversity of the goods sold by its customers. Likewise on the supply-side, a firm's suppliers substitutability depends on the technological diversity of its suppliers. When supply chains are highly substitutable, shocks are less likely to propagate.

In particular, firm cash-flow growth has the following reduced-form equation: 
\begin{align}
    \label{eq:firm level cash flows}
    \Delta y_{i,t+1} = \Delta z_{i,t+1} + \Delta g_{i,t+1},
\end{align}
where $\Delta z_{i,t+1} = \log(z_{i,t+1} / z_{it})$ is a shock to productivity and $\Delta g_{i,t+1} = \log(g_{i,t+1} / g_{it})$ is a shock to government demand. I assume that dependence across shocks is determined by both the firm's input-output network and the relative location of its final good in product-technology space. Productivity growth follows the process:
\begin{align}
    \Delta z_{i,t+1} = \gamma_{u}\cdot a_{t+1} - \beta_u\cdot  \varepsilon_{iu,t+1},
    \label{eq:productivity growth}
\end{align}
where $a_t\sim_{iid} \mathcal{N}(0,\sigma_a^2)$ is aggregate productivity growth at time $t$, $\gamma_u$ and $\beta_u$ are positive scalars, and $\varepsilon_{iut}$ is a Bernoulli shock that negatively affects productivity and originates upstream. Similarly, government demand growth follows the process:
\begin{align}
    \Delta g_{i,t+1} = \gamma_{d}\cdot g_{t+1} - \beta_d\cdot  \varepsilon_{id,t+1},
    \label{eq:government demand growth}
\end{align}
where $g_t\sim_{iid} \mathcal{N}(0,\sigma_g^2)$ is aggregate growth in government spending at time $t$, $\gamma_d$ and $\beta_d$ are positive scalars, and $\varepsilon_{idt}$ is a Bernoulli shock which negatively affects demand and originates downstream. In other words, $\varepsilon_{idt}$ ($\varepsilon_{iut}$) is equal to one when firm $i$ experiences a demand (supply) shock which originates at $i$ and/or propagates from its downstream customers (upstream suppliers). Shocks propagating in different directions are independent (i.e., $\varepsilon_{idt} \perp \varepsilon_{iut}$ for all $i$ and $t$).

\subsection{Network Structure}

The propagation of shocks depends on the sequence of input-output network connections between firms, defined as follows. The sequence of upstream and downstream graphs $(\mathcal{G}_{n,u,t})_{n}$ and $(\mathcal{G}_{n,d,t})_n$ are $n$-node graphs with weighted edges given by $w_{ijt}^u$ and $w_{ijt}^d$, respectively. Weights capture the importance of the directed relationship $i\to j$ from the perspective of $i$ and are fixed exogenously at the start of period $t$.

To ensure that the input-output network is realistic, I assume that all weights are between 0 and 1 and introduce some additional restrictions on the growth rates of input-output connections relative to the number of firms. In particular, I assume that the number of shared customers and suppliers between two firms cannot grow at a rate faster than the total number of firms in the economy $n$, and that the maximum number of firm suppliers or customers must grow slower than the total number of possible edges. First consider the following definitions.

\begin{definition}[Paths] A $k$-path between nodes $i$ and $j$ in graph $\mathcal{G}$ is a length $k$-sequence $\{a_\ell\}_{\ell=1}^k$ where $a_1 = i$, $a_k = j$, and $w_{a_\ell a_{\ell+1}} > 0$ for all $\ell = 1,...,k-1$. Denote by $A_{ij}(\mathcal{G})$ the set of paths between nodes $i$ and $j$ and by $A_{i} := \{k: A_{ki}(\mathcal{G}) \neq \emptyset \}$ the set of nodes for which a path to $i$ exists.
\end{definition}
\begin{definition}[Maximal Dependency] The maximal dependency of an $n$-vertex graph $\mathcal{G}_{n}$ is given by:
\begin{align}
\bar{M}_n(\mathcal{G}_{n}) := \sup_{i,j} \bigg[ \text{card} \big( A_{i}(\mathcal{G}_{n}) \cap A_{j} (\mathcal{G}_{n}) \big) \bigg]
\end{align}

\end{definition}
\begin{definition}[Maximal Degree] The maximal (unweighted) degree in an $n$-vertex graph $\mathcal{G}_{n}$ is given by:
\begin{align}
    \bar{D}_n(\mathcal{G}_{n}) = \sup_{i} \bigg[ \sum_{j=1}^n \text{card} \big( A_{ji}(\mathcal{G}_n)\big) \bigg]
\end{align}
\end{definition}
\noindent If only direct connections exist, then $\bar{M}_n(\mathcal{G}_n) = \sup_{i,j}\sum_{k=1}^n \mathbbm{1}_{\{w_{ki} > 0\}}\mathbbm{1}_{\{w_{kj} > 0\}}$ and $\bar{D}_n(\mathcal{G}_n) = \sup_{i}\sum_{j=1}^n \mathbbm{1}_{\{w_{ji} > 0\}} $. Given these definitions, the following assumptions formally restricts the growth rate of input-output connections as the number of firms $n$ grows. These assumptions are fairly general and relevant for deriving tractable theoretical properties of the model. 

\begin{assumption}[Bounded Growth Rate of Maximal Degree Sequence] \label{ass:bounded maximal dependency} For all $q$ and $t$, the maximal degree sequence grows at a rate strictly less than $n^2$:
\begin{align*}
    \bar{D}_{nq} = o(n^2)
\end{align*}
\end{assumption}
These assumptions are intuitive and weaker than the restriction that no firms can serve as a customer or supplier to all other firms. In this case, both the maximal dependency and the maximal degrees must grow at a rate slower than $n$.

\begin{assumption}[Bounded Growth Rate of Maximal Dependency] \label{ass:bounded maximal dependency 2} For all $q$ and $t$, the maximal dependency sequence grows at a rate strictly less than $n$:
\begin{align*}
    \bar{M}_{nq} = o(n)
\end{align*}
\end{assumption}

\subsection{Shock Propagation Mechanism} For tractability, productivity and demand shocks  propagate in a single direction within period $t$ and die out in the following period. Network connections induce correlation across firm-level shocks. At the start of period $t$, shocks are drawn from distributions $\varepsilon_{idt} \sim \text{Bernoulli}(p_{idt})$ and $\varepsilon_{iut} \sim \text{Bernoulli}(p_{iut})$ where $p_{idt}$ and $p_{iut}$ represent time-varying propensities for firms to experience downstream (demand-side) or upstream (supply-side) shocks, respectively. Propensities are a function of the network structure and firm substitutability, both of which are fixed exogenously at the start of each period.

Intuitively, firms with more substitutability across customers (suppliers) should have a lower average propensity  $p_{idt}$ ($p_{iut}$) to experience shocks. Mathematically, I assume propensities follow a logistic (sigmoid) curve: 

\begin{align}
    p_{iqt} = g(s_{iqt}; k_{iq}, x_{iq}) = \frac{1}{1 + \exp\big\{k_{iqt} \cdot (s_{iqt} - x_{iqt})\big\}}, \quad q \in \{u,d\},
    \label{eq:propensity function}
\end{align}

where $s_{iut}$ ($s_{idt}$) is the supply-side (demand-side) substitutability of firm $i$, $k_{iqt}$ is the sensitivity (steepness) of firm propagation to substitutability, and $x_{iqt}$ is a scalar midpoint. The cross-sectional normalization ensures that firms with substitutability $x_{iqt}$ have a 50\% chance of being shocked.  Substitutability captures network-weighted dispersion in $i$'s supplier-technology (customer-product) space, while $x_{iqt}$ and $k_{iqt}$ jointly characterize the firm-specific risk of firm $i$. Inverting terms in the ``between" concentration measure from \eqref{eq:between empirical}, I assume substitutability can be written:
\begin{align}
    s_{iqt} = \log\sum_{j\neq k} w_{ijt}^q \cdot w_{ikt}^q \cdot \delta_{jkt}^q,
    \label{eq:substitutability definition model}
\end{align}
where $w_{ijt}^q$ represents the importance of trade between $j$ and $i$ in the $q$-stream direction and $\delta_{jkt}^q$ is normalized distance between industries $j$ and $k$ in latent product ($q = d$) or technological ($q=u$) space. See Appendix \ref{asec:product varieties latent space} for details. Shared customer and supplier connections induce correlation in substitutability $s_{iqt}$ across firms. This also implies the shock transmission propensities $p_{iqt}$ are also correlated. Time-variation in firm product differentiation generates correlated changes in substitutability across firms who share customers and suppliers. When there are no network connections, firms are hit by shocks with probability $p_{iqt} = 1/(1+\text{exp}({-k_{iqt}x_{iqt}}))$. For remaining sections, I assume that $k_{iqt} = k_{iq}$ and $x_{iqt} = x_{iq}$ are time invariant.

\subsection{Consumption Growth and the Stochastic Discount Factor}
I assume that representative households in this economy own shares in each firm and have the following preferences:
\begin{align}
    u(c_{1t},...,c_{nt},\ell_t) = \frac{1}{1-\gamma} \cdot \bigg(\prod_{i=1}^n c_{it}^{\beta_i}\bigg)^{1-\gamma},
\end{align}
where $c_{it}$ is the consumption of good $i$ with preference weights $\beta_i$ such that $\sum_{i}\beta_i = 1$, $\gamma$ is risk aversion, and $g(.)$ is a decreasing and differentiable function capturing disutility of labor $\ell_t$. In Appendix \ref{asec:General Equilibrium Model of Input-Output Linkages}, I show that equilibrium consumption growth and output growth are equal such that $\Delta \tilde{c}_{i,t+1} := \log (c_{i,t+1} / c_{it}) = \Delta \tilde{y}_{i,t+1}$ for all $i$ and $t$. I also derive an appropriate price normalization such that equilibrium consumption expenditure is given by $C_t = \prod_i c_{it}^{\beta_i} =  \sum_{i} p_{it}c_{it}$ for a given set of positive prices $p_{it}$. Finally, the following proposition derives the expression for growth in aggregate consumption expenditure under the same assumptions. 
\begin{proposition}[Aggregate Consumption and Output Growth]
\label{prop:aggregate consumption growth}
Assuming  $\beta_i = 1/n$ for all $i$ and
under the price normalization in Appendix \ref{asec:General Equilibrium Model of Input-Output Linkages},  aggregate consumption growth can be written: 
\begin{align}
    \Delta \tilde{c}_{t+1} & = \gamma_u \cdot a_{t+1} + \gamma_d \cdot g_{t+1} - \beta_u \cdot W_{u,t+1} - \beta_d \cdot W_{d,t+1},
    \label{eq:consumption growth}
\end{align}
where $W_{ut} = \frac{1}{n}\sum_{i=1}^n \varepsilon_{iut}$ and $W_{dt} = \frac{1}{n}\sum_{i=1}^n \varepsilon_{idt}$ and $\gamma_u, \gamma_d, \beta_u, \beta_d$ are positive scalars.
\end{proposition}
\begin{proof}
See Appendix \ref{asec:proof of aggregate consumption growth}.
\end{proof}
This proposition decomposes aggregate consumption growth into four components. The first two components capture innovations to aggregate productivity and demand growth ($a_t$ and $g_t$, respectively), both of which are positively related to output and consumption growth. On the other hand, the next two components are negatively related to output and consumption growth and capture the average impact of bad shocks to productivity originating upstream ($W_{ut}$), and the average impact of bad shocks to demand originating downstream ($W_{dt}$). Combining this result with \eqref{eq:sdf1}, the log stochastic discount factor (SDF) can be written:
\begin{align}
    m_{t+1} = \log \beta - \gamma \big( \gamma_u \cdot a_{t+1} + \gamma_d \cdot g_{t+1} - \beta_u \cdot W_{u,t+1} - \beta_d \cdot W_{d,t+1} \big)
    \label{eq:model log SDF},
\end{align}
where $\beta$ is the intertemporal discount factor and $\gamma$ is risk aversion. This implies that aggregate productivity and demand growth have a positive price of risk while average upstream and downstream propagation have a negative price of risk.

\subsection{Additional Theoretical Results}
This section summarizes some additional relevant theoretical results from the model. The following proposition states that the conditional distribution of consumption growth in this model is asymptotically normal as the number of firms grows.  
\begin{proposition}[Distribution of Consumption Growth] \label{prop:distribution of consumption growth} Under Assumption \ref{ass:bounded maximal dependency}, the sequence of consumption growth is asymptotically normal as $n\to\infty$, conditional on time $t$ for all $t$:
\begin{align}
    \Delta \tilde{c}_{t+1}  \xrightarrow[]{d} \mathcal{N}(\mu_{c,t+1}, \sigma_{c,t+1}^2),
\end{align}
where:
\begin{align*}
    \mu_{ct} & := \mathbb{E}_t[\Delta \tilde{c}_{t} ] = \frac{1}{n}\sum_{i=1}^n (p_{iut} + p_{idt}), \\
    \sigma_{ct}^2 &:= \var_t[\Delta \tilde{c}_{t}] =  \sigma_g^2 + \sigma_a^2 + \var_t(W_{ut} + W_{dt}).
\end{align*}

\end{proposition}

\begin{proof}
See Appendix \ref{asec:proof distribution of consumption growth}.
\end{proof}
\noindent Although the conditional mean of consumption growth is known in this model, there is no closed form expression for the conditional variance term. This follows from the fact that shock transmission propensities follow a logistic normal distribution (see Appendix \ref{asec:product varieties latent space}). After deriving the asymptotic distribution of consumption growth, the next corollary characterizes the probability the $W_{nqt}$ deviates from its cross-sectional mean when propensities are known.

\begin{corollary}[Concentration of Network Factors]
\label{prop:concentration of network factors}
Under Assumption \ref{ass:bounded maximal dependency 2} and if $\bar{M}_{nqt} > 1$, the propagation factor $W_{qt}$ can be written:
\begin{align}
    W_{nqt} = \mu_{nqt|t} + \epsilon_{nqt|t}
\end{align}
where $\mu_{nqt|t} = \frac{1}{n}\sum_{i}p_{iqt|t}$ and $\epsilon_{nqt|t} \sim \mathcal{N}(0,\sigma_{nqt|t}^2) $ where:
\begin{align}
    \sigma_{nqt|t}^2 \leq \frac{\bar{M}_{nqt}}{n} = o(1)
\end{align}
Moreover, for any $k>0$, the magnitude of $\epsilon_{nqt|t}$ can be upper bounded as follows:
\begin{align*}
    \text{Pr}\big(|\varepsilon_{nqt|t}| \geq 2k(\bar{M}_{qt}/n) \big) \leq \frac{1}{k^2}
\end{align*}
\end{corollary}

\begin{proof}
See Appendix \ref{asec:proof concentration of network factors}.
\end{proof}










\section{Testable Implications}
\label{sec:testable implications}
In this section, I verify the main quantitative predictions of the model using financial and macroeconomic data. According to equation \eqref{eq:model log SDF}, innovations in average supply and demand shock propagation have a negative price of risk. In addition, level changes in these components should be negatively correlated with aggregate consumption growth.

\subsection{Data and Calibration}
I construct a panel of firms between 1997-2019 whose North American Industry Classification System (NAICS) are in the set of industries for which BEA Input-Output accounts are available. I obtain annual and quarterly firm variables from Compustat and stock return data from CRSP for share codes 10, 11, and 12.\footnote{Firm and return variables are winsorized at the 1\% level unless otherwise specified.} I obtain aggregate time series of Total Factor Productivity growth from \citet{Fernald2012TFP}, government demand growth from the procurement proxy in \citet{BrigantiSellemiProcurement}, and annual market and risk-free returns from Kenneth French's Website.

I begin by computing input-output propagation factors, denoted by $\hat{W}_{ut}$ and $\hat{W}_{dt}$. In Section \ref{sec:empirical evidence}, I introduce a latent distance approach to compute the panel of industry concentration and substitutability between customers and suppliers from equation \eqref{eq:between empirical}. Assuming that substitutability is the same for firms in a given industry, I can then directly compute $\hat{s}_{iqt}$ for any firm with industry data available. The expression for $\hat{p}_{iqt} = g(\hat{s}_{iqt}; k_{iq}, x_{iq})$ follows directly from equation \eqref{eq:propensity function} conditional on scalar parameters $k_{iq}$ and $x_{iq}$. To calibrate these parameters, I first estimate the following panel regression:
\begin{align}
    \Delta \tilde{y}_{i,t+1} & = \gamma_u a_{t+1} + \gamma_d g_{t+1} + \text{controls} + \epsilon_{i,t+1},
    \label{eq: empirical sales growth panel regression}
\end{align}
where $\Delta \tilde{y}_{i,t+1}$ is year-on-year sales growth, $a_{t+1}$ is TFP growth, and $g_{t+1}$ is growth in the federal defense. Controls include year and industry fixed effects, lagged firm size, age, and return on assets to ensure that changes in $\epsilon_{i,t+1}$ is unrelated to aggregate economy-wide or industry-level forces or trends in large, young, or profitable firms.\footnote{Industry fixed effects are at the two-digit NAICS granularity.} Then let $\hat{\epsilon}_{i,t+1}$ denote residual sales growth, and let $\omega_{iu}$ ($\omega_{id}$) denote the average cost share (sales share) of intermediate inputs in $i$'s industry, and choose values of $k_{iq}\geq 0$ and $s_{iq}\in\mathbb{R}$ such that:
\begin{align*}
 \widehat{\var}(\hat{\epsilon}_{i,t+1}) = \frac{\exp(k_{iq}(\bar{s}_{iq} - x_{iq}))}{\big(1 + \exp(k_{iq}(\bar{s}_{iq} - x_{iq}))\big)^2}, \quad \text{ and } \quad \omega_{iq} = \frac{\big(1 + \exp(k_{iq}(\bar{s}_{iq} - x_{iq}))\big)^2}{\big(1 + \exp(-k_{iq} x_{iq}))\big)^2},
\end{align*}
for $q\in\{u,d\}$ where $\bar{s}_{iq} = \frac{1}{T}\sum_{t}s_{iqt}$ is firm $i$'s average substitutability over time. The first restriction is based on equation \eqref{eq:firm level cash flows} and ensures that the variance of a typical Bernoulli($p_{iqt}$) shock is equal to residual sales growth variance, while the second restriction requires $\omega_{iq}$ proportion of this variance to be attributed to network propagation. Together, the system of equations uniquely identify $k_{iq}$ and $x_{iq}$. Table \ref{tab:calibrated parameters} summarizes the calibrated parameter values. I then approximate each realized network propagation factor with its cross-sectional empirical mean as follows: 
\begin{align}
    \hat{W}_{qt} \approx \frac{1}{n}\sum_{i=1}^n \hat{p}_{iqt}, \quad q\in \{u,d\},
    \label{eq:empirical propagation factors}
\end{align}
where $\hat{p}_{iqt}$ is the empirical propensity. I use the cross-sectional mean since realized shocks cannot be identified even when firm propensities $p_{iqt}$ are known. In practice, this is not a large concern, as Proposition \ref{prop:concentration of network factors} shows that the measurement error can be bounded arbitrarily by increasing the sample size.\footnote{As a heuristic evaluation of this bound, suppose I restrict our sample to only firms that show up in the Customer Segments database ($\bar{M}_{nqt} = 149$ and $n = 12489$), then the probability that the measurement error more than 10\% is less than 1\%.} I plot the estimated series in Figure \ref{fig:propagation factors} and report summary statistics in Table \ref{tab:summary propagation factors}. See Appendix \ref{asec:additional firm results} for more details.

\begin{table}[H]
    \centering
        \caption{\textbf{Descriptive Statistics Network Propagation Factors}}
    \begin{tabular}{l|ccccccccc}
    \toprule
    \multicolumn{1}{r}{} & $W_{ut}$   & $S_{ut}$   & $W_{dt}$   & $S_{dt}$  & $g_t$    & $a_t$    & $\sigma_{t}^{civ}$  & $\sigma_{t}^{mkt}$ & AC(1) \\
    \midrule
    $W_{ut}$    & 1     & -0.81 & 0.13  & 0.17  & 0.25  & -0.20 & -0.38 & -0.31 & 0.26 \\
   $S_{ut}$   &       & 1     & -0.36 & -0.03 & -0.21 & -0.16 & 0.62  & 0.52  & 0.20 \\
    $W_{dt}$  &       &       & 1     & -0.44 & -0.05 & -0.13 & -0.22 & -0.06 & 0.01 \\
    $S_{dt}$   &       &       &       & 1     & -0.15 & -0.17 & 0.13  & -0.14 & 0.04 \\
    $g_t$    &       &       &       &       & 1     & 0.05  & 0.00  & 0.26  & 0.62 \\
    $a_t$   &       &       &       &       &       & 1     & -0.19 & -0.43 & 0.30 \\
    $\sigma_{t}^{civ}$ &       &       &       &       &       &       & 1     & 0.53  & -0.30 \\
   $\sigma_{t}^{mkt}$  &       &       &       &       &       &       &       & 1     & -0.21 \\
    \bottomrule
    \end{tabular}%
     \fnote{\scriptsize\textit{ Notes:} This table reports the time-series correlation and first-order autocorrelation of network propagation factors $W_{ut}$ and $W_{dt}$, average industry substitutability $S_{ut}$ and $S_{dt}$, procurement demand growth $g_t$, productivity growth $a_t$, innovations to common idiosyncratic volatility $\sigma_{t}^{civ}$, and innovations to market volatility $\sigma_{t}^{mkt}$. I calculate $g_t$ as the first log difference in the federal procurement proxy from \citet{BrigantiSellemiProcurement}, $a_t$ as the first difference in the TFP series from \citet{Fernald2012TFP}, innovations in common idiosyncratic volatility as the first log difference in the first principal component of firm volatility following \citet{HKLNCIV2016}, and innovations in market volatility as the first log difference in market return volatility.   }
    \label{tab:summary propagation factors}
\end{table}

\begin{figure}[H]
    \centering
    \subfloat[Panel A: Propagation Factors]{\includegraphics[width = .5\textwidth]{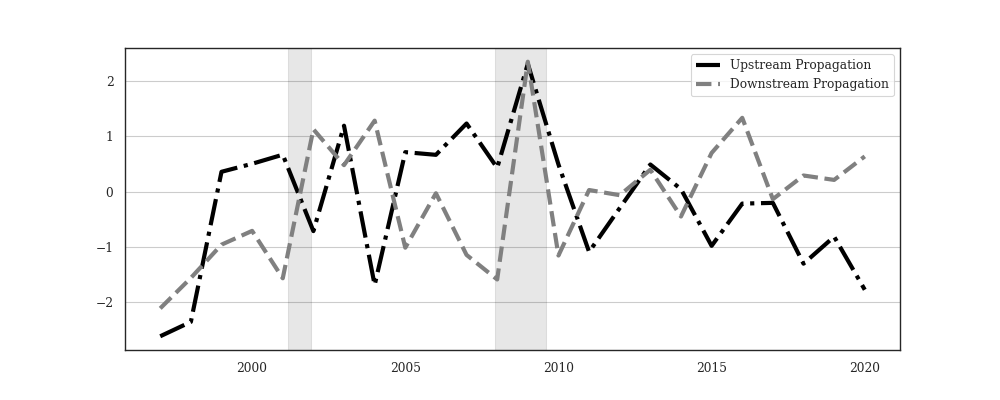}}
    \subfloat[Panel B: Innovations in Avg Substitutability]{\includegraphics[width = .5\textwidth]{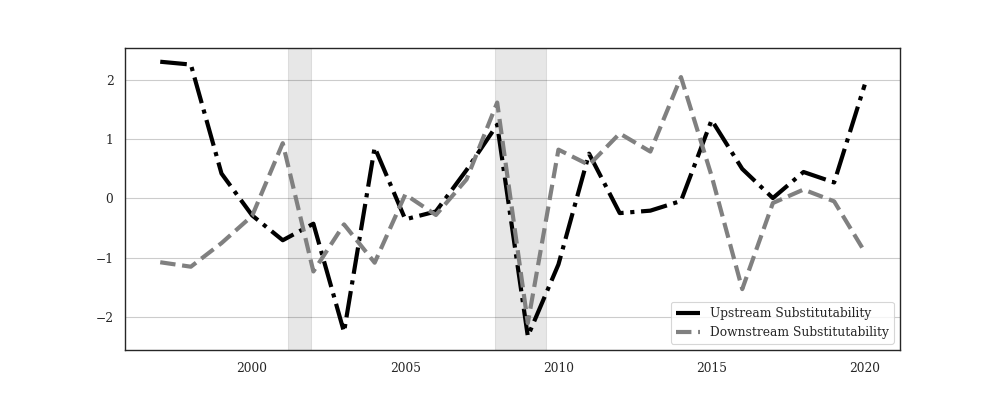}}
     \caption{\textbf{Network Propagation Risk Factors}}
    \label{fig:propagation factors}
    \fnote{\scriptsize\textit{Notes:} This figure plots the time series of network propagation risk factors (Panel A), the cross-sectional average industry substitutability (Panel B). Shaded regions indicate NBER-dated recession periods.  }
\end{figure}

\newpage

\subsection{Asset Pricing Tests}
To verify the prices of risk predicted in \eqref{eq:model log SDF}, I sort stocks based on their exposure to factors and form quintile-sorted portfolios. In particular, for every stock $i$ I regress annual excess returns $r_{it} - r_{ft}$ on a constant, aggregate demand and productivity growth, and additional controls.\footnote{I test several specifications including controlling for lag factor levels. Results are robust to several specifications on the set of controls, including the baseline without controls.} The main regression is given by:
\begin{align}
    r_{it} - r_{ft} = \alpha_i + \beta_{ia} a_t + \beta_{ig}  g_t + \beta_{iu} W_{ut} + \beta_{id} W_{dt} +  \varepsilon_{it},
    \label{eq:portfolio sorting}
\end{align}
where equation \eqref{eq:model log SDF} implies that stocks with high $\beta_{ia}$ and $\beta_{ig}$ should have higher expected excess returns and stocks with high $\beta_{iu}$ and $\beta_{id}$ should have lower expected excess returns. For each year $t$, I compute stock exposure to factors on a 15-year rolling window from $t-14$ to $t$ using \eqref{eq:portfolio sorting}, and then sort stocks into five portfolios on each beta both separately (one-way sort) and pairwise (two-way sort). Then I construct value and equal-weighted portfolios over the subsequent year $t+1$ and compute average out-of-sample excess returns for each portfolio. 

Table \ref{tab:one way sort} provides evidence of a significant return spread in one-way beta sorted portfolios. In particular, the highest quintile upstream propagation beta portfolio earns -11.42\% lower annual returns than the lowest quintile portfolio, while the highest quintile downstream propagation beta portfolio earns -4.18\% lower annual returns than the lowest quintile portfolio. Both return spreads are statistically significant, although more pronounced for upstream propagation beta sorted portfolios.\footnote{I also test for monotonicity of returns in upstream and downstream propagation betas, following \citet{timmermannMRtest}. I reject this null hypothesis at the 10\% level for upstream beta sorted portfolios, but fail to reject for downstream beta sorted portfolios.} This is consistent with \citet{herskovic_firm_2020}, who argue that upstream propagation is the more important channel. 

I also observe a return spread in post-sample alphas from the CAPM and Fama and French (FF3) three factor models, which implies that network propagation risk is not captured by market returns or FF3 factors. In light of the variance results of Section \ref{sec:empirical evidence}, I also verify that return spreads are not explained by market volatility or idiosyncratic volatility factors in Table \ref{tab:one way sort controlling for vol}.\footnote{I measure market volatility as the annual volatility of market returns and idiosyncratic volatility following \citet{HKLNCIV2016}.} Additionally, return spreads cannot be explained by differences in return volatility, average size, or average book-to-market ratios. Finally, the average correlation between upstream and downstream propagation betas is 8.6\%, suggesting that the two network factors are distinct sources of risk.

Return spreads are robust to the choice of trailing window length, equal or value weighting in portfolios, control variables, and show up in double-sorted portfolios as well. See Appendix \ref{asec:additional firm results} for more details.

\begin{table}[H]
   \centering
  \caption{\textbf{One-Way Sorted Portfolios on Network Propagation Factors}}
    \begin{tabular}{lcccccccc}
    \toprule
    \multicolumn{9}{c}{Panel A: One-way sorts on upstream propagation beta (controlling for $a_t$ and $g_t$)} \\
    \midrule
          & \multicolumn{1}{c}{1 (Low)} & \multicolumn{1}{c}{2} & \multicolumn{1}{c}{3} & \multicolumn{1}{c}{4} & \multicolumn{1}{c}{5 (High)} & \multicolumn{1}{c}{H-L} & \multicolumn{1}{c}{t(H-L)} & \multicolumn{1}{c}{MR p-val} \\
    \midrule
   $\mathbb{E}[r] - r_f$ & 18.10  & 12.81 & 10.59 & 9.42  & 6.69  & \multicolumn{1}{r}{-11.42} & \multicolumn{1}{r}{-13.22} & \multicolumn{1}{r}{0.07} \\
    $\alpha_{capm}$ & 0.29  & -0.1  & -0.23 & -0.31 & -1.02 & \multicolumn{1}{r}{-1.32} & \multicolumn{1}{r}{-15.71} & \multicolumn{1}{r}{0.05} \\
     $\alpha_{ff3}$ & 0.08  & -0.09 & -0.22 & -0.29 & -0.54 & \multicolumn{1}{r}{-0.63} & \multicolumn{1}{r}{-8.61} & \multicolumn{1}{r}{0.09} \\
    Volatility (\%) & 15.54 & 13.89 & 13.59 & 13.03 & 19.66 & -     & -     & - \\
    Book-to-market & 0.52  & 0.56  & 0.53  & 0.55  & 0.50   & -     & -     & - \\
    Market value (\$bn) & 6.46  & 16.99 & 10.62 & 15.15 & 9.11  & -     & -     & - \\
    \midrule
    \multicolumn{9}{c}{Panel B: One-way sorts on downstream propagation beta (controlling for $a_t$ and $g_t$)} \\
    \midrule
          & \multicolumn{1}{c}{1 (Low)} & \multicolumn{1}{c}{2} & \multicolumn{1}{c}{3} & \multicolumn{1}{c}{4} & \multicolumn{1}{c}{5 (High)} & \multicolumn{1}{c}{H-L} & \multicolumn{1}{c}{t(H-L)} & \multicolumn{1}{c}{MR p-val} \\
    \midrule
    $\mathbb{E}[r] - r_f$& 13.54 & 13.23 & 11.02 & 9.77  & 9.36  & \multicolumn{1}{r}{-4.18} & \multicolumn{1}{r}{-7.56} & \multicolumn{1}{r}{0.25} \\
    $\alpha_{capm}$& -0.04 & -0.18 & -0.28 & -0.38 & -0.60  & \multicolumn{1}{r}{-0.56} & \multicolumn{1}{r}{-4.78} & \multicolumn{1}{r}{0.00} \\
    $\alpha_{ff3}$ & -0.11 & -0.14 & -0.23 & -0.28 & -0.36 & \multicolumn{1}{r}{-0.25} & \multicolumn{1}{r}{-3.62} & \multicolumn{1}{r}{0.03} \\
     Volatility (\%)  & 15.44 & 13.95 & 18.58 & 12.99 & 13.88 & -     & -     & - \\
    Book-to-market & 0.52  & 0.56  & 0.55  & 0.53  & 0.51  & -     & -     & - \\
    Market value (\$bn) & 15.84 & 7.45  & 4.54  & 17.6  & 12.72 & -     & -     & - \\
    \midrule
    \multicolumn{9}{c}{Panel C: One-way sorts on upstream propagation beta (no controls)} \\
    \midrule
          & \multicolumn{1}{c}{1 (Low)} & \multicolumn{1}{c}{2} & \multicolumn{1}{c}{3} & \multicolumn{1}{c}{4} & \multicolumn{1}{c}{5 (High)} & \multicolumn{1}{c}{H-L} & \multicolumn{1}{c}{t(H-L)} & \multicolumn{1}{c}{MR p-val} \\
    \midrule
    $\mathbb{E}[r] - r_f$ & 15.15 & 12.61 & 11.46 & 9.44  & 7.23  & \multicolumn{1}{r}{-7.91} & \multicolumn{1}{r}{-11.57} & \multicolumn{1}{r}{0.07} \\
    $\alpha_{capm}$ & 0.09  & -0.15 & -0.17 & -0.28 & -1.09 & \multicolumn{1}{r}{-1.18} & \multicolumn{1}{r}{-17.53} & \multicolumn{1}{r}{0.26} \\
    $\alpha_{ff3}$ & -0.12 & -0.14 & -0.18 & -0.22 & -0.58 & \multicolumn{1}{r}{-0.46} & \multicolumn{1}{r}{-9.96} & \multicolumn{1}{r}{0.31} \\
     Volatility (\%)  & 15.26 & 14.23 & 13.57 & 12.61 & 20.97 & -     & -     & - \\
    Book-to-market & 0.54  & 0.58  & 0.52  & 0.52  & 0.50   & -     & -     & - \\
    Market value (\$bn) & 6.87  & 17.48 & 10.92 & 16.38 & 6.56  & -     & -     & - \\
    \midrule
    \multicolumn{9}{c}{Panel D: One-way sorts on downstream propagation beta (no controls)} \\
    \midrule
          & \multicolumn{1}{c}{1 (Low)} & \multicolumn{1}{c}{2} & \multicolumn{1}{c}{3} & \multicolumn{1}{c}{4} & \multicolumn{1}{c}{5 (High)} & \multicolumn{1}{c}{H-L} & \multicolumn{1}{c}{t(H-L)} & \multicolumn{1}{c}{MR p-val} \\
    \midrule
    $\mathbb{E}[r] - r_f$ & 12.66 & 11.94 & 11.8  & 8.34  & 5.13  & \multicolumn{1}{r}{-7.53} & \multicolumn{1}{r}{-8.65} & \multicolumn{1}{r}{0.42} \\
    $\alpha_{capm}$ & -0.15 & -0.18 & -0.19 & -0.4  & -0.64 & \multicolumn{1}{r}{-0.49} & \multicolumn{1}{r}{-11.37} & \multicolumn{1}{r}{0.41} \\
    $\alpha_{ff3}$  & -0.10  & -0.21 & -0.21 & -0.29 & -0.32 & \multicolumn{1}{r}{-0.22} & \multicolumn{1}{r}{-4.76} & \multicolumn{1}{r}{0.44} \\
     Volatility (\%)  & 14.09 & 13.9  & 14.44 & 13.22 & 31.95 & -     & -     & - \\
    Book-to-market & 0.54  & 0.49  & 0.58  & 0.51  & 0.54  & -     & -     & - \\
    Market value (\$bn) & 15.97 & 12.79 & 6.33  & 16.88 & 6.34  & -     & -     & - \\
    \bottomrule
     \end{tabular}%
    \fnote{\scriptsize \textit{Notes:} This table reports average excess returns and post-sample alphas in annual percentages for value-weighted portfolios sorted into quintiles on annual upstream and downstream propagation factors. Sample is between 1997-2021 for more than 10,000 stocks belonging to the BEA 66 non-government industry classifications. Panels A and B control for productivity growth and federal procurement demand growth, while Panels C and D have no controls. I also report average return volatility, book-to-market ratio and market value for each portfolio. To test for significant return spreads,  I report $t$-statistics for the null hypothesis $H_0: xr_{5} = xr_{1}$, where $xr_q$ is the average return of the $q^{th}$ quintile single sorted portfolio. Moreover, I report p-values for the test $H_0: xr_{q+1} < xr_q  \forall q \leq 4$, calculated via bootstrap following \citet{timmermannMRtest}. }
  \label{tab:one way sort}%
\end{table}%

\newpage
\subsection{Verifying Macroeconomic Predictions}

Equation \eqref{eq:consumption growth} predicts that upstream and downstream propagation factors should be negatively correlated with consumption, output growth, and aggregate dividend growth. To test this, I construct aggregate series between 1997-2021 for consumption and output growth from the National Income and Product Accounts (NIPA) and corporate dividend growth from BEA data. Then I regress each outcome on network propagation factors, controlling for aggregate productivity and federal procurement demand growth. I standardize each variable to have zero mean and unit standard deviation. Consistent with the predictions of the model, Table \ref{tab:network risk and macro factors} reports negative and statistically significant coefficients on both upstream and downstream propagation risk factors. The factors explain a large portion of time variation in consumption, output, and dividend growth with $R^2$ values of 56\%, 68\%, and 26\%, respectively. 

The coefficients on downstream propagation are -0.17 ($t = -1.89$), -0.60 ($t = -2.57$), and -0.01 ($t = -0.52$) for aggregate consumption, output, and dividend growth regressions, respectively. On the other hand, the coefficients on upstream propagation are not significant for aggregate consumption and output growth regressions, although the coefficient in the dividend growth regression is $-0.03$ ($t = -1.80$). Additionally, the coefficients on upstream propagation are significant when the dependent variable is limited to only durable consumption or output growth, -0.111 ($t = -2.39$) and -1.38 ($t=-2.39$), respectively. This suggests that durable consumption is more sensitive to upstream (supply-side) risk.

\begin{table}[H]
    \centering
    \caption{\textbf{Network Propagation and Macroeconomic Factors}}
    \begin{tabular}{lccccccc}
     \toprule
    \multicolumn{1}{c}{Variable} & $\Delta c_{t}$   &  $\Delta c_{t}^{dur}$   &  $\Delta c_{t}^{nondur}$  & $\Delta y_{t}$    & $\Delta y_{t}^{dur}$     & $\Delta y_{t}^{nondur}$    & $\Delta D_{t}$ \\
    \midrule
    $W_{ut}$    & 0.003 & -0.111** & -0.158 & -0.112 & -1.383** & -0.416 & -0.033* \\
          & (0.103) & (0.048) & (0.087) & (0.275) & (0.607) & (0.229) & (0.021) \\
    $W_{dt}$    & -0.174* & -0.031 & -0.022 & -0.598** & -0.387 & -0.058 & -0.010 \\
          & (0.074) & (0.054) & (0.101) & (0.222) & (0.684) & (0.267) & (0.021) \\
    $a_t$  & 0.321** & 0.128** & 0.197** & 1.018** & 1.607** & 0.520** & 0.018 \\
          & (0.114) & (0.055) & (0.108) & (0.271) & (0.689) & (0.285) & (0.010) \\
    $g_t$    & 0.271 & -0.056 & 0.604 & -0.011 & -0.704 & 1.592 & 0.029 \\
          & (0.332) & (0.276) & (0.638) & (0.956) & (3.450) & (1.682) & (0.109) \\\midrule
    Intercept & -0.485 & -0.111 & -0.205 & 1.218 & 4.09  & 2.023 & 0.042 \\
    Obs   & 24    & 24    & 24    & 24    & 24    & 24    & 24 \\
   $R^2$ & 0.56  & 0.537 & 0.376 & 0.679 & 0.537 & 0.376 & 0.257 \\
    \bottomrule
    \end{tabular}%
  \fnote{\scriptsize \textit{Notes:} This table reports results of OLS regressions of aggregate consumption growth, output growth, and dividend growth on input-output network propagation risk factors, controlling for productivity and federal procurement demand growth. The columns represent different dependent variables corresponding to aggregate PCE growth, durable consumption growth, non-durable consumption growth, output growth, durable output growth, non-durable output growth, and dividend growth. All series are standardized to have zero mean and unit variance. Sample is at an annual frequency between 1997-2021.}
    \label{tab:network risk and macro factors}
\end{table}%

\section{Conclusion}

In this work, I propose a production-based asset pricing model with input-output networks in which productivity shocks propagate downstream from suppliers to customers and demand shocks propagate upstream from customers to suppliers. This model is consistent with a benchmark reduced-form equation that links output growth of each node (industry or firm) in the network to output growth of other nodes and a node-specific shock. Almost all research in this area assumes that node-specific shocks are idiosyncratic- that is, drawn independently across units. In this work, I prove that the idiosyncratic shock assumption in this reduced-form equation is not consistent with realistic input-output networks. I argue that researchers who make use of network models should account for potential inter-node correlation in propagated shocks. 

When accounting for non-negligible correlation in shocks, the variance expression for output growth gains an additional component which depends on a weighted sum of covariances of shocks between each node's trade partners. As a result, units in the network are exposed to risk associated with the homogeneity of their customer and supplier connections. This generates cross-sectional differences in nodes' ability to substitute away from correlated shocks propagating through the network. I define substitutability as the negative value of the new covariance term, termed concentration ``between" suppliers.  Consistent with theory, I provide empirical evidence that customer and supplier substitutability can explain differences in realized variance across industries. In particular, higher substitutability (lower between-concentration) explains lower realized variance, conditional on other input-output and variance related characteristics. 

Although the covariance between shocks is defined in a purely statistical sense, I also investigate industry characteristics which can provide a structural explanation for high correlation in shocks. To this end, I argue that technological proximity explains correlation in supply-side shocks that propagate downstream, while product similarity explains correlation in demand-side shocks that propagate upstream. As a result, supply chain substitutability can be calculated as a function of distances between units in latent technology and product space. Although there is evidence of other sources of correlation in shocks, this simplification is simple, tractable, and can be incorporated in theoretical models to provide a structural origin for correlation in shocks. 

I incorporate this mechanism in an extension of the network model which directly models the propagation of shocks between firms in the input-output network. The propensity that shocks propagate is a decreasing function in substitutability, such that firms with more substitutable supply chains are less likely to experience shocks. Moreover, demand shock substitutability is captured by network-weighted dispersion of a firm's customers in latent product space, while supply shock substitutability is captured by network-weighted dispersion of a firm's suppliers in latent technology space. In the extended model, systematic risk is driven by the average propagation of upstream and downstream transmitted shocks. 

Under some reasonable assumptions, I calibrate the model to data on publicly traded firms and construct empirical analogues of upstream and downstream propagation risk factors. I provide evidence that these risk factors are negatively priced in the cross-section of returns. Additionally, consistent with the model's theoretical predictions, both factors are associated with lower aggregate consumption growth, lower output growth, and lower dividend growth. Empirical results suggest the importance of both directions of propagation as significant sources of systematic risk in the economy. 

Future work might investigate in more depth the sources of correlation in supply and demand shocks and generate more granular estimates of the agents' ability to substitute away from them.

\newpage

\bibliography{risk_network_economies.bib}

\begin{thebibliography}{39}
\providecommand{\natexlab}[1]{#1}
\providecommand{\url}[1]{\texttt{#1}}
\expandafter\ifx\csname urlstyle\endcsname\relax
  \providecommand{\doi}[1]{doi: #1}\else
  \providecommand{\doi}{doi: \begingroup \urlstyle{rm}\Url}\fi

\bibitem[Acemoglu et~al.(2012)Acemoglu, Carvalho, Ozdaglar, and
  Tahbaz-Salehi]{Acemoglu2012}
D.~Acemoglu, V.~M. Carvalho, A.~Ozdaglar, and A.~Tahbaz-Salehi.
\newblock The network origins of aggregate fluctuations.
\newblock \emph{Econometrica}, 80\penalty0 (5):\penalty0 1977--2016, 2012.
\newblock \doi{https://doi.org/10.3982/ECTA9623}.
\newblock URL \url{https://onlinelibrary.wiley.com/doi/abs/10.3982/ECTA9623}.

\bibitem[Acemoglu et~al.(2016)Acemoglu, Akcigit, and
  Kerr]{AcemogluAkjitKerr2016}
D.~Acemoglu, U.~Akcigit, and W.~Kerr.
\newblock Networks and the macroeconomy: An empirical exploration.
\newblock \emph{NBER Macroeconomics Annual}, 30\penalty0 (1):\penalty0
  273--335, 2016.
\newblock URL
  \url{https://EconPapers.repec.org/RePEc:ucp:macann:doi:10.1086/685961}.

\bibitem[Ahern(2013)]{Ahern2013NetworkCA}
K.~R. Ahern.
\newblock Network centrality and the cross section of stock returns.
\newblock \emph{Economics of Networks eJournal}, 2013.

\bibitem[Ang et~al.(2006)Ang, Hodrick, Xing, and Zhang]{AngVolatilityPuzzle}
A.~Ang, R.~J. Hodrick, Y.~Xing, and X.~Zhang.
\newblock The cross-section of volatility and expected returns.
\newblock \emph{The Journal of Finance}, 61\penalty0 (1):\penalty0 259--299,
  2006.
\newblock \doi{https://doi.org/10.1111/j.1540-6261.2006.00836.x}.
\newblock URL
  \url{https://onlinelibrary.wiley.com/doi/abs/10.1111/j.1540-6261.2006.00836.x}.

\bibitem[Antr{\`{a}}s et~al.(2012)Antr{\`{a}}s, Chor, Fally, and
  Hillberry]{Antrs2012}
P.~Antr{\`{a}}s, D.~Chor, T.~Fally, and R.~Hillberry.
\newblock Measuring the upstreamness of production and trade flows.
\newblock \emph{American Economic Review}, 102\penalty0 (3):\penalty0 412--416,
  May 2012.
\newblock \doi{10.1257/aer.102.3.412}.
\newblock URL \url{https://doi.org/10.1257/aer.102.3.412}.

\bibitem[Autor et~al.(2013)Autor, Dorn, and Hanson]{autordornhanson2013}
D.~H. Autor, D.~Dorn, and G.~H. Hanson.
\newblock The china syndrome: Local labor market effects of import competition
  in the united states.
\newblock \emph{American Economic Review}, 103\penalty0 (6):\penalty0 2121--68,
  October 2013.
\newblock \doi{10.1257/aer.103.6.2121}.
\newblock URL \url{https://www.aeaweb.org/articles?id=10.1257/aer.103.6.2121}.

\bibitem[Baqaee and Farhi(2019)]{baqaeeFarhiMacroimpact}
D.~R. Baqaee and E.~Farhi.
\newblock The macroeconomic impact of microeconomic shocks: Beyond hulten's
  theorem.
\newblock \emph{Econometrica}, 87\penalty0 (4):\penalty0 1155--1203, 2019.
\newblock \doi{https://doi.org/10.3982/ECTA15202}.
\newblock URL \url{https://onlinelibrary.wiley.com/doi/abs/10.3982/ECTA15202}.

\bibitem[Barrot and Sauvagnat(2016)]{BarrotSauvagnat2016}
J.-N. Barrot and J.~Sauvagnat.
\newblock { Input Specificity and the Propagation of Idiosyncratic Shocks in
  Production Networks *}.
\newblock \emph{The Quarterly Journal of Economics}, 131\penalty0 (3):\penalty0
  1543--1592, 05 2016.
\newblock ISSN 0033-5533.
\newblock \doi{10.1093/qje/qjw018}.
\newblock URL \url{https://doi.org/10.1093/qje/qjw018}.

\bibitem[Becker et~al.(2016)Becker, Gray, and Marvakov]{nbercesdata}
A.~Becker, W.~Gray, and J.~Marvakov.
\newblock Nber-ces manufacturing industry database: Technical notes.
\newblock \emph{National Bureau of Economic Research}, 2016.
\newblock URL
  \url{https://www.nber.org/research/data/nber-ces-manufacturing-industry-database}.

\bibitem[Bloom and Shankerman(2013)]{bloomtechspillovers2013ecma}
N.~Bloom and M.~Shankerman.
\newblock Identifying technology spillovers and product market rivalry.
\newblock \emph{Econometrica}, 81\penalty0 (4):\penalty0 1347--1393, 2013.
\newblock \doi{10.3982/ecta9466}.
\newblock URL \url{https://doi.org/10.3982/ecta9466}.

\bibitem[Breza et~al.(2020)Breza, Chandrasekhar, McCormick, and
  Pan]{BrezaChandrasekar2020}
E.~Breza, A.~G. Chandrasekhar, T.~H. McCormick, and M.~Pan.
\newblock Using aggregated relational data to feasibly identify network
  structure without network data.
\newblock \emph{American Economic Review}, 110\penalty0 (8):\penalty0 2454--84,
  August 2020.
\newblock \doi{10.1257/aer.20170861}.
\newblock URL \url{https://www.aeaweb.org/articles?id=10.1257/aer.20170861}.

\bibitem[Briganti and Sellemi(2022)]{BrigantiSellemiProcurement}
E.~Briganti and V.~Sellemi.
\newblock Who anticipates government spending? evidence from defense
  procurement.
\newblock \emph{Working Paper}, 2022.

\bibitem[Engle(2002)]{engeldcc}
R.~Engle.
\newblock Dynamic conditional correlation: A simple class of multivariate
  generalized autoregressive conditional heteroskedasticity models.
\newblock \emph{Journal of Business \& Economic Statistics}, 20\penalty0
  (3):\penalty0 339--350, 2002.
\newblock ISSN 07350015.
\newblock URL \url{http://www.jstor.org/stable/1392121}.

\bibitem[Fernald(2012)]{Fernald2012TFP}
J.~G. Fernald.
\newblock {A quarterly, utilization-adjusted series on total factor
  productivity}.
\newblock Working Paper Series 2012-19, Federal Reserve Bank of San Francisco,
  2012.
\newblock URL \url{https://ideas.repec.org/p/fip/fedfwp/2012-19.html}.

\bibitem[Fioruci et~al.(2013)Fioruci, Ehlers, and Filho]{Fioruci2013}
J.~A. Fioruci, R.~S. Ehlers, and M.~G.~A. Filho.
\newblock Bayesian multivariate {GARCH} models with dynamic correlations and
  asymmetric error distributions.
\newblock \emph{Journal of Applied Statistics}, 41\penalty0 (2):\penalty0
  320--331, Oct. 2013.
\newblock \doi{10.1080/02664763.2013.839635}.
\newblock URL \url{https://doi.org/10.1080/02664763.2013.839635}.

\bibitem[Gabaix(2011)]{Gabaixorigins}
X.~Gabaix.
\newblock The granular origins of aggregate fluctuations.
\newblock \emph{Econometrica}, 79:\penalty0 733--772, 2011.

\bibitem[Gofman et~al.(2020)Gofman, Segal, and Wu]{GofmanSegalWu2020}
M.~Gofman, G.~Segal, and Y.~Wu.
\newblock {Production Networks and Stock Returns: The Role of Vertical Creative
  Destruction}.
\newblock \emph{The Review of Financial Studies}, 33\penalty0 (12):\penalty0
  5856--5905, 03 2020.
\newblock ISSN 0893-9454.
\newblock \doi{10.1093/rfs/hhaa034}.
\newblock URL \url{https://doi.org/10.1093/rfs/hhaa034}.

\bibitem[Gomes et~al.(2009)Gomes, Kogan, and Yogo]{Gomes2009durability}
J.~F. Gomes, L.~Kogan, and M.~Yogo.
\newblock Durability of output and expected stock returns.
\newblock \emph{Journal of Political Economy}, 117\penalty0 (5):\penalty0
  941--986, Oct. 2009.
\newblock \doi{10.1086/648882}.
\newblock URL \url{https://doi.org/10.1086/648882}.

\bibitem[Herskovic(2018{\natexlab{a}})]{Herskovic2018JoF}
B.~Herskovic.
\newblock Networks in production: Asset pricing implications.
\newblock \emph{The Journal of Finance}, 73\penalty0 (4):\penalty0 1785--1818,
  2018{\natexlab{a}}.
\newblock \doi{https://doi.org/10.1111/jofi.12684}.
\newblock URL \url{https://onlinelibrary.wiley.com/doi/abs/10.1111/jofi.12684}.

\bibitem[Herskovic(2018{\natexlab{b}})]{herskovic_networks_2018}
B.~Herskovic.
\newblock Networks in {Production}: {Asset} {Pricing} {Implications}:
  {Networks} in {Production}: {Asset} {Pricing} {Implications}.
\newblock \emph{The Journal of Finance}, 73\penalty0 (4):\penalty0 1785--1818,
  Aug. 2018{\natexlab{b}}.
\newblock ISSN 00221082.
\newblock \doi{10.1111/jofi.12684}.
\newblock URL \url{http://doi.wiley.com/10.1111/jofi.12684}.

\bibitem[Herskovic et~al.(2016)Herskovic, Kelly, Lustig, and {Van
  Nieuwerburgh}]{HKLNCIV2016}
B.~Herskovic, B.~Kelly, H.~Lustig, and S.~{Van Nieuwerburgh}.
\newblock The common factor in idiosyncratic volatility: Quantitative asset
  pricing implications.
\newblock \emph{Journal of Financial Economics}, 119\penalty0 (2):\penalty0
  249--283, 2016.
\newblock ISSN 0304-405X.
\newblock \doi{https://doi.org/10.1016/j.jfineco.2015.09.010}.
\newblock URL
  \url{https://www.sciencedirect.com/science/article/pii/S0304405X15001774}.

\bibitem[Herskovic et~al.(2020{\natexlab{a}})Herskovic, Kelly, Lustig, and
  Van~Nieuwerburgh]{Herskovic2020JPE}
B.~Herskovic, B.~Kelly, H.~Lustig, and S.~Van~Nieuwerburgh.
\newblock Firm volatility in granular networks.
\newblock \emph{Journal of Political Economy}, 0\penalty0 (0):\penalty0
  000--000, 2020{\natexlab{a}}.
\newblock \doi{10.1086/710345}.
\newblock URL \url{https://doi.org/10.1086/710345}.

\bibitem[Herskovic et~al.(2020{\natexlab{b}})Herskovic, Kelly, Lustig, and
  Van~Nieuwerburgh]{herskovic_firm_2020}
B.~Herskovic, B.~Kelly, H.~Lustig, and S.~Van~Nieuwerburgh.
\newblock Firm {Volatility} in {Granular} {Networks}.
\newblock \emph{Journal of Political Economy}, 128\penalty0 (11):\penalty0
  4097--4162, Nov. 2020{\natexlab{b}}.
\newblock ISSN 0022-3808, 1537-534X.
\newblock \doi{10.1086/710345}.
\newblock URL \url{https://www.journals.uchicago.edu/doi/10.1086/710345}.

\bibitem[Hoberg and Phillips(2016)]{Hoberg2016textindustries}
G.~Hoberg and G.~Phillips.
\newblock Text-based network industries and endogenous product differentiation.
\newblock \emph{Journal of Political Economy}, 124\penalty0 (5):\penalty0
  1423--1465, 2016.
\newblock \doi{10.1086/688176}.
\newblock URL \url{https://doi.org/10.1086/688176}.

\bibitem[Hornik et~al.(1989)Hornik, Stinchcombe, and White]{WhiteUniv89}
K.~Hornik, M.~Stinchcombe, and H.~White.
\newblock Multilayer feedforward networks are universal approximators.
\newblock \emph{Neural Netw.}, 2\penalty0 (5):\penalty0 359–366, July 1989.
\newblock ISSN 0893-6080.

\bibitem[Hottman et~al.(2016)Hottman, Redding, and Weinstein]{Hottman2016}
C.~J. Hottman, S.~J. Redding, and D.~E. Weinstein.
\newblock Quantifying the sources of firm heterogeneity.
\newblock \emph{The Quarterly Journal of Economics}, 131\penalty0 (3):\penalty0
  1291--1364, Mar. 2016.
\newblock \doi{10.1093/qje/qjw012}.
\newblock URL \url{https://doi.org/10.1093/qje/qjw012}.

\bibitem[Janson(1988)]{Janson1988}
S.~Janson.
\newblock {Normal Convergence by Higher Semiinvariants with Applications to
  Sums of Dependent Random Variables and Random Graphs}.
\newblock \emph{The Annals of Probability}, 16\penalty0 (1):\penalty0 305 --
  312, 1988.
\newblock \doi{10.1214/aop/1176991903}.
\newblock URL \url{https://doi.org/10.1214/aop/1176991903}.

\bibitem[Kramarz et~al.(2020)Kramarz, Martin, and
  Mejean]{kramarz_volatility_2020}
F.~Kramarz, J.~Martin, and I.~Mejean.
\newblock Volatility in the small and in the large: {The} lack of
  diversification in international trade.
\newblock \emph{Journal of International Economics}, 122:\penalty0 103276, Jan.
  2020.
\newblock ISSN 00221996.
\newblock \doi{10.1016/j.jinteco.2019.103276}.
\newblock URL
  \url{https://linkinghub.elsevier.com/retrieve/pii/S0022199618301296}.

\bibitem[Kruttli et~al.(2019)Kruttli, , Tran, and
  and]{Kruttli2019pricingposeidon}
M.~S. Kruttli, , B.~R. Tran, and S.~W.~W. and.
\newblock Pricing poseidon: Extreme weather uncertainty and firm return
  dynamics.
\newblock \emph{Finance and Economics Discussion Series}, 2019\penalty0 (054),
  July 2019.
\newblock \doi{10.17016/feds.2019.054}.
\newblock URL \url{https://doi.org/10.17016/feds.2019.054}.

\bibitem[McCormick and Zheng(2015)]{McCormickZhengLatentARD}
T.~H. McCormick and T.~Zheng.
\newblock Latent surface models for networks using aggregated relational data.
\newblock \emph{Journal of the American Statistical Association}, 110\penalty0
  (512):\penalty0 1684--1695, 2015.
\newblock \doi{10.1080/01621459.2014.991395}.
\newblock URL \url{https://doi.org/10.1080/01621459.2014.991395}.

\bibitem[Mian and Sufi(2014)]{Mian2014}
A.~Mian and A.~Sufi.
\newblock What explains the 2007-2009 drop in employment?
\newblock \emph{Econometrica}, 82\penalty0 (6):\penalty0 2197--2223, Nov. 2014.
\newblock \doi{10.3982/ecta10451}.
\newblock URL \url{https://doi.org/10.3982/ecta10451}.

\bibitem[Oberfield(2018)]{oberfield_theory_2018}
E.~Oberfield.
\newblock A {Theory} of {Input}-{Output} {Architecture}.
\newblock \emph{Econometrica}, 86\penalty0 (2):\penalty0 559--589, 2018.
\newblock ISSN 0012-9682.
\newblock \doi{10.3982/ECTA10731}.
\newblock URL \url{https://www.econometricsociety.org/doi/10.3982/ECTA10731}.

\bibitem[Ozdagli and Weber(2017)]{OzdagliWeber2017}
A.~Ozdagli and M.~Weber.
\newblock Monetary policy through production networks: Evidence from the stock
  market.
\newblock Working Paper 23424, National Bureau of Economic Research, May 2017.
\newblock URL \url{http://www.nber.org/papers/w23424}.

\bibitem[Patton and Timmermann(2010)]{timmermannMRtest}
A.~Patton and A.~Timmermann.
\newblock Monotonicity in asset returns: New tests with applications to the
  term structure, the capm, and portfolio sorts.
\newblock \emph{Journal of Financial Economics}, 98\penalty0 (3):\penalty0
  605--625, 2010.
\newblock URL
  \url{https://EconPapers.repec.org/RePEc:eee:jfinec:v:98:y:2010:i:3:p:605-625}.

\bibitem[Ram\`{i}rez(2017)]{Ramirez2017}
C.~Ram\`{i}rez.
\newblock Inter-firm relationships and asset prices.
\newblock \emph{Finance and Economics Discussion Series 2017-014. Washington:
  Board of Governors of the Federal Reserve System}, 2017.

\bibitem[Shea(2002)]{johnshea2002}
J.~Shea.
\newblock Complementarities and comovements.
\newblock \emph{Journal of Money, Credit and Banking}, 34\penalty0
  (2):\penalty0 412--433, 2002.
\newblock ISSN 00222879, 15384616.
\newblock URL \url{http://www.jstor.org/stable/3270695}.

\bibitem[Syverson(2004)]{Syverson2004sub}
C.~Syverson.
\newblock Product substitutability and productivity dispersion.
\newblock \emph{Review of Economics and Statistics}, 86\penalty0 (2):\penalty0
  534--550, May 2004.
\newblock \doi{10.1162/003465304323031094}.
\newblock URL \url{https://doi.org/10.1162/003465304323031094}.

\bibitem[Taschereau-Dumouchel(2020)]{taschereau-dumouchel_cascades_nodate}
M.~Taschereau-Dumouchel.
\newblock Cascades and {Fluctuations} in an {Economy} with an {Endogenous}
  {Production} {Network}.
\newblock \emph{Review of Economic Studies}, page~62, 2020.

\bibitem[Tuzel and Zhang(2017)]{tuzel2017local}
S.~Tuzel and M.~B. Zhang.
\newblock Local risk, local factors, and asset prices.
\newblock \emph{The Journal of Finance}, 72\penalty0 (1):\penalty0 325--370,
  2017.

\end{thebibliography}

\newpage
\begin{appendices}

\section{General Equilibrium Model of Input-Output Linkages}
 \label{asec:General Equilibrium Model of Input-Output Linkages}
 In this section, I show that \eqref{eq:reduced form static propagation} can be cast as an outcome of a production-based asset pricing model. This model provides a structural foundation for the theoretical contributions of this work and is closely related to \citet{Acemoglu2012}, \citet{AcemogluAkjitKerr2016}, \citet{Ramirez2017}, and \citet{herskovic_networks_2018}. Consider a competitive economy with $n$ production units (firms or industries) with Cobb-Douglas technology, representative households with constant relative risk aversion (CRRA) preferences over a basket of goods and who work and own shares in all firms and live off wages and dividends, and a government that finances purchases with a lump-sum tax. In this economy, Hicks-neutral productivity shocks propagate downstream from suppliers to customers, while government demand  shocks propagate upstream from customers to suppliers.

 \subsection{Production} Production unit $i$'s output is a constant returns to scale function of labor and intermediate inputs: 
 \begin{align}
     y_{it} = \exp(z_{it}) \ell_{it}^{\alpha_{i\ell}} \prod_{j=1}^n x_{ijt}^{w_{ijt}},
     \label{eq:firm production}
 \end{align}
 where $x_{ijt}$ is the amount of product $j$ used as input by industry $i$ at time $t$, $\ell_{it}$ is labor input, and $z_{it}$ is a Hicks-neural productivity shock, respectively. I assume that for all $i$ and $t$, the labor share of production is positive (i.e., $\alpha_{i\ell} > 0$) and intermediate input shares are non-negative ($w_{ijt} \geq 0$) and sum to the capital share of production (i.e., ($\sum_{j=1}^n w_{ijt} = 1 - \alpha_{i\ell}$). 
 
 Taking both spot market prices and input shares as given, production units optimize dividends (denoted $D_{it}$) as a function of input and labor purchases: 
 \begin{align}
     D_{it} = \max_{\{x_{ijt}\}_{j=1}^n, \ell_{it}} p_{it} y_{it} - \sum_{j=1}^n p_{jt} x_{ijt} - p_{wt}\ell_{it}
     \label{eq:dividends}
 \end{align}
 subject to \eqref{eq:firm production} and $\ell_{it} \in (0,1)$. Suppose further that $M_{t+1}$ is the stochastic discount factor (SDF) that prices all assets in the economy. Then the cum dividend value of firm $i$ (denoted $V_{it}$) satisfies the following Bellman equation: 
 \begin{align}
     V_{it} = D_{it} + \mathbb{E}_t[M_{t+1}V_{i,t+1}].
 \end{align}
 
 \subsection{Government}
The government purchases goods $G_{it}$ from each unit $i$ at time $t$ and finances them via a lump-sum tax $T_t$. Taking  prices as given, the government's financing constraint implies that $T_t = \sum_{i=1}^n p_{it} G_{it}$.

\subsection{Households}
Assume that the representative household owns shares in each unit and has the following preferences:
\begin{align}
    u(c_{1t},...,c_{nt}, \ell) = \frac{1}{1-\gamma} \cdot \bigg(\prod_{i=1}^n c_{it}^{\beta_i} \bigg)^{1-\gamma} \cdot g(\ell_t),
\end{align}
where $c_{it}$ is the consumption of good $i$ with preference weights $\beta_i$ such that $\sum_{i}\beta_i$ = 1 and $g(.) = (1-\ell_t)^\nu$ is a decreasing and differentiable function capturing disutility of labor. Households also have a time-discount factor of $\beta$ and cannot store goods from one period to another. In equilibrium, households hold a zero net position in a risk-free asset and choose to own $\vartheta_{it}$ in each unit according to the following budget constraint:
\begin{align}
    T_t +p_{wt}\ell_t + \sum_{i=1}^n p_{it} c_{it} + \sum_{i=1}^n \vartheta_{i,t+1} (V_{it} - D_{it}) = \sum_{i=1}^n \vartheta_{i,t} V_{it},
    \label{eq:HH budget constraint}
\end{align}
where the right hand side is total value of investments and the left hand side is the sum of taxes paid, wages earned, cost of consumption, and unrealized capital gains, respectively. The household's optimization problem satisfies the Bellman equation:
\begin{align}
U_t = \max_{\{c_{it}, \vartheta_{i,t+1}, \ell\}_{i=1}^n} u(.) + \beta \mathbb{E}_t[U_{t+1}],
\end{align}
subject to \eqref{eq:HH budget constraint}.
\subsection{Equilibrium}

The competitive equilibrium of the economy consists of spot market prices $\{p_{it}\}_{i=1}^n$, consumption bundles  $\{c_{it}\}_{i=1}^n$, share holdings  $\{\vartheta_{it}\}_{i=1}^n$, labor supply $\ell_t$, wages $p_{wt}$, and input bundles $\{x_{ijt}\}_{i,j=1}^n$ such that both production units and households exhibit optimal behavior and good/asset markets clear. 
\subsubsection{Market Clearing}
  In equilibrium, all good markets clear such that:
 \begin{align*}
         y_{it} = \underbrace{c_{it}}_{\text{final consumption demand}} + \underbrace{\sum_{j=1}^n x_{jit}}_{\text{intermediate demand}} +  \underbrace{G_{it}}_{\text{government consumption}},
 \end{align*}
 and all asset markets clear $\vartheta_{it} = 1$ for all $i$ and $t$. 

\subsubsection{Producer Optimality} Taking prices as given, unit $i$'s first order dividend maximizing conditions satisfy 
\begin{align}
   w_{ijt} = \frac{p_{jt} x_{ijt}}{p_{it}y_{it}} \equiv \frac{sales_{j\to i}}{sales_i}
   \label{eq:foc production}
\end{align}
and 
\begin{align}
    \alpha_{i\ell} = \frac{p_{wt}\ell_{it}}{p_{it}y_{it}}
    \label{eq:foc labor}
\end{align}

\subsubsection{Consumer Optimality} Given the Cobb-Douglas aggregation in preferences over goods (i.e., $C_t := \prod_i c_{it}^{\beta_i}$), utility maximizing households consume $\beta_i$ of income on good $i$ and hold shares fixed at $\vartheta_{it} = 1$. More specifically, letting $\lambda_t$ be the Lagrange multiplier for the period $t$ household budget constraint, the first-order condition for consumption is written:
\begin{align}
    \lambda_t = \frac{C_t^{-\gamma}}{p_{it}} \cdot \frac{\partial C_t}{\partial c_{it}}. \label{eq:lagrange multiplier}
\end{align}
This implies that equilibrium consumption satisfies:
\begin{align}
    p_{it}c_{it} = \beta_i \bigg(p_{wt}\ell_t^* + \sum_{j=1}^n D_{jt} - T_t\bigg).
    \label{eq:foc consumption}
\end{align}
where $\ell_t^*$ solves:
\begin{align}
\frac{p_{wt}\ell_t^*}{p_{wt}\ell_t^* + \sum_{j=1}^n D_{jt} - T_t} = -\frac{\ell_t^* g'(\ell_t^*)}{g(\ell_t^*)}
\label{eq:foc labor hh}
\end{align}
\subsubsection{Asset Prices} From \eqref{eq:lagrange multiplier}, the stochastic discount factor can be written:
\begin{align}
    M_{t+1} = \beta \frac{\lambda_{t+1}}{\lambda_t} = \beta \bigg(\frac{C_{t+1}}{C_t}\bigg)^{-\gamma} \frac{p_{it} \cdot \partial C_{t+1} / \partial c_{i,t+1}}{p_{i,t+1}\cdot \partial C_{t} / \partial c_{it}}.
    \label{eq:sdf1}
\end{align}
Following \citet{Herskovic2018JoF}, I assume that prices are normalized such that $p_{it} = \partial C_t / \partial c_{it}$, or equivalently that $\prod_j p_{jt}^{\beta_j} = \prod_j \beta_j^{\beta_j}$ for all $i$ and $t$. This implies the the utility aggregator is equal to the household's consumption expenditure $C_t = \sum_{i=1}^n p_{it}c_{it}$.\footnote{I further assume that $C_t = p_{wt} \ell_t$.}. Then \eqref{eq:sdf1} simplifies to:
\begin{align}
    M_{t+1} = \beta \bigg( \frac{\sum_{i=1}^n p_{i,t+1} c_{i,t+1}}{\sum_{i=1}^n p_{it} c_{it}} \bigg)^{-\gamma}.
\end{align}

\subsubsection{Shock Propagation}
I now derive closed form expressions for the effects of productivity and government demand shocks on output growth in this model. The main takeaway is that output growth is captured by the following reduced form expression:
\begin{align*}
    \mathbf{d}\log \mathbf{y}_t = \mathbf{H}_{down,t} \mathbf{dz}_t + \mathbf{H}_{up,t} \mathbf{dG}_t,
\end{align*}
where $\mathbf{H}_{.,t}$ are $n\times n$ Leontief inverse propagation matrices. I provide a derivation for each component separately. 

\paragraph{Productivity Shocks} Totally differentiate the expression in \eqref{eq:firm production} to obtain:
\begin{align}
    d\log y_{it} = d z_{it} + \alpha_{i\ell} d\log \ell_{it} + \sum_{j=1}^n w_{ijt} d\log x_{ijt}
\end{align}
Totally differentiating \eqref{eq:foc production}, \eqref{eq:foc labor}, and \eqref{eq:foc consumption} and plugging in to this expression yields:
\begin{align*}
     d\log y_{it} & = d z_{it} +  \alpha_{i\ell} d\log \ell_{it} + \sum_{j=1}^n w_{ijt} (d\log y_{it} + d\log p_{it} - d\log p_{jt}) \\
     & =  d z_{it} +  \alpha_{i\ell} (d\log y_{it} - d\log c_{it}) + \sum_{j=1}^n w_{ijt} (d\log y_{it} - d\log c_{it} + d\log c_{jt}).
\end{align*}
Given constant returns to scale ($\alpha_{i\ell} + \sum_j w_{ij}=1$), this expression can be further simplified as follows:
\begin{align*}
    d\log c_i = dz_{it} + \sum_{j=1}^n w_{ijt} d\log c_j
\end{align*}
or in vector notation:
\begin{align*}
    \mathbf{d}\log \mathbf{c}_t = \mathbf{dz}_t + \mathbf{W}_t  \mathbf{d}\log \mathbf{c}_t,
\end{align*}
where $\mathbf{W}_t$ has entries $w_{ijt}$. Note that market clearing and profit maximization conditions together imply that: 
\begin{align*}
    \frac{y_{jt}}{c_{jt}} = 1 + \sum_{i=1}^n w_{ijt} \frac{\beta_i y_{it}}{\beta_j c_{it}}
\end{align*}
which implies that equilibrium consumption growth is equal to equilibrium output growth:
\begin{align}
    \mathbf{d}\log \mathbf{c}_t = \mathbf{d}\log \mathbf{y}_t
\end{align}
and thus that:
\begin{align}
    \mathbf{d}\log \mathbf{y}_t = (\mathbf{I} - \mathbf{W}_t)^{-1} \mathbf{dz}_t,
\end{align}
where $\mathbf{H}_{down,t} :=  (\mathbf{I} - \mathbf{W}_t)^{-1}$ is the Leontief inverse of $\mathbf{W}_t$. Here, the $\mathbf{W}_t$ matrix determines the strength of downstream propagation of productivity shocks. 

\paragraph{Demand Shocks} To study the effects of government spending shocks in the model, normalize $\mathbf{z}_t = 0$ and consider the unit cost function for $i$:
\begin{align*}
    C_{it}(\mathbf{p}_t,p_{wt}) = A_{it} p_{wt}^{\alpha_{i\ell}} \prod_{j=1}^n p_{jt}^{w_{ijt}}, \quad \text{where } A_{it} = \alpha_{i\ell}^{-\alpha_{i\ell}} \prod_{j=1}^n w_{ijt}^{-w_{ijt}}.
\end{align*}
The zero productivity normalization implies zero dividends for production units, and combined with the price normalization for wages, this implies that:
\begin{align*}
\log p_{it} = \log A_{it} + \sum_{j=1}^n w_{ijt} \log p_{jt} 
\end{align*}
Conditional on productivity vector $\mathbf{z}_t$ and defining the vector $\mathbf{a}_t$ with entries $\log A_{it}$, prices are a function of the network and cost but not government purchases: 
\begin{align*}
    \log \mathbf{p}_t = (\mathbf{I} - \mathbf{W}_t)^{-1} \mathbf{a}_t.
\end{align*}
Setting $\nu = 1$ Then \eqref{eq:foc labor hh} and the fact that $T_t = \sum_{i} p_{it} G_{it}$ implies that: 
\begin{align}
    \ell_t = \frac{1}{2} + \frac{1}{2} \sum_{i=1}^n p_{it}G_{it}
\end{align}
and thus:
\begin{align*}
    p_{it} c_{it} = \beta_i [p_{wt}\ell_t - T_t] = \frac{\beta_i}{2}\bigg(1 - \sum_{j=1}^n p_{jt}G_{jt}\bigg).
\end{align*}
Differentiating and combining with the resource constraint and profit maximization conditions yields: 
\begin{align*}
    \frac{d(p_{it} y_{it})}{p_{it}y_{it}} = \sum_{j=1}^n w_{jit} \frac{d(p_{jt}y_{jt}}{p_{it}y_{it}} + \frac{dG_{it}}{y_{it}} - \frac{\beta_i}{2} \sum_{j=1}^n \frac{d(p_{jt} G_{jt})}{p_{it} y_{it}}.
\end{align*}
Since prices are constant (i.e., $d(p_{it}y_{it}) / p_{it}y_{it} = d\log y_{it}$), I can write in vector notation:
\begin{align}
    \mathbf{d} \log \mathbf{y}_t = \mathbf{H}_{up,t} \mathbf{dG_t}, 
\end{align}
where $\mathbf{H}_{up,t} = (\mathbf{I} - \mathbf{W}_t^\top)^{-1} \mathbf{\Lambda}_t$ is the upstream propagation Leontief inverse and $\mathbf{\Lambda}_t$ is a scaling matrix with diagonal entries $(1-\beta_i/2)/p_{it}y_{it}$ and off-diagonal entries $-(\beta_{i}/2)/p_{it}y_{it}$
for row indices $i$.


\newpage

\section{Product Varieties in Latent Space}
\label{asec:product varieties latent space}

Suppose that each industry (or firm) is associated with a random position $z_{it} := (\cos \theta_{it}, \sin \theta_{it})$ on a circular surface on the $3$-dimensional hypersphere $\mathcal{S}^{p+1}$. Suppose that the surface represents the space of varieties in production technology space. I also assume each unit corresponds to a position in latent product variety space, and that product positions are independent of technological positions. The stochastic process for positions is the same in both spaces and depends on changes in the angle $\theta_{it}$ as follows:
\begin{align}
    \theta_{it} = \rho \cdot \theta_{i,t-1} + \epsilon_{it}, \qquad \epsilon_{it} \sim_{iid} \mathcal{N}(0,\sigma_\theta^2),
    \label{eq:angular evolution latent circle}
\end{align}
where $\theta_{it}$ is measured in radians. The distance between two points $i$ and $j$ can then be written: 
\begin{align}
    \delta_{ijt} = \frac{1}{2\pi}|\theta_{it} - \theta_{jt}|.
\end{align}
This implies a correlation structure in the distances between units as follows: 
\begin{align}
    \text{corr}(\delta_{ijt}, \delta_{kmt}) = \begin{cases} 0 & \text{if } i\notin \{k,m\} \cap j\notin \{k,m\},\\
     1 & \text{if } i \in \{k,m\} \cup j \in \{k,m\},
    \end{cases}
    \label{eq:correlation distances},
\end{align}
or equivalently that:
\begin{align}
    \text{cov}(\delta_{ijt}, \delta_{kmt}) = \begin{cases} 0 & \text{if } i\notin \{k,m\} \cap j\notin \{k,m\},\\
     \sqrt{\var(\delta_{ijt})\var(\delta_{kmt})} & \text{if } i \in \{k,m\} \cup j \in \{k,m\},
    \end{cases}
    \label{eq:covariance distances},
\end{align}
where:
\begin{align*}
    \sigma_d^2 := \var(\delta_{ijt}) & = \var(\theta_{it}) + \var(\theta_{jt}) - \mathbb{E}[|\theta_{it} - \theta_{jt}|]^2
= \frac{2\sigma_\theta^2}{1-\varphi^2} - \frac{(4/\pi) \sigma_\theta^2}{1-\varphi^2},
\end{align*}

Define the set of $i$'s $q$-stream located trade partners by $A_{iqt} := \{k: A_{ki}(\mathcal{G}_{qt}) \neq \emptyset \}$. Note that in technology space, $q$ refers to upstream propagation. Combining with equation \eqref{eq:substitutability definition model} yields:
\begin{align*}
    \cov(s_{iqt},s_{jqt}) & = \sum_{k\neq \ell} \sum_{m \neq p} w_{ikt} w_{i\ell t} w_{jmt} w_{jpt} \cdot \cov(\delta_{k\ell t}, \delta_{mpt}) \\
    & = \sum_{k,\ell \in A_{iqt}; m,p \in A_{jqt} } w_{ikt} w_{i\ell t} w_{jmt} w_{jpt} \cdot \sigma_d^2 \cdot \mathbbm{1}\{(m \in A_{iqt}) \cup (p \in A_{iqt})\}
\end{align*}
Additionally, $s_{iqt}$ are jointly distributed as a folded truncated normal with variance-covariance matrix $\mathbf{\Sigma}_{qt} = [\cov(s_{iqt},s_{jqt})]_{ij}$ and mean vector $\boldsymbol{\mu}_{qt} = [\mu_{iqt}]_i$ with entries:
\begin{align*}
    \mu_{iqt} := \mathbb{E}[s_{iqt}] & = - \sigma_d \sqrt{8/\pi} \cdot \sum_{j<k; j,k \in A_{iqt}} w_{ijt} w_{ikt}.
\end{align*}
When $i = j$, I can further simplify as follows:
\begin{align*}
    \var(s_{iqt}) & = 4 \sigma_d^2 \cdot \sum_{j<k; j,k \in A_{iqt}} (w_{ijt} w_{ikt})^2.
\end{align*}
Next, using equation \eqref{eq:propensity function}, I can write the log-odds function as follows:
\begin{align*}
    l_{iqt} := \log \bigg(\frac{p_{iqt}}{1-p_{iqt}}\bigg) & = k\bigg(s_{iqt} - \frac{1}{n} \sum_{i=1}^n s_{iqt}\bigg) = k\bigg(\hat{e}_{in}^\top - \frac{1}{n}\iota_n^\top\bigg)\mathbf{s}_{qt},
\end{align*}
where $\hat{e}_{in}$ is the $i$th column of an $n\times n$ identity matrix  and $\iota_n$ is an $n\times 1$ vector of ones. Equivalently, the vector $\mathbf{l}_{qt} := (l_{1qt},...,l_{nqt})^\top$ can be written:
\begin{align*}
    \mathbf{l}_{qt} & = k\bigg(\mathbf{I}_n - \frac{\iota_n\iota_n^\top}{n}\bigg)\mathbf{s}_{qt} \sim \mathcal{N}\big(\mathbf{B}_{k,n} \boldsymbol{\mu}_{qt}, \mathbf{B}_{k,n}^\top\mathbf{\Sigma}_{qt} \mathbf{B}_{k,n}\big),
\end{align*}
where $\mathbf{B}_{k,n}:= k(\mathbf{I}_n - \iota_n\iota_n^\top/n)$. Notice that $p_{iqt} = F(l_{iqt})$ where $F(x) = (1 + \exp(-x))^{-1}$ and thus the vector $\mathbf{p}_{qt}$ has a logistic normal distribution and thus no closed form representation for its mean vector and variance-covariance matrix. However, I can write $
    \text{median}(\mathbf{p}_{qt}) = F(\mathbf{B}_{k,n} \boldsymbol{\mu}_{qt})
$.


\newpage
\section{Proofs}

\subsection{Proof of Proposition \ref{prop:necessary restrictions on W}}
\label{asec:necessary restrictions on W}
I begin by characterizing the family of matrices $\mathbf{W} \in M_n$ that are consistent with Assumptions \ref{ass:stable weighting matrix} and \ref{ass:idiosyncratic static shocks} with a sequence of if and only if relationships. First define $\mathbf{\Sigma}_y := \mathbb{V}[\mathbf{y}]$ to be the variance-covariance matrix of sectoral production, and notice that Assumption \ref{ass:idiosyncratic static shocks} implies that $\mathbf{\Sigma}_y = (\mathbf{I} - \mathbf{W})^{-1}\mathbf{D} (\mathbf{I} - \mathbf{W}^\top)^{-1}$ is non-negative. This is because $\mathbf{\Sigma}_y$ is the product of non-negative matrices by definition since $ (\mathbf{I} - \mathbf{W})^{-1} = \mathbf{I} + \mathbf{W} + \mathbf{W}^2 + \mathbf{W}^3 +... \geq 0$ and $\mathbf{D} \geq 0$. The former decomposition is only possible when $\rho(\mathbf{W}) \leq 1$. Rearranging terms, I obtain: 
\begin{align*}
     \mathbf{D} & = \mathbb{V}[(\mathbf{I} - \mathbf{W})\mathbf{y}]  = (\mathbf{I} - \mathbf{W}) \mathbf{\Sigma}_y (\mathbf{I} - \mathbf{W}^\top).
\end{align*}
For ease of notation, let $h_{ij} := \mathbbm{1}_{\{i=j\}} - w_{ij}$ and $\rho_{ij}\sigma_i\sigma_j$ to be the $(i,j)$ entries of $\mathbf{H} := \mathbf{I} - \mathbf{W}$ and $\mathbf{\Sigma}_y$, respectively. Then the matrix $\mathbf{Q}_n := \mathbf{H\Sigma H}^\top = [q_{ij}]$ is diagonal with entries:
\begin{align*}
    q_{ij} = \sum_{k=1}^n \sum_{m=1}^n h_{ik} h_{jm} \rho_{km} \sigma_k \sigma_m
\end{align*}
Notice that $\mathbf{Q}_n$ is symmetric and since it is also diagonal, then the following must hold: 
\begin{align*}
    \text{det} \big( \mathbf{Q}_n\big) = \prod_{i=1}^n \bigg( \sum_{k=1}^n \sum_{m=1}^n h_{ik} h_{im} \rho_{km} \sigma_k \sigma_m \bigg).
\end{align*}
Additionally, $\mathbf{P}\mathbf{Q}_n\mathbf{P}^\top$ is also diagonal for any permutation matrix $\mathbf{P} \in M_n$ so I can exchange the ordering of units without loss of generality. Consider the Laplace expansion of $\mathbf{Q}_n$ by minors along row $i$. Set $i=1$ arbitrarily and notice that:
\begin{align*}
    \text{det}\big(\mathbf{Q}_n\big) = q_{11} \cdot \text{det}\big(\mathbf{Q}_n[1,1]\big)  + \sum_{k=2}^n (-1)^{1+k}q_{1k} \cdot \text{det}\big(\mathbf{Q}_{n}[1,k]\big),
\end{align*}
where $\mathbf{Q}_{n}[i,j]$ is a size $n-1$ sub-matrix of $\mathbf{Q}_{n}$ obtained by deleting row $i$ and column $j$. Since the principal sub-matrix $\mathbf{Q}_{n}[1,1]$ is also diagonal, its determinant is $\prod_{i=2}^n q_{ii}$ and thus the second term in the above expression must be zero. This implies that the determinant for any size $2$ principal sub-matrix of $\mathbf{Q}_n$ with size $n-2$ index set $\alpha$ and entries $a_{ij}$ must satisfy:
\begin{align*}
    \text{det} \big( \mathbf{Q}_n[\alpha]\big) = a_{11}a_{22} -  a_{12}a_{21} = a_{11}a_{22}
\end{align*}
Note that the indices (1,2) refer without loss of generality to any arbitrary pair of sectors in the set $\{1,...,n\}$, and all of the following results must hold for all $n\text{C}2$ pairwise combinations of units. Substitute terms and notice that all principal sub-matrices of a symmetric matrix must also be symmetric (i.e., $a_{12} = a_{21}$) to obtain the equivalent restriction:
\begin{align*}
    0  = a_{12} & = w_{12} \sigma_2^2 - \rho_{12} \sigma_1 \sigma_2 (1+w_{21}w_{12}) + w_{21}\sigma_1^2 
\end{align*}
Or identically: 
\begin{align}
\label{eq:condition for diagonal idiosyncratic variance}
         w_{21} \mathbb{V}[y_1] + w_{12} \mathbb{V}[y_2] - \text{cov}(y_1,y_2) (1+w_{12}w_{21}) = 0,
\end{align}
where $\mathbf{V}[y_i]$ and $\text{cov}(y_i,y_j)$ are the $(i,i)$th and $(i,j)$th entries of $\mathbf{\Sigma_y}$, respectively. I consider two different cases. 
\paragraph{Case 1.} At least one of $w_{21}$ or $w_{12}$ is zero. If both are zero, then \eqref{eq:condition for diagonal idiosyncratic variance} holds trivially. However, if $w_{ij} = 0$ for all $i\neq j$, then $\mathbb{W}$ no longer has full rank. If only one is zero (e.g., $w_{12} = 0$), this implies that $\text{cov}(y_1,y_2) = w_{21}\mathbb{V}[y_1]$ which is consistent with \eqref{eq:reduced form static propagation}. 

\paragraph{Case 2.} Both $w_{21}$ and $w_{12}$ are positive. Redefine $\tilde{y}_1:= y_1/\sqrt{w_{21}}$ and $\tilde{y}_2 := y_2 / \sqrt{w_{12}}$ and apply the Cauchy-Schwarz Inequality:
\begin{align}
\label{eq:cauchy schwarz inequality proof 1}
    \frac{(\mathbb{V}[\tilde{y}_1] + \mathbb{V}[\tilde{y}_2])^2}{(w_{21}w_{12})\cdot  (1+w_{21}w_{12})^2} = \text{cov}(\tilde{y}_1,\tilde{y}_2)^2  \leq \mathbb{V}(\tilde{y}_1)\mathbb{V}(\tilde{y}_2)
\end{align}
For now write $c := w_{21}w_{12} (1+w_{21}w_{12})^2$:
\begin{align*}
    \mathbb{V}[\tilde{y}_1]^2 +  \mathbb{V}[\tilde{y}_2]^2  +  2\mathbb{V}[\tilde{y}_1] \mathbb{V}[\tilde{y}_2] \leq c \mathbb{V}(\tilde{y}_1)\mathbb{V}(\tilde{y}_2) \\
     \mathbb{V}[\tilde{y}_1]^2 +  \mathbb{V}[\tilde{y}_2]^2  +  (2 - c)\mathbb{V}[\tilde{y}_1] \mathbb{V}[\tilde{y}_2] \leq 0
\end{align*}
The left hand side is a quadratic equation in $\mathbf{R}_+^2$ of the variables $\mathbf{V}[\tilde{y}_1]$ and $\mathbf{V}[\tilde{y}_2]$. Suppose without loss of generality that $\mathbb{V}[\tilde{y}_1] = \mathbb{V}[\tilde{y}_2] > 0$, then $c$ must satisfy:
\begin{align*}
    (4-c)\mathbb{V}[\tilde{y}]^2 \leq 0,
\end{align*}
which implies that $c \geq 4$, or equivalently:
\begin{align*}
    w_{21}w_{12}(1+w_{12}w_{21})^2\geq 4
\end{align*}
Recall that since $\mathbf{W}$ is a non-negative matrix, any pair of weights which satisfy this inequality must have $w_{12}w_{21} \geq 1$. Note that due to \eqref{eq:cauchy schwarz inequality proof 1}, this restriction is necessary but not sufficient. Moreover, since the spectral radius $\rho(\mathbf{W}) \leq 1$, there must be at least one row or column bounded above by one (see e.g., Theorem 8.1.22 in Horn and Johnson (2013)). 

\newpage
\subsection{Proof of Proposition \ref{prop:aggregate consumption growth}}
\label{asec:proof of aggregate consumption growth}
Assuming $\beta_i = 1/n$ for all $i$, equilibrium consumption expenditure is given by $C_t = \sum_{i=1}^n p_{it} c_{it} = \prod_{i=1}^n c_{it}^{1/n}$.   Given the price normalization in Section A.4.4. from Appendix \ref{asec:General Equilibrium Model of Input-Output Linkages}, I can write:
\begin{align*}
    \Delta c_{t+1} := \log(C_{t+1}/C_t) & = \log\bigg(\prod_{i=1}^n \bigg(\frac{c_{i,t+1}}{c_{it}}\bigg)^{1/n}\bigg) \\
    & = \frac{1}{n}\sum_{i=1}^n \log(c_{i,t+1}/c_{it}) \\
     & = \frac{1}{n}\sum_{i=1}^n \log(y_{i,t+1}/y_{it})
\end{align*}
Plugging in the reduced form from \eqref{eq:firm level cash flows} yields: 
\begin{align*}
    \Delta c_{t+1} & = \frac{1}{n} \sum_{i=1}^n \big(\Delta z_{i,t+1} + \Delta g_{i,t+1}\big) \\
    & = \gamma_u \cdot \frac{1}{n}\sum_{i=1}^n a_{t+1} + \gamma_d \cdot \frac{1}{n}\sum_{i=1}^n g_{t+1} - \beta_u \cdot \frac{1}{n}\sum_{i=1}^n \varepsilon_{iu,t+1} - \beta_d \cdot \frac{1}{n}\sum_{i=1}^n \varepsilon_{id,t+1}
\end{align*}

\subsection{Proof of Proposition \ref{prop:concentration of network factors}}
\label{asec:proof concentration of network factors}
The proof of Proposition \ref{prop:distribution of consumption growth} established the asymptotic normality of $W_{nqt}$ for all $q$ and $t$. I omit subscripts and write:
\begin{align*}
    \sigma_n^2 = \var(W_{n}) & = \frac{1}{n^2}\sum_{i=1}^n \var[\varepsilon_{i}] + \frac{1}{n^2}\sum_{i\neq j} \text{cov} (\varepsilon_{i}, \varepsilon_{j}) \\
    & = \frac{1}{n^2}\sum_{i=1}^n p_{i}(1-p_i) + \frac{1}{n^2}\sum_{i\neq j} \text{cov} (\varepsilon_{i}, \varepsilon_{j}) \cdot \mathbbm{1}\{\mathbbm{1}\{i,j \in A_{i}(\mathcal{G}_{nq}) \cap A_{j}(\mathcal{G}_{nq})\} \}
\end{align*}
When $n>1$. This quantity can be bounded above by:
\begin{align*}
    \\
   \sigma_n^2 &  \leq \frac{1}{n} + \frac{\bar{M}_n}{n} \leq 2\frac{\bar{M}_n}{n} = o(1),
\end{align*}
where the final equality follows from Assumption \ref{ass:bounded maximal dependency 2}. Moreover, for $k>0$ Chebyshev's Inequality yields:
\begin{align*}
    \text{Pr}\big(|W_n - \mathbb{E}[W_n]| \geq 2k\bar{M_n} / n\big) & \leq  \text{Pr}\big(|W_n - \mathbb{E}[W_n]| \geq k\sigma_n\big) \leq \frac{1}{k^2}
\end{align*}

\newpage

\subsection{Proof of Proposition \ref{prop:distribution of consumption growth}}
\label{asec:proof distribution of consumption growth}
This proof is an application of Theorem 2 from \citet{Janson1988}. Fix $t$ and $q$ and omit the time-subscript without loss of generality. Since $\varepsilon_{iq} \sim \text{Bernoulli}(p_{iq})$, the mean and variance of firm cash flow shocks conditional on $p_{iq}$ can be written:
\begin{align*}
    \mathbb{E}[\varepsilon_{iq}] = p_{iq}, \quad \var(\varepsilon_{iq}) = p_{iq}(1-p_{iq}),
\end{align*}
and the covariance between firm shocks can be written:
\begin{align*}
    \cov(\varepsilon_{iq},\varepsilon_{jq}) & = \mathbb{E}[\varepsilon_{iq}\varepsilon_{jq}] - \mathbb{E}[\varepsilon_{iq}] \mathbb{E}[\varepsilon_{jq}]  = \underbrace{\text{Pr}(\varepsilon_{iq} = 1, \varepsilon_{jq}=1)}_{=:p_{ijq}} - p_{iq}p_{jq} > 0.
\end{align*}
Therefore, the mean and variance of $W_{nq} = \frac{1}{n} \sum_{i=1}^n \varepsilon_{iq}$ can be written:
\begin{align*}
    \mu_{nq} := \mathbb{E}[W_{nq}] & = \frac{1}{n}\sum_{i=1}^n \mathbb{E}[\varepsilon_{iq}] = \frac{1}{n}\sum_{i=1}^n p_{iq}\\
    \sigma_{nq}^2 := \var\big[W_{nq}\big] & =  \frac{1}{n^2}\sum_{i=1}^n \var[\varepsilon_{iq}] + \frac{1}{n^2}\sum_{i\neq j} \text{cov} (\varepsilon_{iq}, \varepsilon_{jq}) \\
    & \geq \frac{1}{n^2}\sum_{i=1}^n \var[\varepsilon_{iq}],\\
    & = \frac{1}{n^2}\sum_{i=1}^n p_{iq}(1-p_{iq})\\
    & \geq \frac{1}{n} \cdot \min_{i} p_{iq}
\end{align*}

Let $\bar{D}_{nq}$ denote the maximal number of edges incident to a single vertex in graph $\mathcal{G}_{nq}$. Theorem 2 from \citet{Janson1988} requires that $\mathcal{X}_{n,m} = o(1)$ for some integer $m$, where $\mathbb{X}_{n,m}$ is given by:
\begin{align*}
    \mathcal{X}_{n,m} & := \bigg(\frac{n}{\bar{D}_{nq}}\bigg)^{1/m} \cdot \frac{\bar{D}_{nq}}{n \sigma_{nq}}  = \frac{\bar{D}_{nq}^{1-1/m}}{n}\cdot \frac{1}{n\sigma_{nq}} \leq \frac{\bar{D}_{nq}^{1-1/m}}{n}.
\end{align*}
Assumption \ref{ass:bounded maximal dependency} implies:
\begin{align*}
    \mathcal{X}_{n,2} = o(1),
\end{align*}
and therefore:
\begin{align*}
    (W_{nq} - \mu_{nq}) / \sigma_{nq} \xrightarrow{d} \mathcal{N}(0,1)
\end{align*}
Combine this result with the assumptions that $a_t \sim \mathcal{N}(0,\sigma_a^2)$ and $g_t \sim \mathcal{N}(0,\sigma_g^2)$ to get the final distribution of consumption growth.


\newpage
\section{Upstream and Downstream Propagation Matrices}
\label{asec:constructing industry upstream and downstream propagation matrices}

\subsection{Construction from BEA Data}
 In this section, I discuss a simple method for constructing from the data the upstream and downstream propagation matrices developed in the previous section.\footnote{Similar procedures are discussed in \citet{AcemogluAkjitKerr2016}, \citet{OzdagliWeber2017}, and \citet{GofmanSegalWu2020}.} These matrices capture the strength of a connection between an industry and its customer or supplier industries. Shocks transmitted from customer $i$ to supplier $j$ (supplier $i$ to customer $j$) should depend on the strength of the connection in the upstream (downstream) direction. I construct these matrices directly from the BEA make and use tables described in \href{https://www.bea.gov/sites/default/files/methodologies/IOmanual_092906.pdf}{Horowitz, Planting, et al. (2006)}. Consider again an economy with $n$ industries.
 
 \paragraph{The Make Table} I extract from the BEA make table an $n\times n$ industry-by-commodity matrix with entries: 
 \begin{align*}
     (MAKE)_{ij} = OUT_{i\to j} \equiv \text{dollar value of commodity $j$ produced by industry $i$}
 \end{align*}
 Note that the BEA makes a slight distinction between commodities and industries, since in principle an industry might produce another industry's commodity as a by-product of its own output. Next, I denote the total production of commodity $j$ by $OUT_j := \sum_{i=1}^n OUT_{i\to j}$. Using the notation in \href{https://www.bea.gov/sites/default/files/methodologies/IOmanual_092906.pdf}{Horowitz, Planting, et al. (2006)}, I define the \emph{market share matrix} with the following entries:
 \begin{align*}
     (MKTSHARE)_{ij} = \frac{(MAKE)_{ij}}{OUT_j} = \frac{OUT_{i\to j}}{OUT_j}
 \end{align*}
 Here, the $(i,j)$ entry describes the share of industry $i$ in the total production of commodity $j$. Equivalently, I can write $MKTSHARE = MAKE \odot (\boldsymbol \iota_{n} \cdot \mathbf{S})$ where $\odot$ denotes the Hadamard (elementwise) product and $S = (OUT_1^{-1},...,OUT_n^{-1})$ is the $1\times n$ scaling vector. 
 \paragraph{The Use Table} Similarly, I extract from the BEA use table an $n\times n$ commodity-by-industry matrix with entries: 
 \begin{align*}
     (USE)_{ij} = IN_{i\to j} \equiv \text{dollar value of commodity $i$ used as input by industry $j$}
 \end{align*}
 Define the total output of industry $i$ by $y_i$.\footnote{I calculate total industry output from the BEA use table as the sum of total intermediates, scrap, and value added. An industry's value added is defined by the BEA as the ``market value it adds in production, or the difference between the price at which it sells its products and the cost of the inputs it purchases from other industries". } Then I construct the \emph{input requirement matrix} by rescaling the value of an industry's inputs by the industry's total value as measured by output. The entries of this matrix are given by:
 \begin{align}
     (INPUTREQ)_{ij} = \frac{(USE)_{ij}}{y_j} = \frac{IN_{i\to j}}{y_j}.
 \end{align}
 The $(i,j)$ entry of the above matrix describes the importance of industry $j$'s inputs from industry $i$ relative to $j$'s total size. 
 
 \paragraph{Scrap Adjustment} The BEA input-output tables include scrap as a commodity which includes any by-products of production with zero market demand. I therefore redefine the total output of an industry as the non-scrap output. Mathematically, this adjustment is implemented using the non-scrap ratio, calculated as follows: 
 \begin{align*}
     \theta_i = \frac{y_i - scrap_i}{y_i},
 \end{align*}
 where $scrap_i$ denotes the total scrap produced by industry $i$. I then write the entries of the scrap-adjusted market share matrix as follows:
 \begin{align*}
     \Tilde{(MKTSHARE)}_{ij} = \frac{OUT_{i\to j}}{OUT_j} \cdot \frac{1}{\theta_j}
 \end{align*}
 \paragraph{Direct Requirements} From the market share and input requirements matrices, I then construct the industry-by-industry direct requirement table, denoted by $\mathbf{W}$:
 \begin{align*}
     \underbrace{\mathbf{W}}_{\text{industry} \times \text{industry}} = \underbrace{\Tilde{MKTSHARE}}_{\text{industry} \times \text{commodity}} \cdot \underbrace{(INPUTREQ)}_{\text{commodity} \times \text{industry}}.
 \end{align*}
 To understand this construction, consider the entries of $\mathbf{W}$:
 \begin{align*}
     (\mathbf{W})_{ij} = \frac{1}{y_j} \sum_{k=1}^n \underbrace{ \frac{OUT_{i\to k}}{OUT_k \cdot \theta_k}}_{\text{share of $i$ in producing $k$}} \cdot \underbrace{IN_{k \to j}}_{\text{use of $k$ by $j$}}.
 \end{align*}
 Here, the $(i,j)$ entry captures industry $j$'s total dependence on inputs from industry $i$ relative to it's total output $y_j$. The key assumption here is as follows. If industry $j$ purchases $X$ of commodity $K$, then the proportion of $K$ coming from industry $i$ is equal to $i$'s adjusted market share of production of $K$. This is a reasonable assumption on average. Given this assumption, the following identity holds:
 \begin{align*}
     (\mathbf{W})_{ij} \equiv \frac{SALES_{i\to j}}{SALES_j} \iff (\mathbf{W}^\top)_{ij} = \frac{SALES_{j\to i}}{SALES_i}
 \end{align*}
 The downstream weighting matrix is thus defined: 
 \begin{align*}
     \mathbf{A} \equiv \mathbf{W}_{down} := \mathbf{W}^\top,
 \end{align*}
 whose $(i,j)$ entry in the above matrix represents the dependence of industry $i$ on input from industry $j$ (i.e., shocks to supplier $j$ propagate downstream to customer $i$ according to the corresponding downstream weight). The sum of row $i$ in this matrix is equal to $x_i/y_i$ where $x_i$ is industry $i$'s total input purchases relative to its size (normalized by $i$'s total output). For the upstream weighting matrix, I require the intermediate rescaling matrix $\mathbf{R}$ with entries $(\mathbf{R})_{ij} = y_j / y_i$. The upstream weighting matrix is thus defined:
 \begin{align*}
     \hat{\mathbf{A}}^\top \equiv \mathbf{W}_{up}:= \mathbf{W} \odot \mathbf{R},
 \end{align*}
 where $\odot$ is the Hadamard (elementwise) product. The $(i,j)$ entry in the above matrix represents the dependence of industry $i$ on sales to industry $j$ (i.e., shocks to customer $j$ propagate upstream to supplier $i$ according to the corresponding upstream weight) 
 \begin{align*}
     (\mathbf{W}\odot\mathbf{R})_{ij} = \frac{SALES_{i\to j}}{SALES_i}.
 \end{align*}

 \subsection{Descriptive Statistics}
 \label{asec:summary input output tables}
 
 \begin{figure}[H]
    \centering
     \subfloat[\centering Kernel Density]{{\includegraphics[width=.5\textwidth]{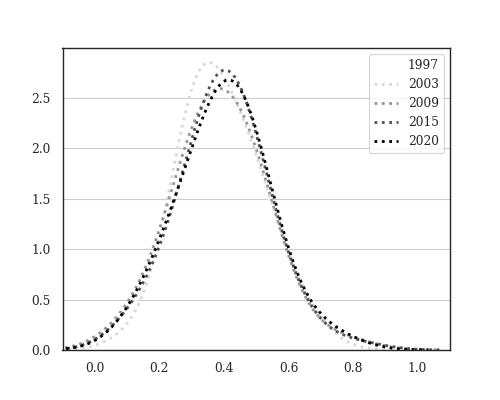} }}%
    \subfloat[\centering Empirical CCDF]{{\includegraphics[width=.5\textwidth]{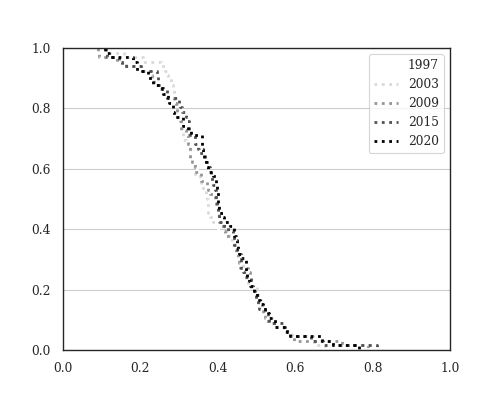} }}%
    \caption{\textbf{Weighted In-Degree Distribution}}
    \fnote{\scriptsize \textit{This figure visualizes the empirical distribution of weighted in and out-degrees across 66 non-government industries as defined in the Bureau of Economic Analysis (BEA) Make and Use Tables. The left panel shows the Gaussian kernel density estimate of the distribution, and the right panel shows the empirical counter-cumulative distribution function.}}
    \label{fig:in degree distribution}
\end{figure}

\begin{figure}[H]
    \centering
     \subfloat[\centering Kernel Density]{{\includegraphics[width=.5\textwidth]{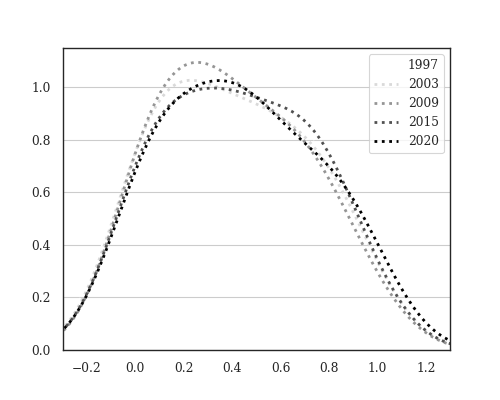} }}%
    \subfloat[\centering Empirical CCDF]{{\includegraphics[width=.5\textwidth]{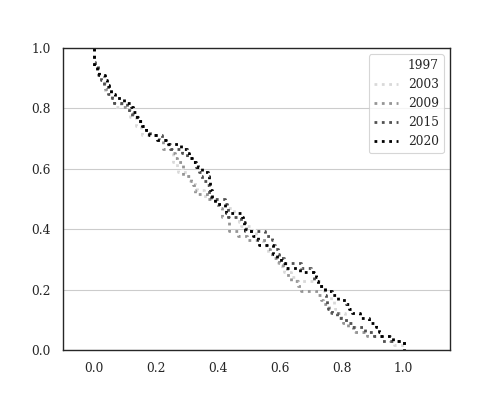} }}%
    \caption{\textbf{Weighted Out-Degree Distribution}}
    \fnote{\scriptsize \textit{This figure visualizes the empirical distribution of weighted in and out-degrees across 66 non-government industries as defined in the Bureau of Economic Analysis (BEA) Make and Use Tables. The left panel shows the Gaussian kernel density estimate of the distribution, and the right panel shows the empirical counter-cumulative distribution function.}}
    \label{fig:out degree distribution}
\end{figure}

\begin{figure}[H]
    \centering
     \subfloat[\centering Kernel Density]{{\includegraphics[width=.5\textwidth]{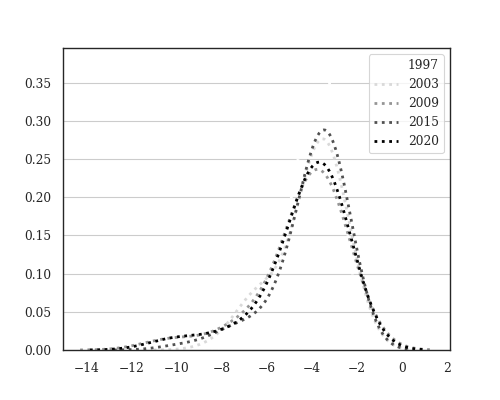} }}%
    \subfloat[\centering Empirical CCDF]{{\includegraphics[width=.5\textwidth]{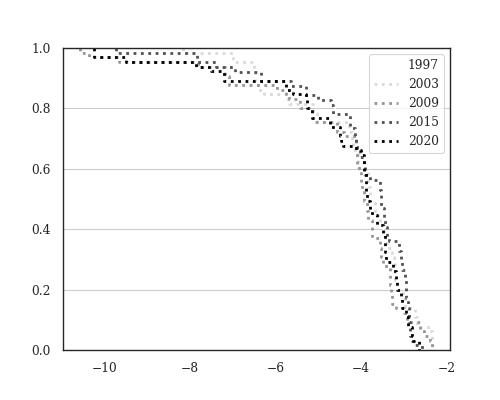} }}%
    \caption{\textbf{Distribution of Downstream Centrality}}
    \fnote{\scriptsize \textit{This figure visualizes the empirical distribution of weighted log eigenvector centrality of the upstream propagation matrix across 66 non-government industries as defined in the Bureau of Economic Analysis (BEA) Make and Use Tables. The left panel shows the Gaussian kernel density estimate of the distribution, and the right panel shows the empirical counter-cumulative distribution function.}}
    \label{fig:downstream centrality}
\end{figure}

\begin{figure}[H]
    \centering
     \subfloat[\centering Kernel Density]{{\includegraphics[width=.5\textwidth]{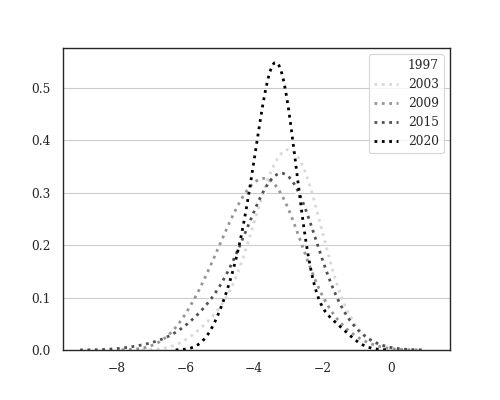} }}%
    \subfloat[\centering Empirical CCDF]{{\includegraphics[width=.5\textwidth]{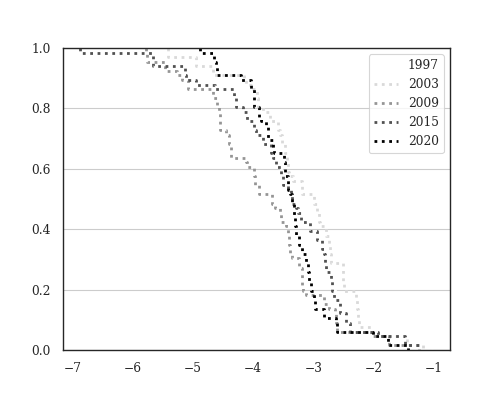} }}%
    \caption{\textbf{Distribution of Upstream Centrality}}
    \fnote{\scriptsize \textit{This figure visualizes the empirical distribution of weighted log eigenvector centrality of the upstream propagation matrix across 66 non-government industries as defined in the Bureau of Economic Analysis (BEA) Make and Use Tables. The left panel shows the Gaussian kernel density estimate of the distribution, and the right panel shows the empirical counter-cumulative distribution function.}}
    \label{fig:upstream centrality}
\end{figure}

\begin{table}[H]
  \centering
  \caption{\textbf{Correlation of Input-Output Shares}}
    \begin{tabular}{r|cccccccc}
    \toprule
    \multicolumn{1}{c}{Entry} & $w_{d,ij,t-1}$ & $w_{u,ij,t-1}$ & $h_{d,ij,t-1}$ & $h_{u,ij,t-1}$ & $w_{d,ij,t}$ & $w_{u,ij,t}$ & $h_{d,ij,t}$ & $h_{u,ij,t}$ \\
    \midrule
    $w_{d,ij,t-1}$ & 1     & 0.086 & 0.970 & 0.096 & 0.991 & 0.086 & 0.961 & 0.096 \\
    $w_{u,ij,t-1}$ &       & 1     & 0.094 & 0.967 & 0.084 & 0.994 & 0.093 & 0.961 \\
    $h_{d,ij,t-1}$ &       &       & 1     & 0.110 & 0.962 & 0.095 & 0.990 & 0.111 \\
    $h_{u,ij,t-1}$ &       &       &       & 1     & 0.095 & 0.962 & 0.109 & 0.994 \\
    $w_{d,ij,t}$ &       &       &       &       & 1     & 0.086 & 0.970 & 0.096 \\
    $w_{u,ij,t}$ &       &       &       &       &       & 1     & 0.095 & 0.967 \\
    $h_{d,ij,t}$ &       &       &       &       &       &       & 1     & 0.111 \\
    $h_{u,ij,t}$ &       &       &       &       &       &       &       & 1 \\
    \bottomrule
    \end{tabular}%
    \fnote{\scriptsize \textit{Note:} This table report correlations between entries in the upstream and downstream propagation matrices $W_{qt}$ and their Leontief inverses $H_{qt}$. I construct annual matrices from the BEA Input-Output Accounts for 66 non-government industries between 1997-2020.}
  \label{tab:correlation of input-output shares}
\end{table}%

\begin{figure}[H]
    \centering
    \subfloat[\centering In-Degree (downstream)]{{\includegraphics[width=8cm]{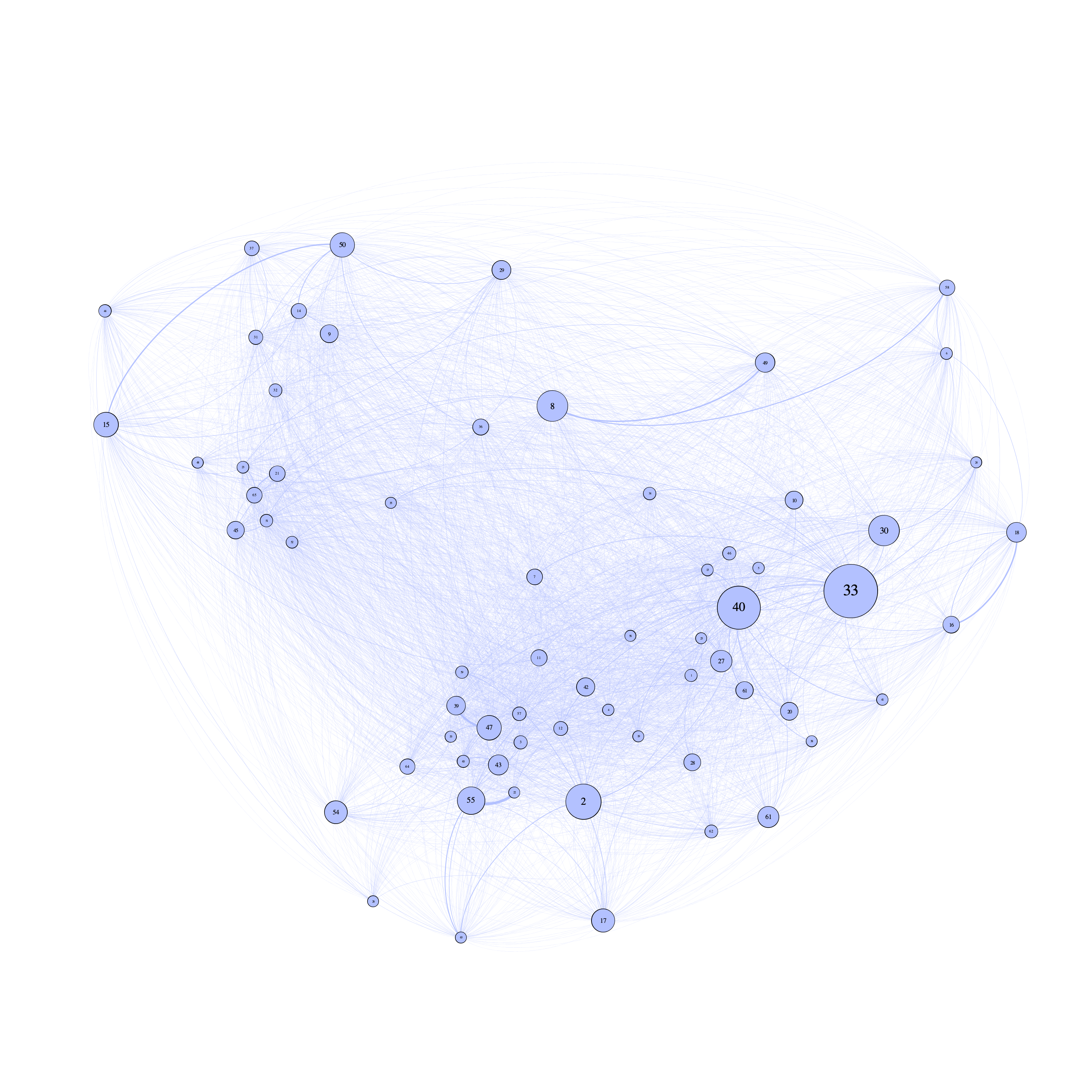} }}%
    \subfloat[\centering Out-Degree (upstream)]{{\includegraphics[width=8cm]{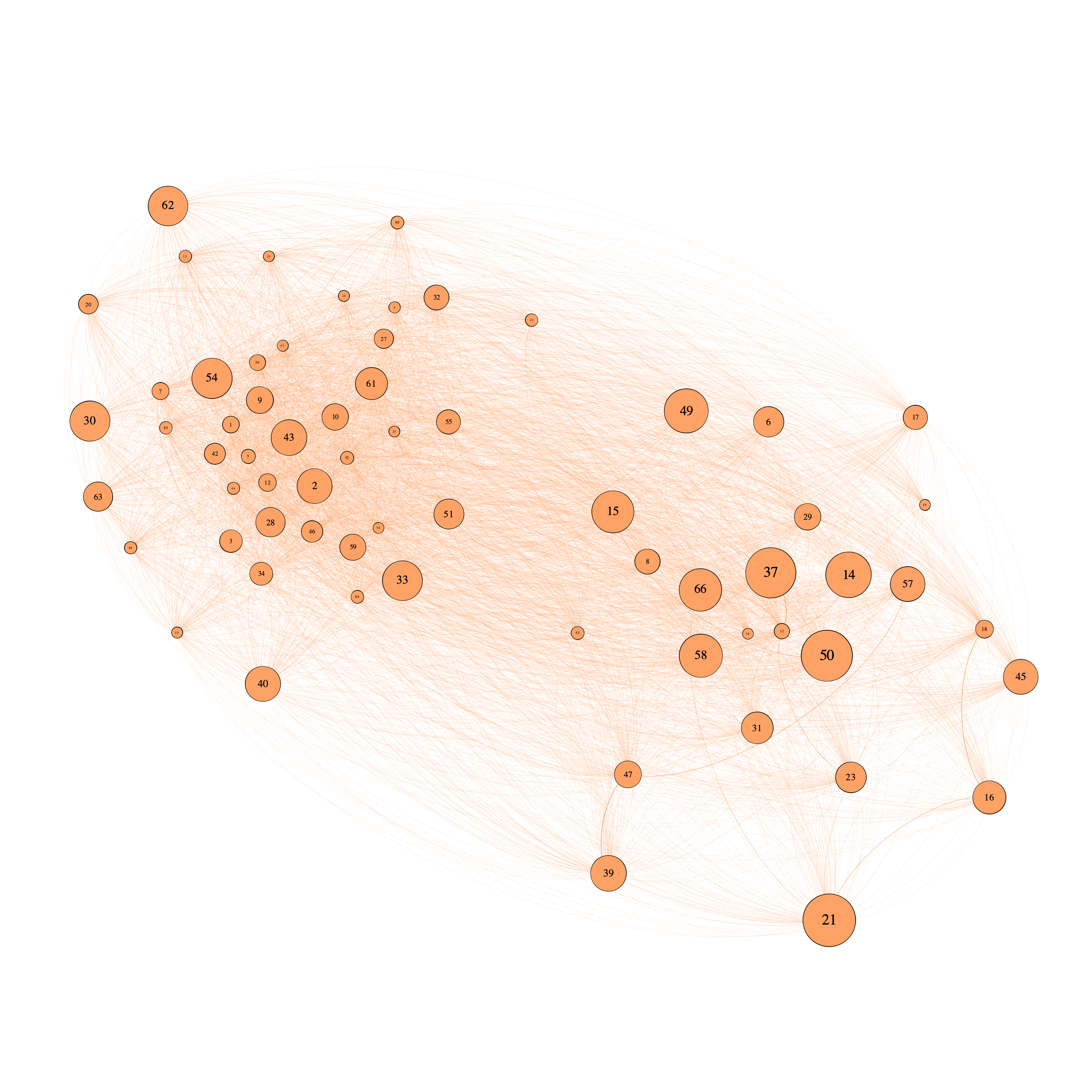} }}%
    \caption{\textbf{Visualizing Weighted Degree by Industry}}
    \fnote{\scriptsize The left panel shows the average weighted in-degree by sector as constructed from the downstream propagation network $\textbf{W}_{down}$. The right panel shows the average weighted in-degree by sector as constructed from the upstream propagation network $\textbf{W}_{up}$. The sectors are numbered identically in both panels, according to the BEA 66 non-government industry classification. }%
    \label{fig:in and out degree graphs sector}%
\end{figure} 


\newpage
\section{Determinants of Network Variance}
\label{asec:determinants of network variance}

\begin{table}[H]
  \centering
  \tiny
  \caption{\textbf{Correlation in Predictors of Realized Variance} }
    \begin{tabular}{r|ccccccccccc}
    \toprule
    \multicolumn{1}{c|}{Predictor} & \multicolumn{1}{c}{$\beta_{vol}$} & \multicolumn{1}{c}{$C_u$} & \multicolumn{1}{c}{$C_d$} & \multicolumn{1}{c}{Between (d)} & \multicolumn{1}{c}{Between (s)} & \multicolumn{1}{c}{Size} & \multicolumn{1}{c}{Across (s)} & \multicolumn{1}{c}{Across (d)} & \multicolumn{1}{c}{VP} & \multicolumn{1}{c}{mkt ivol} & mkt vol \\
\cmidrule{1-12}    $\beta_{vol}$  & \multicolumn{1}{c}{1} & -0.329 & -0.130 & -0.308 & 0.072 & 0.285 & 0.061 & -0.128 & -0.225 & 0.287 & \multicolumn{1}{r}{0.406} \\
    $C_u$ &       & \multicolumn{1}{c}{1} & 0.284 & 0.356 & -0.251 & 0.131 & -0.441 & 0.319 & 0.423 & -0.383 & \multicolumn{1}{r}{-0.396} \\
    $C_d$ &       &       & \multicolumn{1}{c}{1} & 0.588 & -0.018 & -0.128 & -0.116 & 0.288 & 0.289 & -0.003 & \multicolumn{1}{r}{-0.022} \\
    Between (d) &       &       &       & \multicolumn{1}{c}{1} & -0.205 & -0.113 & -0.156 & 0.704 & 0.708 & -0.270 & \multicolumn{1}{r}{-0.311} \\
    Between (s) &       &       &       &       & \multicolumn{1}{c}{1} & -0.133 & 0.056 & -0.163 & -0.313 & 0.073 & \multicolumn{1}{r}{0.054} \\
    Size  &       &       &       &       &       & \multicolumn{1}{c}{1} & -0.230 & 0.018 & -0.111 & -0.386 & \multicolumn{1}{r}{-0.275} \\
    Across (s) &       &       &       &       &       &       & \multicolumn{1}{c}{1} & -0.088 & -0.159 & 0.205 & \multicolumn{1}{r}{0.255} \\
    Across (d) &       &       &       &       &       &       &       & \multicolumn{1}{c}{1} & 0.627 & -0.224 & \multicolumn{1}{r}{-0.249} \\
    VP    &       &       &       &       &       &       &       &       & \multicolumn{1}{c}{1} & -0.279 & \multicolumn{1}{r}{-0.322} \\
    mkt ivol &       &       &       &       &       &       &       &       &       & \multicolumn{1}{c}{1} & \multicolumn{1}{r}{0.965} \\
    mkt vol &       &       &       &       &       &       &       &       &       &       & 1 \\
    \bottomrule
    \end{tabular}%
    \fnote{\scriptsize \textit{Notes:} This table reports the average correlation between predictors in the panel regression from Table \ref{tab:var regressions main}. $C_u$ and $C_d$ refer to upstream and downstream centrality, the $(d)$ and $(s)$ labels denote demand and supply-side concentration, size is average output, and VP is vertical position. Idiosyncratic volatility is calculated from FF3 residual returns. }
  \label{tab:corr vol predictors}%
\end{table}%

\newpage

\begin{table}[H]
  \centering
  \scriptsize
  \caption{\textbf{Network Determinants of Industry Variance (TFP growth shocks)}}
    \begin{tabular}{lcccccc}
    \toprule
    \multicolumn{7}{c}{Panel A: Market Return Variance} \\
    \midrule
          & (1)   & (2)   & (3)   & (4)   & (5)   & (6) \\
    \midrule
    Self-origin (demand) & 0.365*** & 0.246*** &       &       & 0.217*** & 0.119*** \\
          & (0.038) & (0.039) &       &       & (0.021) & (0.022) \\
    Across (demand) & 0.100*** & 0.084*** &       &       & 0.085*** & 0.119*** \\
          & (0.020) & (0.018) &       &       & (0.020) & (0.019) \\
   Between (demand) & 0.043** & 0.010 &       &       & 0.091*** & -0.076*** \\
          & (0.018) & (0.022) &       &       & (0.019) & (0.020) \\
    Self-origin (supply) &       &       & 0.039* & 0.149*** & 0.217*** & 0.119*** \\
          &       &       & (0.029) & (0.038) & (0.021) & (0.022) \\
    Across (supply) &       &       & 0.119*** & 0.147*** & 0.270*** & 0.170*** \\
          &       &       & (0.021) & (0.026) & (0.027) & (0.037) \\
    Between (supply) &       &       & 0.085*** & 0.042** & 0.086*** & 0.179*** \\
          &       &       & (0.019) & (0.019) & (0.025) & (0.034) \\
    Size &       & -0.202*** &       & -0.325*** &       & -0.183*** \\
          &       & (0.030) &       & (0.023) &       & (0.034) \\
    Upstream centrality &       & -0.145** &       & 0.129** &       & -0.238*** \\
          &       & (0.062) &       & (0.051) &       & (0.066) \\
    Downstream centrality &       & 1.956*** &       & 0.296** &       & 2.137*** \\
          &       & (0.145) &       & (0.125) &       & (0.165) \\
    Durability &       & 1.088*** &       & 0.389*** &       & 1.164*** \\
          &       & (0.153) &       & (0.108) &       & (0.156) \\
    Vertical position &       & 3.890*** &       & 0.243*** &       & 4.063*** \\
          &       & (0.116) &       & (0.083) &       & (0.134) \\
    Constant & -2.123 & -0.634 & -2.821 & -0.096 & -1.473 & -0.693 \\
    Obs  & 2359  & 1626  & 2861  & 2138  & 2114  & 1471 \\
    Adj $R^2$ & 0.225 & 0.607 & 0.083 & 0.191 & 0.288 & 0.633 \\
    \midrule
    \multicolumn{7}{c}{Panel B: Cash Flow Variance} \\
    \midrule
          & (1)   & (2)   & (3)   & (4)   & (5)   & (6) \\
    \midrule
    Self-origin (demand) & 0.363*** & 0.009 &       &       & 0.190*** & 0.017 \\
          & (0.072) & (0.093) &       &       & (0.037) & (0.049) \\
   Across (demand) & 0.038 & -0.050* &       &       & 0.011 & 0.099** \\
          & (0.039) & (0.039) &       &       & (0.040) & (0.042) \\
    Between (demand) & 0.136*** & 0.128*** &       &       & 0.156*** & 0.184*** \\
          & (0.037) & (0.042) &       &       & (0.039) & (0.047) \\
    Self-origin (supply) &       &       & 0.026 & 0.286*** & 0.190*** & 0.017 \\
          &       &       & (0.065) & (0.076) & (0.037) & (0.049) \\
    Across (supply) &       &       & 0.250*** & 0.384*** & 0.309*** & 0.107* \\
          &       &       & (0.051) & (0.062) & (0.061) & (0.091) \\
    Between (supply) &       &       & 0.093** & 0.169*** & 0.202*** & 0.162** \\
          &       &       & (0.044) & (0.055) & (0.057) & (0.086) \\
    Size &       & -0.161** &       & -0.333*** &       & -0.155** \\
          &       & (0.064) &       & (0.052) &       & (0.075) \\
    Upstream centrality &       & -0.069 &       & 0.562*** &       & -0.100 \\
          &       & (0.150) &       & (0.132) &       & (0.162) \\
    Downstream centrality &       & 2.193*** &       & -0.070 &       & 2.315*** \\
          &       & (0.343) &       & (0.305) &       & (0.379) \\
    Durability &       & 2.204*** &       & 1.172*** &       & 2.307*** \\
          &       & (0.272) &       & (0.182) &       & (0.285) \\
    Vertical position &       & 3.502*** &       & 1.142* &       & 3.590*** \\
          &       & (0.268) &       & (0.181) &       & (0.320) \\
    Constant & -4.201 & -0.711 & -5.246 & 0.023 & -4.123 & -0.75 \\
    Obs  & 2359  & 1626  & 2861  & 2138  & 2114  & 1471 \\
    Adj $R^2$   & 0.074 & 0.248 & 0.017 & 0.08  & 0.087 & 0.271 \\
    \bottomrule
    \end{tabular}%
  \fnote{\tiny \textit{Notes}: This table reports panel regressions of realized industry variance on a variety of characteristics including the average log variance of supply and demand shocks, log concentration across trade partners, log concentration between trade partners, log total employment (size), log centrality of the upstream and downstream propagation networks, durability of output, vertical position in the supply chain, and industry cluster and year fixed effects. I calculate network components as the average value of 1000 bootstrap samples that randomly drop 10\% of estimated pairwise non-zero correlations. In Panel A, the dependent variable is the log variance of annualized monthly returns on an equal-weighted industry portfolio. In Panel B, the dependent variable is the log variance of total quarterly year-on-year industry sales growth. I obtain return data from CRSP and sales data from Compustat. Concentration between and across trade partners are calculated as in \ref{eq:across empirical} and \ref{eq:between empirical}, where I calculate the variance-covariance matrix of supply and demand shocks directly from four-factor TFP growth in the NBER-CES Database \citep{nbercesdata}. Following \citet{Ahern2013NetworkCA}, I compute industry centrality as the eigenvector centrality of upstream and downstream propagation adjacency matrices. I obtain durability classifications from \citet{Gomes2009durability} and calculate vertical position of each industry as in \citet{Antrs2012} and \citet{GofmanSegalWu2020}. $^{***}$, $^{**}$ and $^{*}$ indicate significance at the 1\%, 5\%, and 10\% levels, respectively. Standard error are clustered at the BEA 15 major industry  group level.  Sample is at an annual frequency from 1988 to 2017 for 479 BEA manufacturing industries.  }
  \label{tab:industry variance panel regressions (tfp)}%
\end{table}%

\newpage
\begin{table}[H]
  \scriptsize
  \centering
  \caption{\textbf{Network Determinants of Industry Variance (federal procurement shocks)}}
    \begin{tabular}{lcccccc}
    \toprule
    \multicolumn{7}{c}{Panel A: Market Return Variance} \\
    \midrule
          & (1)   & (2)   & (3)   & (4)   & (5)   & (6) \\
    \midrule
    Self-origin (demand) & 1.058*** & 0.903*** &       &       & 0.166* & 0.558*** \\
          & (0.157) & (0.200) &       &       & (0.111) & (0.121) \\
    Across (demand) & 0.088** & 0.314*** &       &       & 0.062* & 0.356*** \\
          & (0.041) & (0.059) &       &       & (0.044) & (0.070) \\
    Between (demand) & 0.149*** & 0.153*** &       &       & 0.437*** & 0.261*** \\
          & (0.040) & (0.039) &       &       & (0.055) & (0.077) \\
    Self-origin (supply) &       &       & 0.219*** & 0.363*** & 0.166* & 0.558*** \\
          &       &       & (0.044) & (0.056) & (0.111) & (0.121) \\
    Across (supply) &       &       & 0.142*** & 0.094* & 0.072 & 0.088 \\
          &       &       & (0.036) & (0.052) & (0.072) & (0.079) \\
    Between (supply) &       &       & 0.150*** & 0.062 & -0.063 & -0.13 \\
          &       &       & (0.034) & (0.039) & (0.075) & (0.077) \\
    Size  &       & 0.059 &       & -0.242*** &       & -0.033 \\
          &       & (0.053) &       & (0.047) &       & (0.053) \\
    Upstream centrality &       & -2.434** &       & -4.885*** &       & -2.140* \\
          &       & (1.053) &       & (1.116) &       & (1.269) \\
    Downstream centrality &       & 10.76*** &       & 13.06*** &       & 10.55*** \\
          &       & (2.686) &       & (2.501) &       & (3.043) \\
    Durability &       & 4.128*** &       & 1.991*** &       & 3.884*** \\
          &       & (0.511) &       & (0.407) &       & (0.460) \\
    Vertical position &       & 11.74*** &       & 4.334*** &       & 10.66*** \\
          &       & (0.915) &       & (0.595) &       & (0.930) \\
    Constant & -6.278 & -3.489 & -5.461 & -4.234 & -6.781 & -3.419 \\
    Obs   & 839   & 666   & 1003  & 828   & 839   & 666 \\
    Adj R2 & 0.595 & 0.76  & 0.071 & 0.331 & 0.666 & 0.762 \\
    \midrule
    \multicolumn{7}{c}{Panel B: Cash Flow Variance} \\
    \midrule
          & (1)   & (2)   & (3)   & (4)   & (5)   & (6) \\
    \midrule
    Self-origin (demand) & 0.425 & 1.853* &       &       & 0.267 & 1.025** \\
          & (0.421) & (0.951) &       &       & (0.293) & (0.488) \\
    Across (demand) & 0.057 & 0.790*** &       &       & 0.183* & 0.768*** \\
          & (0.090) & (0.215) &       &       & (0.121) & (0.207) \\
    Between (demand) & 0.309*** & 0.461*** &       &       & 0.524*** & 0.471** \\
          & (0.091) & (0.121) &       &       & (0.146) & (0.199) \\
    Self-origin (supply) &       &       & 0.310*** & 0.004 & 0.267 & 1.025** \\
          &       &       & (0.109) & (0.185) & (0.293) & (0.488) \\
    Across (supply) &       &       & -0.184* & 0.002 & 0.136 & 0.327* \\
          &       &       & (0.101) & (0.139) & (0.193) & (0.258) \\
    Between (supply) &       &       & 0.235** & 0.028 & -0.247* & -0.268 \\
          &       &       & (0.092) & (0.112) & (0.192) & (0.244) \\
    Size  &       & 0.183 &       & -0.406*** &       & -0.141 \\
          &       & (0.209) &       & (0.122) &       & (0.214) \\
    Upstream centrality &       & 1.430 &       & 3.297* &       & 4.077* \\
          &       & (3.309) &       & (2.975) &       & (3.421) \\
    Downstream centrality &       & 5.582 &       & -5.585 &       & 0.810 \\
          &       & (8.923) &       & (6.799) &       & (8.718) \\
    Durability &       & 5.285*** &       & 1.175 &       & 5.536*** \\
          &       & (1.493) &       & (0.847) &       & (1.677) \\
    Vertical position &       & 12.20*** &       & 0.948 &       & 12.73*** \\
          &       & (3.184) &       & (1.236) &       & (4.096) \\
    Constant & -1.223 & -1.809 & -7.991 & 1.81  & -10.389 & -0.263 \\
    Obs   & 839   & 666   & 1003  & 828   & 839   & 666 \\
    Adj R2 & 0.239 & 0.211 & 0.043 & 0.08  & 0.248 & 0.215 \\
    \bottomrule
    \end{tabular}%
  \fnote{\tiny \textit{Notes}: This table reports panel regressions of realized industry variance on a variety of characteristics including the average log variance of supply and demand shocks, log concentration across trade partners, log concentration between trade partners, log total employment (size), log centrality of the upstream and downstream propagation networks, durability of output, vertical position in the supply chain, and industry cluster and year fixed effects. I calculate network components as the average value of 1000 bootstrap samples that randomly drop 10\% of estimated pairwise non-zero correlations In Panel A, the dependent variable is the log variance of annualized monthly returns on an equal-weighted industry portfolio. In Panel B, the dependent variable is the log variance of total quarterly year-on-year industry sales growth. I obtain return data from CRSP and sales data from Compustat. Concentration between and across trade partners are calculated as in \ref{eq:across empirical} and \ref{eq:between empirical}, where I calculate the variance-covariance matrix of changes in total obligations from newly awarded federal procurement contracts measured from FPDS. Following \citet{Ahern2013NetworkCA}, I compute industry centrality as the eigenvector centrality of upstream and downstream propagation adjacency matrices. I obtain durability classifications from \citet{Gomes2009durability} and calculate vertical position of each industry as in \citet{Antrs2012} and \citet{GofmanSegalWu2020}. $^{***}$, $^{**}$ and $^{*}$ indicate significance at the 1\%, 5\%, and 10\% levels, respectively. Standard error are clustered at the BEA 15 major industry level. Sample is at an annual frequency from 1991 to 2011 for 479 BEA manufacturing industries.  }
  \label{tab:industry variance panel regressions (trade)}%
\end{table}%

\newpage
\begin{table}[H]
\scriptsize
  \centering
  \caption{\textbf{Network Determinants of Industry Variance (technological proximity)}}
    \begin{tabular}{lcccccc}
    \toprule
    \multicolumn{7}{c}{Panel A: Market Return Variance} \\
    \midrule
          & (1)   & (2)   & (3)   & (4)   & (5)   & (6) \\
    \midrule
    Self-origin (demand) & 0.257*** & 0.352*** &       &       & 0.286*** & 0.308*** \\
          & (0.046) & (0.043) &       &       & (0.050) & (0.047) \\
    Across (demand) & 0.102*** & 0.120*** &       &       & 0.124*** & 0.145*** \\
          & (0.011) & (0.013) &       &       & (0.011) & (0.013) \\
    Between (demand) & 0.246*** & 0.117*** &       &       & 0.227*** & 0.058*** \\
          & (0.013) & (0.018) &       &       & (0.013) & (0.018) \\
    Self-origin (supply) &       &       & 0.290*** & 0.258*** & 0.229*** & 0.419*** \\
          &       &       & (0.035) & (0.039) & (0.040) & (0.043) \\
    Across (supply) &       &       & 0.009 & 0.032 & 0.223*** & -0.035 \\
          &       &       & (0.027) & (0.036) & (0.031) & (0.042) \\
    Between (supply) &       &       & 0.190*** & 0.168*** & 0.284*** & 0.143** \\
          &       &       & (0.031) & (0.043) & (0.034) & (0.051) \\
    Size  &       & -0.343*** &       & -0.312*** &       & -0.334*** \\
          &       & (0.031) &       & (0.027) &       & (0.030) \\
    Upstream centrality &       & -0.755*** &       & -0.102* &       & -0.543*** \\
          &       & (0.108) &       & (0.087) &       & (0.104) \\
    Downstream centrality &       & 3.933*** &       & 1.052*** &       & 3.632*** \\
          &       & (0.279) &       & (0.190) &       & (0.275) \\
    Durability &       & -0.195* &       & -0.177* &       & -0.076 \\
          &       & (0.112) &       & (0.098) &       & (0.108) \\
    Vertical position &       & 1.669*** &       & 0.130 &       & 2.160*** \\
          &       & (0.152) &       & (0.086) &       & (0.172) \\
    Constant & -7.926 & -1.275 & -3.85 & -0.341 & -8.468 & -1.177 \\
    Obs   & 1994  & 1467  & 1994  & 1467  & 1994  & 1467 \\
    Adj R2 & 0.593 & 0.642 & 0.127 & 0.279 & 0.61  & 0.665 \\
    \midrule
    \multicolumn{7}{c}{Panel B: Cash Flow Variance} \\
    \midrule
          & (1)   & (2)   & (3)   & (4)   & (5)   & (6) \\
    \midrule
    Self-origin (demand) & 0.391*** & 0.521*** &       &       & 0.473*** & 0.548*** \\
          & (0.103) & (0.103) &       &       & (0.102) & (0.103) \\
    Across (demand) & 0.065** & 0.061* &       &       & 0.050* & 0.049* \\
          & (0.028) & (0.033) &       &       & (0.029) & (0.034) \\
    Between (demand) & 0.253*** & 0.139*** &       &       & 0.172** & 0.217* \\
          & (0.030) & (0.045) &       &       & (0.077) & (0.122) \\
    Self-origin (supply) &       &       & -0.019 & -0.245*** & -0.162** & -0.185* \\
          &       &       & (0.076) & (0.087) & (0.079) & (0.097) \\
    Across (supply) &       &       & 0.175*** & 0.138* & 0.072 & 0.160* \\
          &       &       & (0.067) & (0.100) & (0.068) & (0.102) \\
    Between (supply) &       &       & 0.307*** & 0.228** & 0.271*** & 0.179*** \\
          &       &       & (0.076) & (0.119) & (0.031) & (0.048) \\
    Size  &       & -0.298*** &       & -0.332*** &       & -0.298*** \\
          &       & (0.070) &       & (0.066) &       & (0.071) \\
    Upstream centrality &       & 0.113 &       & 0.355* &       & -0.068 \\
          &       & (0.267) &       & (0.276) &       & (0.280) \\
    Downstream centrality &       & 3.217*** &       & 0.311 &       & 3.555*** \\
          &       & (0.665) &       & (0.599) &       & (0.680) \\
    Durability &       & 0.574*** &       & 0.265* &       & 0.481** \\
          &       & (0.210) &       & (0.211) &       & (0.218) \\
    Vertical position &       & 1.922*** &       & -0.311 &       & 1.493*** \\
          &       & (0.323) &       & (0.208) &       & (0.373) \\
    Constant & -10.991 & -1.043 & -5.631 & -0.101 & -11.776 & -1.152 \\
    Obs   & 1994  & 1467  & 1994  & 1467  & 1994  & 1467 \\
    Adj R2 & 0.21  & 0.23  & 0.011 & 0.051 & 0.216 & 0.233 \\
    \bottomrule
    \end{tabular}%
   \fnote{\tiny \textit{Notes}: This table reports panel regressions of realized industry variance on a variety of characteristics including the average log variance of supply and demand shocks, log concentration across trade partners, log concentration between trade partners, log total employment (size), log centrality of the upstream and downstream propagation networks, durability of output, vertical position in the supply chain, and industry cluster and year fixed effects. I calculate network components as the average value of 1000 bootstrap samples that randomly drop 10\% of estimated pairwise non-zero correlations. In Panel A, the dependent variable is the log variance of annualized monthly returns on an equal-weighted industry portfolio. In Panel B, the dependent variable is the log variance of total quarterly year-on-year industry sales growth. I obtain return data from CRSP and sales data from Compustat. Concentration between and across trade partners are calculated as in \ref{eq:across empirical} and \ref{eq:between empirical}, where I calculate the variance-covariance matrix of supply and demand shocks using the weighted technological similarity scores following \citet{bloomtechspillovers2013ecma}. Following \citet{Ahern2013NetworkCA}, I compute industry centrality as the eigenvector centrality of upstream and downstream propagation adjacency matrices. I obtain durability classifications from \citet{Gomes2009durability} and calculate vertical position of each industry as in \citet{Antrs2012} and \citet{GofmanSegalWu2020}. $^{***}$, $^{**}$ and $^{*}$ indicate significance at the 1\%, 5\%, and 10\% levels, respectively. Standard errors are clustered at the BEA 15 major industry group level. Sample is at an annual frequency from 1988 to 2018 for 479 BEA manufacturing industries.  }
  \label{tab:industry variance panel regressions (technological similarity)}%
\end{table}%

\newpage
\begin{table}[H]
\scriptsize
  \centering
  \caption{\textbf{Network Determinants of Industry Variance (product similarity)}}
    \begin{tabular}{lcccccc}
    \toprule
    \multicolumn{7}{c}{Panel A: Market Return Variance} \\
    \midrule
          & (1)     & (2)     & (3)     & (4)     & (5)     & (6) \\
    \midrule
    Self-origin (demand) & 0.136** & 0.354*** &       &       & 0.303*** & 0.321*** \\
          & (0.056) & (0.040) &       &       & (0.046) & (0.041) \\
    Across (demand) & 0.103*** & 0.071*** &       &       & -0.065*** & 0.051*** \\
          & (0.013) & (0.011) &       &       & (0.013) & (0.012) \\
    Between (demand) & 0.131*** & 0.081** &       &       & 0.384*** & 0.058*** \\
          & (0.013) & (0.012) &       &       & (0.014) & (0.016) \\
    Self-origin (supply) &       &       & 0.096*** & 0.056* & -0.047 & 0.207*** \\
          &       &       & (0.027) & (0.033) & (0.033) & (0.034) \\
    Across (supply) &       &       & 0.147*** & 0.137*** & 0.186*** & -0.000 \\
          &       &       & (0.015) & (0.020) & (0.020) & (0.021) \\
    Between (supply) &       &       & -0.086*** & -0.106*** & -0.362*** & -0.071*** \\
          &       &       & (0.012) & (0.016) & (0.017) & (0.020) \\
    Size &       & -0.361*** &       & -0.257*** &       & -0.310*** \\
          &       & (0.025) &       & (0.021) &       & (0.024) \\
    Upstream centrality &       & -0.149*** &       & 0.031 &       & -0.176*** \\
          &       & (0.053) &       & (0.049) &       & (0.055) \\
    Downstream centrality &       & 2.660*** &       & 0.811*** &       & 2.891*** \\
          &       & (0.174) &       & (0.117) &       & (0.180) \\
    Durability &       & 0.216** &       & 0.265*** &       & 0.286*** \\
          &       & (0.102) &       & (0.096) &       & (0.098) \\
    Vertical position &       & 3.414*** &       & 0.368*** &       & 3.336*** \\
          &       & (0.086) &       & (0.076) &       & (0.106) \\
    Constant & -5.915 & -0.862 & -3.631 & -0.263 & -8.380 & -0.937 \\
    Obs  & 3800  & 2430  & 3819  & 2430  & 3800  & 2430 \\
    Adj $R^2$ & 0.359 & 0.625 & 0.037 & 0.175 & 0.488 & 0.639 \\
    \midrule
    \multicolumn{7}{c}{Panel B: Cash Flow Variance} \\
    \midrule
           & (1)     & (2)     & (3)     & (4)     & (5)     & (6) \\
    \midrule
    Self-origin (demand) & 0.357*** & 0.549*** &       &       & 0.544*** & 0.558*** \\
          & (0.091) & (0.090) &       &       & (0.091) & (0.091) \\
    Across (demand) & 0.112*** & 0.055** &       &       & -0.059** & 0.011 \\
          & (0.023) & (0.027) &       &       & (0.025) & (0.030) \\
    Between (demand) & 0.118*** & 0.066** &       &       & 0.373*** & 0.108*** \\
          & (0.023) & (0.030) &       &       & (0.027) & (0.040) \\
    Self-origin (supply) &       &       & 0.030 & -0.026 & -0.141* & 0.084* \\
          &       &       & (0.060) & (0.073) & (0.062) & (0.078) \\
    Across (supply) &       &       & 0.129*** & 0.199*** & 0.170*** & 0.091** \\
          &       &       & (0.029) & (0.039) & (0.031) & (0.043) \\
    Between (supply) &       &       & -0.074*** & -0.156*** & -0.347*** & -0.157*** \\
          &       &       & (0.023) & (0.032) & (0.028) & (0.041) \\
    Size &       & -0.313*** &       & -0.242*** &       & -0.273*** \\
          &       & (0.050) &       & (0.048) &       & (0.051) \\
    Upstream centrality &       & 0.185* &       & 0.327** &       & 0.090 \\
          &       & (0.135) &       & (0.132) &       & (0.138) \\
    Downstream centrality &       & 3.213*** &       & 0.791*** &       & 3.587*** \\
          &       & (0.406) &       & (0.300) &       & (0.413) \\
    Durability &       & 1.097*** &       & 1.104*** &       & 1.145*** \\
          &       & (0.171) &       & (0.171) &       & (0.170) \\
    Vertical position &       & 3.495*** &       & 0.226 &       & 3.102*** \\
          &       & (0.193) &       & (0.164) &       & (0.251) \\
    Constant & -10.060 & -1.041 & -5.694 & -0.256 & -12.580 & -1.162 \\
    Obs  & 3800  & 2430  & 3819  & 2430  & 3800  & 2430 \\
    Adj $R^2$    & 0.15  & 0.235 & 0.005 & 0.047 & 0.205 & 0.242 \\
    \bottomrule
    \end{tabular}%
    \fnote{\tiny \textit{Notes}: This table reports panel regressions of realized industry variance on a variety of characteristics including the average log variance of supply and demand shocks, log concentration across trade partners, log concentration between trade partners, log total employment (size), log centrality of the upstream and downstream propagation networks, durability of output, vertical position in the supply chain, and industry cluster and year fixed effects. I calculate network components as the average value of 1000 bootstrap samples that randomly drop 10\% of estimated pairwise non-zero correlations. In Panel A, the dependent variable is the log variance of annualized monthly returns on an equal-weighted industry portfolio. In Panel B, the dependent variable is the log variance of total quarterly year-on-year industry sales growth. I obtain return data from CRSP and sales data from Compustat. Concentration between and across trade partners are calculated as in \ref{eq:across empirical} and \ref{eq:between empirical}, where I calculate the variance-covariance matrix of supply and demand shocks using the product similarity distances from \citet{Hoberg2016textindustries}. Following \citet{Ahern2013NetworkCA}, I compute industry centrality as the eigenvector centrality of upstream and downstream propagation adjacency matrices. I obtain durability classifications from \citet{Gomes2009durability} and calculate vertical position of each industry as in \citet{Antrs2012} and \citet{GofmanSegalWu2020}. $^{***}$, $^{**}$ and $^{*}$ indicate significance at the 1\%, 5\%, and 10\% levels, respectively. Standard error are clustered at the BEA 15 major industry group level.  Sample is at an annual frequency from 1988 to 2018 for 479 BEA manufacturing industries.  }
  \label{tab:industry variance panel regressions (product similarity)}%
\end{table}%


\newpage

\section{Simulation Evidence}
\label{asec:simulation evidence}
Proposition \ref{prop:necessary restrictions on W} claims that there is no matrix $\mathbf{W}$ with entries $ w_{ij} \in (0,1)$ such that Assumption \ref{ass:idiosyncratic static shocks} and \eqref{eq:reduced form static propagation} are satisfied. This implies that $\Delta := ||(\mathbf{I}-\mathbf{W})\mathbf{\Sigma}_y(\mathbf{I}-\mathbf{W})- \mathbf{\Sigma}_u||$ should always be different from zero when $\mathbf{\Sigma}_u$ is diagonal. I test $H_0:\Delta = 0$ numerically as follows:

\begin{enumerate}
     \item Fix dimension $n$, number of iterations $S$, and consider the following constrained optimization problem: 
    \begin{align}
       f_{min}(\mathbf{\Sigma}_u) & = \min_{w_{ij,i\neq j } \in (0,1)} f(\mathbf{W}; \mathbf{\Sigma}_u) = \min_{w_{ij,i\neq j} \in (0,1)} \bigg|\bigg|\mathbf{\Sigma}_u - (\mathbf{I}_n - \mathbf{W})^{-1} \mathbf{\Sigma}_u(\mathbf{I}_n - \mathbf{W}^\top)^{-1}\bigg|\bigg|_F, 
       \label{eq:sim constrained min}
    \end{align}

    \item Draw two random samples $S_u^{d}$ and $S_u^*$ for the variance-covariance matrix of residuals $\mathbf{\Sigma}_u$ such that $S_u^d$ is diagonal and $S_u^*$ is not. I sample the non-diagonal matrix as follows:
    \begin{align*}
        S_u^*(s) & \sim_{iid} \mathcal{W}^{-1}(\mathbf{\Psi}_u,\nu),
    \end{align*} where  $\mathcal{W}^{-1}(\mathbf{\Psi},\nu)$ denotes an  Inverse-Wishart random variable with scale $\mathbf{\Psi}$ and degrees of freedom $\nu>n-1$. On the other hand, the diagonal matrix is given by:
    \begin{align*}
         S_u^d(s)  = \text{diag}(X_1,...,X_n), \quad X_i \sim_{iid} \Gamma^{-1}(\alpha_i,\beta_i),
    \end{align*}
    where $\Gamma^{-1}(\alpha_i,\beta_i)$ denotes an Inverse-Gamma distribution such that $\alpha_i > 1$ and $\mathbb{E}[X_i] = \frac{\beta_i}{\alpha_i-1} = [\Psi_{u}]_ii$ for all $i$. This ensures that the means of diagonal elements are the same across the two samples.
    
     \item Solve equation \eqref{eq:sim constrained min}  using off-the-shelf constrained quasi-Newton algorithms in both cases and construct $S$ realizations of $
        f_s^* = f_{min}(S_u^*)$ and $ f_s^{d} = f_{min}(S_u^d)$ for $s=1,...,S$.

    \item Construct the numerical p-value using $$ p = \frac{1}{S} \sum_{s=1}^S \mathbbm{1}\{f_{s}^* \geq f_s^{d}\},$$ and repeat Step 3 for different values of $n$ and $\mathbf{\Phi}_u$. Reported results in Table \ref{tab:simulation results pvals}.
    
    \item For optimal $W_{min}^{spec} = \argmin f(W,S_u^{spec})$, consider the $n\times n$ elementwise difference $\tilde{\Delta}_s$ for each iteration $s$:
    \begin{align*}
        \tilde{\Delta}_s = f(W_{min}^d,S_u^d) - f({W}_{min}^*,S_u^*)
    \end{align*}
    Then I test the marginal hypotheses $H_0: [\tilde{\Delta}_s]_{ij} = 0$ for all $i,j$. This test corresponds to the numerical t-statistic 
    $t_{ij} = m_{ij} / SE_{ij}$,    where $m_{ij}$ and $SE_{ij}$ are the mean and standard error of $[\tilde{\Delta}_s]_{ij} $, respectively. Reported results in Table \ref{tab:simulation results tstats}.

\end{enumerate}

\begin{table}[htbp]
  \centering
  \caption{\textbf{Simulation p-values}}
    \begin{tabular}{ccccc}
    \toprule
    \multicolumn{1}{l}{Specification} & n=2   & n=3   & n=4   & n=5 \\
    \midrule
    (1)   & 0.000 & 0.002 & 0.036 & 0.004 \\
    (2)   & 0.000 & 0.000 & 0.041 & 0.005 \\
    (3)   & 0.000 & 0.003 & 0.046 & 0.006 \\
    (4)   & 0.000 & 0.001 & 0.048 & 0.009 \\
    (5)   & 0.000 & 0.000 & 0.053 & 0.008 \\
    (6)   & 0.000 & 0.003 & 0.060 & 0.008 \\
    (7)   & 0.000 & 0.001 & 0.061 & 0.006 \\
    (8)   & 0.000 & 0.001 & 0.061 & 0.010 \\
    (9)   & 0.000 & 0.004 & 0.068 & 0.011 \\
    (10)  & 0.000 & 0.001 & 0.072 & 0.011 \\
    \bottomrule
    \end{tabular}%
    \fnote{\scriptsize \textit{Note:} this table reports numerical p-values from the procedure described in Appendix \ref{asec:simulation evidence} for different values of $n$ and $\mathbf{\Phi}$. Number of iterations is $S = 1000$.}
  \label{tab:simulation results pvals}%
\end{table}%

\begin{table}[htbp]
  \centering
  \caption{\textbf{Simulation t-statistics}}
    \begin{tabular}{cccc}
    \toprule
    n=2   & n=3   & n=4   & n=5 \\
    \midrule
    121   & -27   & -54   & -118 \\
    863   & 48    & 133   & 201 \\
    861   & 60    & 163   & 152 \\
    -504  & 48    & 2     & -502 \\
          & -68   & 133   & -152 \\
          & 80    & -247  & 201 \\
          & 60    & -26   & -583 \\
          & 80    & 190   & 172 \\
          & -91   & 163   & 205 \\
          &       & -26   & 61 \\
          &       & -255  & 152 \\
          &       & 204   & 172 \\
          &       & 2     & -177 \\
          &       & 190   & 174 \\
          &       & 204   & 138 \\
          &       & -60   & -502 \\
          &       &       & 205 \\
          &       &       & 174 \\
          &       &       & -89 \\
          &       &       & -334 \\
          &       &       & -152 \\
          &       &       & 61 \\
          &       &       & 138 \\
          &       &       & -334 \\
          &       &       & -75 \\
    \bottomrule
    \end{tabular}%
    \fnote{\scriptsize \textit{Note:} this table reports elementwise t-statistcs from the procedure described in Appendix \ref{asec:simulation evidence} for different values of $n$ and a fixed value $\mathbf{\Phi}$. Number of iterations is $S = 1000$.}
  \label{tab:simulation results tstats}%
\end{table}%



\section{Comovement in Industry Volatility}
Recent research documents significant common variation in both market and fundamental volatility at the granular level.\footnote{\citet{HKLNCIV2016} show that a single common factor explains around 30\% of variation in log variance for the panel of CRSP stocks.  Other work also documents common variation in option-implied volatilities (Engel and Figlewski (2015)), intra-daily returns (Barigozzi et al. (2014)), and dispersion in firm sales growth (Bloom et al. (2018)).} Importantly, the factor structure in volatility is significant even after removing all common variation in returns and cash flow growth, which suggests that this is being driven by underlying sources of systematic risk and rather than an omitted set of returns or sales growth factors.\footnote{\citet{HKLNCIV2016} verify that there is a factor structure even when the pairwise correlation between idiosyncratic return or sales growth residuals is statistically indistinguishable from zero. I also verify that the factor structure holds in residual returns after a non-parametric regression using deep feed-forward neural networks, which have favorable universal approximation properties (see e.g., \citet{WhiteUniv89}). } When economic units are connected via input-output networks, a shock to any given unit can generate systematic effects. Moreover, networks mechanically generate volatility comovement regardless of whether shocks to individual units are uncorrelated. When shocks are correlated, this comovement is even more pronounced.

My results so far establish a significant relationship between network concentration across and between customers and suppliers and realized variance. Consequently, comovement in supply chain concentration should also generate comovement in realized variance. Table \ref{tab:factor regressions industry variance} reports average loadings and $R^2$ values of univariate factor regressions for panel of industry network concentration and realized volatility measures. The main takeaway is that 20-30\% of dynamic variation in input-output concentration and return and sales growth volatilities across industries of these variables can be explained by a single factor. Moreover, there is a significant degree of comovement between these common factors. That is, network concentration factors can explain up to 40\% of time-series variation in both market and sales growth volatility factors. These results are robust to a variety of specifications.

\begin{table}[H]
  \centering
  \caption{\textbf{Comovement in Industry Variance}}
    \begin{tabular}{lcccccc}
    \toprule
    \multicolumn{7}{c}{Panel A: Loadings} \\
    \midrule
    \multicolumn{1}{c}{Factor / Outcome} & Across (d) & Between (d) & Across (s) & Between (s) & Var (mkt) & Var (cf) \\
    \midrule
    Across (d) & 0.788 & 1.576 & 0.364 & 0.507 & 0.410 & 0.574 \\
    Between (d) & 0.584 & 0.735 & 0.394 & 0.567 & 0.206 & 0.383 \\
    Across (s) & 0.701 & 1.185 & 0.700 & 0.518 & 0.156 & 0.497 \\
    Between (s) & 0.455 & 0.422 & 0.401 & 0.736 & 0.157 & 0.310 \\
    Var (mkt) & 0.091 & 0.124 & 0.016 & 0.035 & 0.880 & 0.185 \\
    Var (cf) & 0.109 & 0.181 & 0.156 & 0.218 & 0.283 & 0.721 \\
    \midrule
    \multicolumn{7}{c}{Panel B: $R^2$ (avg univariate)} \\
    \midrule
    \multicolumn{1}{c}{Factor / Outcome} & Across (d) & Between (d) & Across (s) & Between (s) & Var (mkt) & Var (cf) \\
    \midrule
    Across (d) & 0.201 & 0.258 & 0.217 & 0.318 & 0.170 & 0.079 \\
    Between (d) & 0.189 & 0.318 & 0.257 & 0.400 & 0.174 & 0.087 \\
    Across (s) & 0.164 & 0.274 & 0.229 & 0.344 & 0.168 & 0.082 \\
    Between (s) & 0.189 & 0.322 & 0.265 & 0.347 & 0.165 & 0.086 \\
    Var (mkt) & 0.042 & 0.046 & 0.041 & 0.042 & 0.343 & 0.068 \\
    Var (cf) & 0.063 & 0.097 & 0.085 & 0.115 & 0.147 & 0.104 \\
    \midrule
    \multicolumn{7}{c}{Panel C: $R^2$ (aggregate)} \\
    \midrule
    \multicolumn{1}{c}{Factor / Outcome} & Across (d) & Between (d) & Across (s) & Between (s) & Var (mkt) & Var (cf) \\
    \midrule
    Across (d) & 1     & 0.620 & 0.575 & 0.566 & 0.373 & 0.352 \\
    Between (d) &       & 1     & 0.628 & 0.670 & 0.255 & 0.422 \\
    Across (s) &       &       & 1     & 0.668 & 0.196 & 0.469 \\
    Between (s) &       &       &       & 1     & 0.229 & 0.428 \\
    Var (mkt) &       &       &       &       & 1     & 0.165 \\
    Var (cf) &       &       &       &       &       & 1 \\
    \bottomrule
    \end{tabular}%
  \label{tab:factor regressions industry variance}%
  \fnote{\scriptsize \textit{Notes}: This table reports the results of factor regressions for industry variance components. I calculate factors as the first principal component of an industry panel of the variable of interest. Panels A and B report the average loading and $R^2$ from the regressions $outcome_{it} = \alpha_i + \beta_i \cdot f_t + u_{it}$, respectively. Each column, row pair denotes a different outcome, factor pair. Panel C reports the $R^2$ of aggregate time-series regressions of factor pairs (i.e., $y_t = \alpha + \beta \cdot x_t + u_t$, where $y_t$ is the column factor and $x_t$ is the row factor). All variables are log-transformed. Sample is at an annual frequency from 1997 to 2019 for 66 non-government BEA industries.  }
\end{table}%

\newpage
\section{Additional Firm-Level Results}
\label{asec:additional firm results}

\begin{table}[H]
  \centering
  \caption{\textbf{Predicting Firm Sales Growth }}
    \begin{tabular}{lccc}
    \toprule
    \multicolumn{1}{c}{Variable} & (1)   & (2)   & (3) \\
    \midrule
    $a_{t-1}$  & 0.031** & 0.029** & 0.007** \\
          & (0.002) & (0.001) & (0.000) \\
    $g_{t-1}$ & -0.016** & -0.018** & 0.015 \\
          & (0.003) & (0.003) & (0.011) \\
    $ROA_{i,t-1}$ & 0.080** & 0.063** & 0.064** \\
          & (0.003) & (0.002) & (0.002) \\
    $size_{i,t-1}$ & -0.008** & -0.008** & -0.008** \\
          & (0.000) & (0.000) & (0.000) \\
    $age_i$  & 0.037** & 0.050** & 0.046** \\
          & (0.004) & (0.004) & (0.003) \\
    Constant & 0.44  & 0.372 & 0.39 \\
    Obs   & 259,976 & 259,976 & 259,976 \\
    Adj R2 & 0.265 & 0.259 & 0.258 \\
    \bottomrule
    \end{tabular}%
    \fnote{\scriptsize \textit{Notes}: This table reports the regression results based on the model in \eqref{eq: empirical sales growth panel regression}. The dependent variable is year-on-year quarterly sales growth and the covariates are aggregate TFP growth from \citet{Fernald2012TFP}, procurement proxy from \citet{BrigantiSellemiProcurement}, log return on assets, size (log market value), and age as year appears on the database. Each column reports results for different industry aggregations: (1) is 66 BEA non-government industries, (2) is 405 BEA non-government industries, and (3) is 15 major BEA non-government industries. Standard errors are clustered at the same granularity. $**$ and $*$ indicates significance at the 1\% and 5\% levels, resp.  Sample is quarterly between 1997-2019 for 10,700 firms.     }
  \label{tab:firm sales growth regression}%
\end{table}%

 \begin{figure}[H]
    \centering
     \subfloat[\centering $k_i$]{{\includegraphics[width=.3\textwidth]{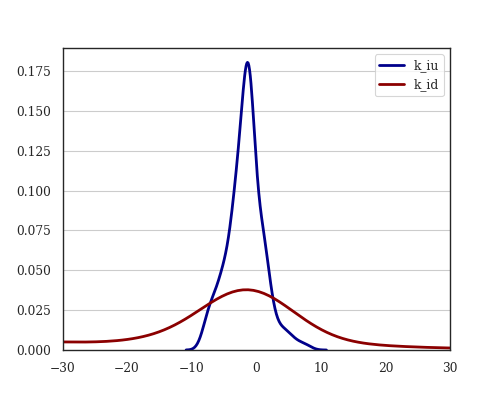} }}%
    \subfloat[\centering $x_i$]{{\includegraphics[width=.3\textwidth]{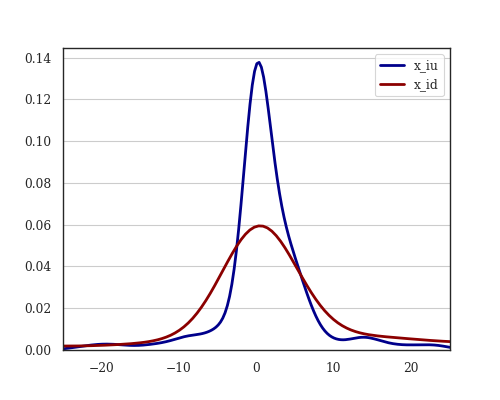} }}%
        \subfloat[\centering $\bar{p}_i$]{{\includegraphics[width=.3\textwidth]{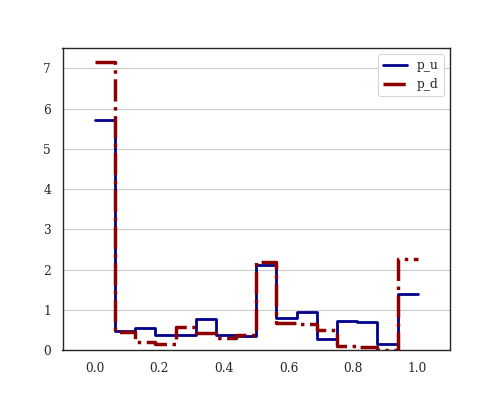} }}%
      \caption{\textbf{Distribution of Calibrated Propensities with No Network Connections}}
    \fnote{\scriptsize \textit{Notes:} This figure plots the kernel density of calibrated parameters $k_{iq}$ and $x_{iq}$ as described in Section \ref{sec:testable implications}. Note that $\bar{p}_{iq} = 1/(1+\exp(-k_{iq}x_{iq}))$. }
    \label{fig:calibrated model parameters}
\end{figure}

\newpage
\begin{figure}[H]
    \centering
    \subfloat[Panel A: Propagation Factors]{\includegraphics[width = \textwidth]{FIGURES2_dissertation/propFactorsTS.png}}\\
    \subfloat[Panel B: Innovations in Average Substitutability]{\includegraphics[width = \textwidth]{FIGURES2_dissertation/subFactorsTS.png}}\\
    \subfloat[Panel C: Productivity and Demand Growth]{\includegraphics[width = \textwidth]{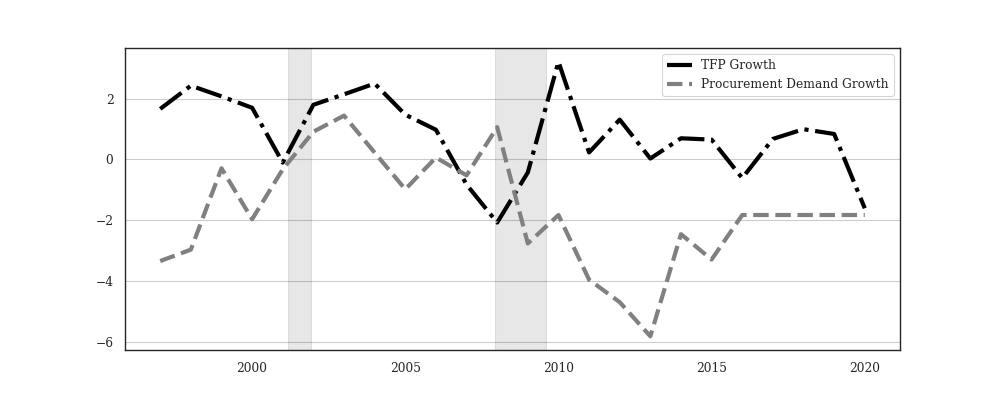}}
     \caption{\textbf{Network Propagation Risk Factors}}
    \label{fig:propagation factors full}
    \fnote{\scriptsize\textit{Notes:} This figure plots the time series of network propagation risk factors (Panel A), the cross-sectional average industry substitutability (Panel B), and aggregate TFP and procurement demand growth (Panel C). Shaded regions indicate NBER-dated recession periods.  }
\end{figure}

\begin{table}[H]
    \centering
    \caption{Descriptive Statistics Calibrated Firm Parameters}
    \label{tab:calibrated parameters}
\end{table}

\newpage
\begin{table}[H]
  \centering
  \caption{\textbf{One-Way Sorted Portfolios on Network Propagation Factors (controlling for volatility factors)}}
    \begin{tabular}{lcccccccc}
    \toprule
    \multicolumn{9}{c}{Panel A: One-way sorts on upstream propagation beta  (controlling for $a_t$, $g_t$, and $\sigma_t^{civ}$)} \\
    \midrule
          & 1 (Low) & 2     & 3     & 4     & 5 (High) & H-L   & t(H-L) & MR p-val \\
    \midrule
    $\mathbb{E}[r] - r_f$ & 12.97 & 12.12 & 11.34 & 10.44 & 10.11 & -2.86 & -5.83 & 0.22 \\
    $\alpha_{capm}$ & -0.04 & -0.21 & -0.21 & -0.28 & -0.68 & -0.64 & -4.06 & 0.00\\
    $\alpha_{ff3}$& -0.10  & -0.14 & -0.17 & -0.25 & -0.39 & -0.29 & -0.83 & 0.01 \\
    Volatility (\%) & 13.37 & 14.81 & 14.71 & 19.1  & 12.64 & -     & -     & - \\
    Book-to-market & 0.52  & 0.53  & 0.58  & 0.54  & 0.50   & -     & -     & - \\
    Market value (\$bn) & 12.18 & 6.21  & 18.07 & 6.14  & 15.75 & -     & -     & - \\
    \midrule
    \multicolumn{9}{c}{Panel B: One-way sorts on downstream propagation beta  (controlling for $a_t$, $g_t$, and $\sigma_t^{civ}$)} \\
    \midrule
          & 1 (Low) & 2     & 3     & 4     & 5 (High) & H-L   & t(H-L) & MR p-val \\
    \midrule
    $\mathbb{E}[r] - r_f$ & 14.97 & 12.94 & 12.15 & 9.78  & 9.35  & -5.61 & -9.36 & 0.09 \\
    $\alpha_{capm}$ & 0.03  & -0.24 & -0.29 & -0.32 & -0.39 & -0.43 & -1.88 & 0.00 \\
    $\alpha_{ff3}$ & -0.02 & -0.10  & -0.26 & -0.29 & -0.30  & -0.29 & -3.40  & 0.01 \\
    Volatility (\%) & 18.69 & 16.42 & 12.11 & 14.20  & 13.02 & -     & -     & - \\
    Book-to-market & 0.59  & 0.55  & 0.51  & 0.51  & 0.52  & -     & -     & - \\
   Market value (\$bn) & 5.11  & 12.13 & 8.96  & 16.43 & 15.64 & -     & -     & - \\
    \midrule
    \multicolumn{9}{c}{Panel C: One-way sorts on upstream propagation beta  (controlling for $a_t$, $g_t$, and $\sigma_t^{mkt}$)} \\
    \midrule
          & 1 (Low) & 2     & 3     & 4     & 5 (High) & H-L   & t(H-L) & MR p-val \\
    \midrule
    $\mathbb{E}[r] - r_f$ & 13.83 & 13.07 & 11.33 & 10.10  & 9.73  & -4.09 & -8.10 & 0.21 \\
     $\alpha_{capm}$ & 0.01  & -0.26 & -0.28 & -0.35 & -0.45 & -0.47 & -6.84 & 0.04 \\
     $\alpha_{ff3}$ & -0.07 & -0.17 & -0.20  & -0.29 & -0.30  & -0.23 & -2.06 & 0.03 \\
    Volatility (\%) & 13.68 & 16.44 & 17.09 & 14.12 & 13.03 & -     & -     & - \\
    Book-to-market & 0.51  & 0.54  & 0.55  & 0.54  & 0.53  & -     & -     & - \\
    Market value (\$bn) & 11.43 & 5.87  & 7.24  & 19.44 & 14.40  & -     & -     & - \\
    \midrule
    \multicolumn{9}{c}{Panel D: One-way sorts on downstream propagation beta (controlling for $a_t$, $g_t$, and $\sigma_t^{mkt}$)} \\
    \midrule
          & 1 (Low) & 2     & 3     & 4     & 5 (High) & H-L   & t(H-L) & MR p-val \\
    \midrule
    $\mathbb{E}[r] - r_f$ & 13.66 & 13.5  & 12.67 & 9.77  & 9.24  & -4.42 & -7.17 & 0.37 \\
     $\alpha_{capm}$ & 0.07  & -0.26 & -0.32 & -0.39 & -0.41 & -0.48 & -1.81 & 0.00 \\
     $\alpha_{capm}$ & 0.02  & -0.17 & -0.24 & -0.25 & -0.33 & -0.35 & -1.52 & 0.01 \\
    Volatility (\%) & 19.47 & 12.94 & 15.7  & 14.27 & 13.03 & -     & -     & - \\
    Book-to-market & 0.63  & 0.54  & 0.44  & 0.53  & 0.52  & -     & -     & - \\
    Market value (\$bn) & 4.60   & 10.24 & 13.33 & 15.77 & 14.30  & -     & -     & - \\
    \bottomrule
    \end{tabular}%
    \fnote{\scriptsize \textit{Notes:} This table reports average excess returns and post-sample alphas in annual percentages for value-weighted portfolios sorted into quintiles on annual upstream and downstream propagation factors. Sample is between 1997-2021 for more than 10,000 stocks belonging to the BEA 66 non-government industry classifications. Panels A and B control for productivity growth and federal procurement demand growth, while Panels C and D have no controls. I also report average return volatility, book-to-market ratio and market value for each portfolio. To test for significant return spreads,  I report $t$-statistics for the null hypothesis $H_0: xr_{5} = xr_{1}$, where $xr_q$ is the average return of the $q^{th}$ quintile single sorted portfolio. Moreover, I report p-values for the test $H_0: xr_{q+1} < xr_q  \forall q \leq 4$, calculated via bootstrap following \citet{timmermannMRtest}.  }
 \label{tab:one way sort controlling for vol}%
\end{table}%

\end{appendices}
\end{document}